\documentclass[11pt]{article}
\pdfoutput=1
\usepackage{amssymb, amsmath, verbatim, amsthm,url, multirow, longtable, fullpage,mathtools, appendix, tikz, xcolor, mathrsfs}
\usepackage{longtable, rotating,makecell,array}
\usepackage[aligntableaux=top]{ytableau}
\usepackage{tikz-cd}
\usepackage{listings}
\usepackage{color}
\usepackage{soul}
\usepackage[colorinlistoftodos,textsize=footnotesize]{todonotes}

\setstcolor{red}

\usetikzlibrary{calc} 
\usetikzlibrary{decorations.pathmorphing} 
\usetikzlibrary{decorations.pathreplacing} 
\usetikzlibrary{bending} 
\usetikzlibrary{patterns} 

\newcommand{\labeledvertices}[3]{
	\foreach[count = \i]  \x in {#1} {
		\draw (-1, -.5*\i) node [label = {[shift={(-.5, -.4)}]$\x$}] {$\bullet$} ;}
	\foreach[count = \i]  \y in {#2} {
		\draw (1, -.5*\y) node [label = {[shift={(.5, -.4)}]$\y$}] {$\bullet$} ;}
	\foreach \xend / \yend in {#3} {
		\draw (-1, -.5*\xend) -- (1, -.5*\yend);}
  
}

\newcommand{\bipartitelabels}[2]{
	\draw (-1, 0) node {$#1$};
	\draw (1, 0) node {$#2$};
}
\tikzstyle{smallpropagator}=[decorate,decoration={snake,segment length=3mm,amplitude=0.5mm}]

\def\centerarc[#1](#2)(#3:#4:#5)
    { \draw[#1] ($(#2)+({#5*cos(#3)},{#5*sin(#3)})$) arc (#3:#4:#5); }

\def\clipcenterarc(#1)(#2:#3:#4)
    { \clip ($(#1)+({#4*cos(#2)},{#4*sin(#2)})$) arc (#2:#3:#4); }

\newcommand{\newbetternode}[3][{label distance=-1mm]left}]{
  \node[label={#1:{\scriptsize $#3$}}] at (\zero + #2*\step:\radius) {\scriptsize $\bullet$};
}

\newcommand{\Z}{\mathbb{Z}}

\newcommand{\R}{\mathbb{R}}

\newcommand{\Gr}{\mathbb{G}_{\R, \geq 0}}

\newcommand{\Grall}{\mathbb{G}_{\R}}

\DeclareMathOperator{\rk}{rk}

\newcommand{\cl}{\textrm{cl}}

\def\ba #1\ea{\begin{align} #1 \end{align}}
\def\bas #1\eas{\begin{align*} #1 \end{align*}}
\def\bml #1\eml{\begin{multline} #1 \end{multline}}
\def\bmls #1\emls{\begin{multline*} #1 \end{multline*}}

\newcommand{\sI}{\mathscr{I}}
\newcommand{\cP}{\mathcal{P}}
\newcommand{\cB}{\mathcal{B}}
\newcommand{\cF}{\mathcal{F}}

\newcommand{\cN}{\mathcal{N}}
\newcommand{\sfM}{\mathsf{M}}

\newcommand{\cM}{\mathcal{M}}

\newcommand{\cI}{\mathcal{I}}
\newcommand{\cC}{\mathcal{C}}

\newcommand{\Prop}{\textrm{Prop}}
\newcommand{\End}{\textrm{End}}
\newcommand{\cE}{\mathcal{E}}
\newcommand{\Int}{\textrm{Int}}

\newtheorem{thm}{Theorem}[section]

\newtheorem{lem}[thm]{Lemma}
\newtheorem{cor}[thm]{Corollary}
\newtheorem{prop}[thm]{Proposition}
\newtheorem{algorithm}[thm]{Algorithm}

\theoremstyle{remark}
\newtheorem{eg}[thm]{Example}

\theoremstyle{definition}
\newtheorem{dfn}[thm]{Definition}
\newtheorem{rmk}[thm]{Remark}

\definecolor{latte}{rgb}{.99,.97,.92}
\definecolor{codeblue}{rgb}{0.08, .3, 1}
\definecolor{codegray}{rgb}{0.5,0.5,0.5}
\definecolor{codered}{rgb}{.76, 0, .22}
\definecolor{codegreen}{rgb}{0,.76,.54}
\lstdefinestyle{mystyle}{
    backgroundcolor=\color{latte},   
    commentstyle=\color{codeblue},
    keywordstyle=\color{codered},
    numberstyle=\tiny\color{codegreen},
    stringstyle=\color{codegray},
    basicstyle=\ttfamily,
    breakatwhitespace=false,         
    breaklines=true,                 
    captionpos=b,                    
    keepspaces=true,                 
    numbers=left,                    
    numbersep=5pt,                  
    showspaces=false,                
    showstringspaces=false,
    showtabs=false,                  
    tabsize=2
}
\title{Rado matroids and a graphical calculus for boundaries of Wilson loop diagrams}
\author{Susama Agarwala, Colleen Delaney, Karen Yeats}
\date{}

\begin{document}
\maketitle

\begin{abstract}
    We study the boundaries of the positroid cells which arise from $\cN=4$ super Yang Mills theory. Our main tool is a new diagrammatic object which generalizes the Wilson loop diagrams used to represent interactions in the theory. We prove conditions under which these new generalized Wilson loop diagrams correspond to positroids and give an explicit algorithm to calculate the Grassmann necklace of said positroids. Then we develop a graphical calculus operating directly on noncrossing generalized Wilson loop diagrams. In this paradigm, applying diagrammatic moves to a generalized Wilson loop diagram results in new diagrams that represent boundaries of its associated positroid, without passing through cryptomorphisms. We provide a Python implementation of the graphical calculus and use it to show that the boundaries of positroids associated to ordinary Wilson loop diagram are generated by our diagrammatic moves in certain cases.
\end{abstract}
\section{Introduction}

Wilson loop diagrams are a variant of the notion of chord diagram that was introduced in high energy physics in order to compute scattering amplitudes in $\cN=4$ super Yang Mills theory (SYM). Each admissible Wilson loop diagram is associated to a positroid, which has led to their study from a combinatorial \cite{generalcombinatoricsI, generalcombinatoricsII}, geometric \cite{non-orientable, cancellation, casestudy}, and matroidal \cite{Wilsonloop, casestudy} point of view, in addition to their investigation from a physics perspective \cite{Wilsonloop, cancellation, Amplituhedronsquared, HeslopStewart}.

In particular, it is natural to ask about the boundaries of positroid cells associated to Wilson loop diagrams. Interest in this question is further motivated by the physics context as understanding the boundaries is necessary for understanding the singularity structure of the integral that physics tells us is associated to each Wilson loop diagram and hence for understanding cancellations of singularities between such integrals.

 Some of these boundaries were already known to be able to be understood at the diagrammatic level by taking a limit where the end of a propagator of the Wilson loop diagram is sent to a neighboring vertex or identified with the end of a neighboring propagator \cite{casestudy, HeslopStewart, Amplituhedronsquared}. The resulting diagrams no longer met the definition of a Wilson loop diagram, but were close enough that the mapping from diagram to positroid could be defined with essentially no change from the Wilson loop diagram case. 

 Upon enumeration of the Wilson loop diagrams and their codimension one boundaries, one finds that the diagrams correspond to some, but not all, positroid cells in $\Gr(k,n)$, and the boundaries formed by moving propagators correspond to some, but not all, of the codimension one boundaries of said positroid cells (and some actually correspond to higher codimension boundaries)\cite{casestudy, cancellation}. This raises the question of whether there is a diagrammatic way to obtain all boundaries of any codimension of positroid cells associated to Wilson loop diagrams. In order to attack this question, we first define a generalization of Wilson loop diagrams where propagators are now weighted hyperedges, i.e.\ can have an arbitrary number of ends, with a new weighting quantity called the capacity. Furthermore, these generalized propagators can have ends on vertices or on edges of the outer loop. We call these diagrams {\it generalized Wilson loop diagrams}, with their definition laid out in Section \ref{sec:gWLDdiagrammatics}. These prove to be interesting objects in and of themselves, both combinatorially and matroidally.

 With the definition in place, we pursue two related directions. First, we study the matroidal properties of generalized Wilson loop diagrams (Section \ref{sec:allWLDtomatroid}), and second, we develop a diagrammatic calculus for constructing boundaries of a generalized Wilson loop diagram (Section \ref{sec:diagrammaticmoves}), though it is not yet strong enough to obtain all boundaries in all circumstances.

For the first direction, we observe that each ordinary Wilson loop diagram can be associated to a transversal matroid representing the connection between the vertices of the diagram and the propagators \cite{basisshapeloci}. This naturally generalizes to associate a Rado matroid structure to each generalized Wilson loop diagram. We pursue this construction in Section \ref{sec:gWLDtomatroid}. Generalized Wilson loop diagrams have particularly clean Rado matroids and we can understand when they give positroids (Section~\ref{sec:positive gWLD}). Further, we give an algorithm to read off their Grassmann necklaces in general (Section~\ref{sec:GN of Rado}). 
Specifically for the case of non-crossing generalized Wilson loop diagrams, we prove they are associated to positroids (Corollary \ref{res:non-crosspos}) and that the algorithm to read off the Grassmann necklace has a particularly explicit formulation for the non-crossing case (Algorithm \ref{algo:non-crossgwld}) which generalizes the algorithm for ordinary Wilson loop diagrams.

The second direction, then, is pursued in Section \ref{sec:diagrammaticmoves} where we define families of diagrammatic moves on non-crossing generalized Wilson loop diagrams and provide a Python implementation. We prove that all our moves give boundaries of the diagram in Section \ref{sec:movesrealizeboundaries}. Then in Section \ref{sec:diagrammaticboundaries} we use our Python implementation to prove that all boundaries of all codimension can be obtained by applying sequences of these moves in the case that the initial diagram was an ordinary admissible Wilson loop diagram with at most two propagators. 
While there are many algorithms for computing boundaries of positroid cells using more algebraic tools, such as reduced subwords \cite{Juggling, Towers} a diagrammatic representation of boundaries is something that is currently missing from the literature on the geometry of positroids. Our moves give a partial diagrammatic calculus for positroid cell boundaries. We discuss how far we are from the full story for diagrammatic boundaries along with other concluding remarks in Section~\ref{sec:finale}. 

\subsection*{Acknowledgements}

The authors thank Zee Fryer for sharing their code to compute positroid boundaries via reduced subwords and Matthew Towers for sharing an algorithm to compute the boundaries on which Zee Fryer's code is based.
The authors also thank SLMath for their hospitality and the opportunity they provided to focus intensely on this project for two weeks.
    
This material is based on upon work supported by the National Science Foundation under Grant No.~DMS-1928930 and the National Security Agency under Grant No.~H98230-23-1-0004 while the authors participated in a program hosted by the Simons Laufer Mathematical Sciences Institute (formerly MSRI) in Berkeley, California during the summer of 2023.

KY is supported by an NSERC Discovery grant, the Canada Research Chairs program, and was supported by the Emmy Noether Fellows program at Perimeter Institute. 
Research at Perimeter Institute is supported in part by the Government of Canada through the
Department of Innovation, Science and Economic Development Canada and by the province of
Ontario through the Ministry of Economic Development, Job Creation and Trade. This research
was also supported in part by the Simons Foundation through the Simons Foundation Emmy
Noether Fellows Program at Perimeter Institute. 

CD is supported by Simons Foundation Grant $\#$888984 (Simons Collaboration on Global Categorical Symmetries). Part of this work was also supported by NSF Mathematical Sciences Postdoctoral Fellowship $\#$2002222.

\section{Generalized Wilson loop diagrams \label{sec:gWLDdiagrammatics}}

\subsection{General set up}

Throughout this paper, we use cyclically ordered sets. Therefore, we start with some notation on these sets and their orderings.
\begin{dfn}\label{dfn:cyclic orders}
Let $S = \{s_1, \ldots, s_n\}$ be a set with a cyclic ordering $s_1 < s_2 \ldots < s_n < s_1$.  We write $<_{s_i}$, for the linear order on the elements starting at $s_i$: \bas s_i <_{s_i} s_{i+1} <_{s_i} \ldots s_n <_{s_i} s_1 <_{s_i} \ldots s_{i-1} \;.\eas 
\end{dfn}

Next we define a generalized Wilson loop diagram and then recall the structure of the ordinary Wilson Loop diagrams as a special case. 

\begin{dfn}\label{dfn:gWLD} 
  Endow the set $C=\{v_1, e_1, v_2, e_2, \ldots, v_n, e_n\}$ of $2n$ elements, with a cyclic ordering by setting $v_i < e_i < v_{i+1}$ and $e_n < v_1$. A \emph{generalized Wilson loop diagram} is a set of multisets of elements of $C$ where each multiset has a positive integer weight called the \emph{capacity} which is at most the size of that multiset. The multisets are called \emph{propagators}.
  We call the elements of a propagator (as a multiset) the \emph{ends} of the propagator.

\end{dfn}

Alternatively, generalized Wilson loop diagrams can be viewed as hypergraphs where the propagators are hyperedges. Multiplicities in the multiset then become additional weights on the elements of the hyperedge, which are in addition to the capacity of the hyperedge as a whole, as defined above.

\begin{rmk}\label{rmk:matroids on props def}
    Through the work in this paper we will begin to see value in yet a further generalization of this notion of generalized Wilson loop diagram where instead of each propagator having a capacity, each propagator would have an associated matroid whose underlying set is the set of ends and with the rank of that matroid playing the role of the capacity. The generalized Wilson loop diagrams of Definition~\ref{dfn:gWLD} are then those where the associated matroids are all uniform.

    For the purposes of this paper, this further generalization will be restricted to some remarks. We expect to investigate these objects in more depth in future work.
\end{rmk}

The notation of the set $C$ as an alternating sequence of $v_i$ and $e_i$ is intended to be suggestive of the manner in which we will draw these objects. Specifically, we will draw a generalized Wilson loop diagram by an outer circle with $n$ marked vertices $\{v_1, \ldots, v_n\}$ ordered counterclockwise along the circle. The $e_i$ are then the edges of the outer circle with $e_i$ joining vertex $v_i$ with vertex $v_{i+1}$ (with indices interpreted modulo $n$, as will be done without comment in what follows). The propagators are then drawn as stars connecting an internal point with each of the ends, either a $v_i$ or $e_i$, by wiggly lines. If the capacity of a propagator is larger than $1$ then the capacity will be written near the internal point. When multiple ends (whether of the same propagator or different propagators) are incident with an edge $e_i$ they will be drawn meeting $e_i$ at separate points along the edge. We will prefer to draw generalized Wilson loop diagrams in ways which minimize propagator crossings. The generalized Wilson loop diagram does not encode the order at which the ends meet elements of $C$, but in the non-crossing case, there is an order induced from the non-crossing structure which will be respected in drawings. See Definition \ref{dfn:global cyclic end order} for further details on the non-crossing case.

Definition~\ref{dfn:gWLD} does not treat the $v_i$ differently from the $e_i$, but in view of our desired connection with positroids and our motivating context, the $v_i$ and $e_i$ will be treated differently in various constructions that we will use on generalized Wilson loop diagrams. To this end, we define the \emph{support} of an end incident to vertex $v_i$ to be $\{v_i\}$ while the support of an end incident to an edge $e_i$ is $\{v_i, v_{i+1}\}$. The support of a propagator, then, is the union of the support of its ends. 

With this in mind, we speak of a generalized Wilson loop diagram defined on $[n]=\{1,2,\ldots, n\}$, meaning $C=\{v_1, e_1, \ldots, v_n, e_n\}$, and where convenient we abuse notation by conflating the vertices $\{v_1, \ldots, v_n\}$ with their indices $[n]$ and by identifying $e_i$ with the set $\{v_i, v_{i+1}\}$ for each $i$.  We use the standard cyclic order on $[n]$, corresponding to the cyclic order on $C$.
We will use the following notation. A generalized Wilson loop diagram is denoted $W=(\cP, [n])$ where $\cP$ is the set of propagators. The capacity of a propagator $p\in \cP$ is denoted $c(p)$, its ends denoted $\cE(p)$, and its \emph{support} indicated by $V(p)$. The support of an end $e$ is given by $V(e)$.

\begin{eg}\label{eg:working example}
The following generalized Wilson loop diagram has 4 propagators, one with capacity 2 and the rest with capacity 1:
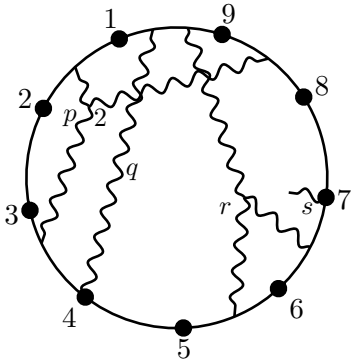
\[\begin{tikzpicture}[rotate=-67.5, line width=1, scale=2]
	\def \n {9}
	\draw circle(1)
	\foreach \v in{1,...,\n}
	{(360*\v/\n-360/\n+180:1)circle(.4pt)circle(.8pt)circle(1.2pt)circle(1.4pt) node[anchor=360/\n*\v-360/\n-67.5]{$\v$}};
	\foreach \x/\y in {-0.940/-0.342,0.058/-0.998,-0.500/0.866}{\draw[decorate,decoration={snake,amplitude=0.8mm}] (-0.661,-0.358) -- (\x,\y);}
	\draw (-0.661,-0.358) node[shift = {(.15,-.15)}] {\small2};
	\draw (-0.661,-0.358) node[shift = {(-.25,-.15)}] {\small $p$};
	\draw[decorate,decoration={snake,amplitude=0.8mm}] (0.500,-0.866) -- (-0.973,0.231);
	\draw (-.25 , -.3) node[shift = {(.15,-.15)}] {\small $q$};
	\foreach \x/\y in {1.00000000000000/0,0.766/0.643,-0.894/0.449}{\draw[ decorate,decoration={snake,amplitude=0.8mm}] (0.291,0.364) -- (\x,\y);}
	\draw (0.291,0.364) node[shift = {(-.25,-.15)}] {\small $r$};
	\draw[decorate,decoration={snake,amplitude=0.8mm}] (0.500,0.866) -- ++({atan(1.73)}:-.25 cm);
	\draw (0.500,0.866) node[shift = {(-.25,-.15)}] {\small $s$};
\end{tikzpicture}\]

In order to refer to the propagators more easily, we name them. Let $p$ be the capacity 2 propagator with 3 ends, $\cE(p) = \{e_1, e_3, e_8\}$. Let $q$ be the three ended propagator with capacity 1, with $\cE(q) = \{e_5, e_6, e_9\}$. Let $r$ be the two ended propagator with capacity 1, and $\cE(r) = \{v_4, e_9\}$. Finally, let $s$ be the single end propagator with $\cE(s) = \{v_7\}$. 

Note that there is no way to depict this set of propagators without the lines representing $p$ crossing the lines representing the propagators $q$ and $r$ at least once each. Because we prefer drawing generalized Wilson loop diagrams so as to have the minimal number of crossings, we do not place the internal point of $p$ to the right of the lines for $q$ or $r$, as that would cause additional crossings. The analogous comment holds for the internal point of $q$. We explore crossings further in Section~\ref{sec:non-crossing}.
\end{eg}

Next, we set up some notation regarding propagators, ends, their supporting vertices, and their capacities.
We will take the convention of denoting sets of propagators, vertices, and ends by capital letters, and use lower case letters for individual elements.

\begin{dfn}\label{dfn:prop stuff}
Let $(\cP, [n])$ be a generalized Wilson loop diagram.
    \begin{itemize}
        \item For any subset of propagators, $P \subseteq \cP$, define the \emph{support} of $P$ to be $V(P) = \bigcup_{p \in P} V(p)$. In particular $V(\cP)$ is the set of all vertices that support at least one propagator.
        \item Given a set of vertices $V \subseteq [n]$, we define $\Prop(V)$ to be the set of propagators supported on $V$, \bas \Prop(V) = \{p \in \cP : V \cap V(p) \neq \emptyset\}\eas
    \end{itemize}
\end{dfn}

Note that each end is implicitly attached to a propagator. Formally, then, we can view an end as a pair $(a, p)$ where $p\in \cP$ and $a$ is an element of the multiset $p$, so $a$ itself is a pair of some $v_i$ or $e_i$ along with an index to distinguish the copies of the same element in the multiset. However, this formalism is unnecessarily heavy for our purposes; it will suffice to refer to ends by their support with the additional information understood from context, or to name the ends as needed.

\begin{dfn}\label{dfn:end stuff}
    Let $(\cP, [n])$ be a generalized Wilson loop diagram.
    \begin{itemize}
        \item For any subset of propagators, $P\subseteq \cP$, we define the set of ends of $P$ to be $\cE(P) = \bigsqcup_{p \in P} \cE(p)$. The use of $\bigsqcup$ is to emphasize that this is multiset union and so multiple ends on the same vertex or edge remain distinct whether they come from the same propagator or from different propagators.
        \item For a set of ends, $E \subseteq \cE(\cP)$, define the \emph{support} of $E$  to be $V(E) = \bigcup_{e \in E} V(e)$ and write $\Prop(E)$ for the set of propagators to which these ends belong. 
        \item Write $\cE = \cE(\cP)$ for the set of all ends of the diagram.
        \item Given a set of vertices $V\subseteq [n]$, we define $\End(V)$ to be the set of ends supported on $V$ \bas \End(V) = \{e \in \cE: V \cap V(e) \neq \emptyset\}\;. \eas
    \end{itemize}
\end{dfn}

\begin{eg} \label{eg:workingexample2}
In the diagram in Example \ref{eg:working example}, $V(p) = \{v_1, v_2, v_3, v_4, v_8, v_9\}$, $V(q) = \{v_1, v_5, v_6, v_7, v_9\}$, $V(r) = \{v_1, v_4, v_9\}$, and $V_s = \{v_7\}$. The propagators supported on, say $v_4$, is $\Prop(v_4) = \{p, r\}$. The ends supported on, say $v_1$, are $\End(v_1) = \{e_2 \textrm{ of } p, e_1 \textrm{ of } q, e_1 \textrm{ of } r\}$, while $\Prop (\End(v_1)) = \{ p, q, r\}$. The multiset of ends of the propagator set $P = \{ q, r, s\}$ is 
\[
\cE(P) = \{\{v_1, v_9\}, \{v_4\}, \{v_5, v_6\}, \{v_6,v_7\},\{v_1, v_9\}, \{v_7\}\}.
\]
Note that since the quantity $\cE(P)$ is defined as a multiset union, the end $e_9 = \{v_1, v_9\}$ appears twice, once for $q$ and once for $r$.
\end{eg}

Finally, we extend the notion of capacity from propagators to sets of ends. 
This leads into a definition of matroidal rank function on the set of ends, see Corollary~\ref{res:endmatroid}.

\begin{dfn}\label{dfn: cap stuff}
Let $(\cP, [n])$ be a generalized Wilson loop diagram.
   \begin{itemize}
       \item For a set of propagators $P \subseteq \cP$, write $c(P) = \sum_{p\in P} c(p)$ for the \emph{total capacity} of the propagators set $P$.
       \item Fix a propagator $p$. For a set of ends $E\subseteq \cE(p)$ of $p$, the \emph{capacity of $E$} is the smaller of the number of ends or the capacity of the propagator: $c(E) = \min\{|E|, c(p)\}$.
       \item For a set of ends $E\subseteq \cE(\cP)$ the \emph{capacity} is the sum across propagators of the capacity of those ends of $E$ which are from that propagator. Specifically \ba c(E) = \sum_{p \in \Prop(E)} c(E \cap \cE(p))\;. \label{eq:capacityfunction}\ea 
   \end{itemize} 
\end{dfn}
Note that the capacity of a set of propagators is just the sum of the individual capacities, but the capacity of a set of ends is constrained when the number of ends belonging to a particular propagator is less than its capacity.

\begin{eg} \label{eg:same support and cap}
Note that propagators may be different while still having the same support and capacity. See for instance the propagators in the three diagrams below  all of which have their sole propagator with the same support and capacity. The diagrams on the left and the right also have the same number of ends. However, as we show in Example \ref{eg:same support and cap matroid}, these three diagrams represent very different objects when used to study the boundaries of Wilson loop diagrams and their matroid structure.

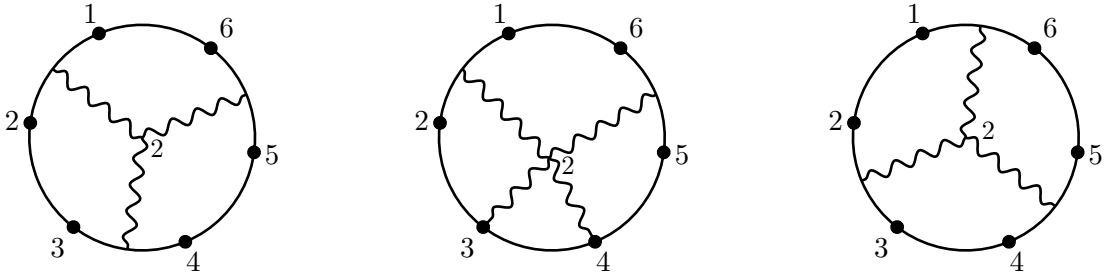
\begin{align*}
\begin{tikzpicture}[rotate=-67.5, line width=1, scale=1.5]
 \def \n {6}
 \draw circle(1)
 \foreach \v in{1,...,\n}
 {(360*\v/\n-360/\n+180:1)circle(.4pt)circle(.8pt)circle(1.2pt)circle(1.4pt) node[anchor=(360/\n*\v-360/\n -67.5]{$\v$}};
 \foreach \x/\y in {-0.866/-0.500,0.866/-0.500,0/1.00000000000000}{\draw[decorate,decoration={snake,amplitude=0.8mm}] (0,0) -- (\x,\y);}
 \draw (0,0) node[shift = {(.2,-.15)}] {\small2};
\end{tikzpicture} 
\qquad & \qquad  
\begin{tikzpicture}[rotate=-67.5, line width=1, scale=1.5]
 \def \n {6}
 \draw circle(1)
 \foreach \v in{1,...,\n}
 {(360*\v/\n-360/\n+180:1)circle(.4pt)circle(.8pt)circle(1.2pt)circle(1.4pt) node[anchor=360/\n*\v-360/\n-67.5]{$\v$}};
 \foreach \x/\y in {-0.866/-0.500,0.500/-0.866,1.00000000000000/0,0/1.00000000000000}{\draw[decorate,decoration={snake,amplitude=0.8mm}] (0.159,-0.091) -- (\x,\y);}
 \draw (0.159,-0.091) node[shift = {(.25,-.1)}] {\small2};
\end{tikzpicture} 
\quad & \quad
\begin{tikzpicture}[rotate=-67.5, line width=1, scale=1.5]
 \def \n {6}
 \draw circle(1)
 \foreach \v in{1,...,\n}
 {(360*\v/\n-360/\n+180:1)circle(.4pt)circle(.8pt)circle(1.2pt)circle(1.4pt) node[anchor=360/\n*\v-360/\n-67.5]{$\v$}};
 \foreach \x/\y in {-0.866/0.500,0/-1.00000000000000,0.866/0.500}{\draw[decorate,decoration={snake,amplitude=0.8mm}] (0,0) -- (\x,\y);}
 \draw (0,0) node[shift = {(.3,.1)}] {\small2};
\end{tikzpicture}
\end{align*}
\end{eg}

\begin{rmk}\label{rmk:ordinary WLD}
The notion of Wilson loop diagram from \cite{Wilsonloop, HeslopStewart, Amplituhedronsquared} is a special case of our notion of generalized Wilson loop diagram wherein every propagator has capacity 1 and exactly two ends with those ends on non-consecutive edges. For such diagrams, a propagator can be specified by the pair of indices of these two edges. This is how they are represented in the papers cited above.

To avoid confusion with generalized Wilson loop diagrams we will use the term \emph{ordinary Wilson loop diagram} to refer to this original notion of Wilson loop diagram.

It is via the ordinary Wilson loop diagrams that we have the connection with scattering amplitudes in $\mathcal{N}=4$ SYM theory, see Section \ref{sec:SYM}. Ordinary Wilson loop diagrams are also a variant of the notion of chord diagrams, where the chords are the propagators and the chords end on edges rather than vertices, separated along the edge when necessary so that all chord end points are distinct.

\end{rmk}

\subsection{Generalized Wilson loop diagrams with non-crossing propagators \label{sec:non-crossing}}

Generalized Wilson loop diagrams with non-crossing propagators are particularly well behaved and will form the foundation of our work in this paper. Loosely speaking, a generalized Wilson loop diagram with non-crossing propagators is any diagram that may be drawn in such a fashion that none of the propagators cross in the interior of the circle (although they may share an endpoint on the boundary circle).

\begin{dfn}\label{dfn:non-crossing}
    Let $W= (\cP, [n])$ be a generalized Wilson loop diagram. 
    Two propagators $p$ and $q$ of $W$ \emph{cross} if either 
    \begin{enumerate}
 
     \item there exist ends $a$ and $b$ of $p$ and ends $c$ and $d$ of $q$ such that $a<c<b<d <a $ in the cyclic order of vertices and edges of the set $\{v_1, e_1, \ldots, v_n, e_n\}$  or 
     \item there are three elements of $\{v_1, e_1, \ldots, v_n, e_n\}$ each of which are ends of both $p$ and $q$.
        \end{enumerate}
        Here we view the ends as elements of $\{v_1, e_1, \ldots, v_n, e_n\}$.

    We call a generalized Wilson loop diagram $W$ \emph{non-crossing} if no two propagators of $W$ cross.
\end{dfn}

Again, we note that the inequalities in the second part of Definition \ref{dfn:non-crossing} are strict so that two propagators may have ends on the same vertex or edge without necessarily causing a crossing. The first part of Definition \ref{dfn:non-crossing} takes into account the special case where there is a crossing without four distinct interlaced ends, namely when there are three shared ends. Note also that since both conditions in Definition \ref{dfn:non-crossing} require each involved propagator have at least two ends which are distinct as elements of $\{v_1, e_1, \ldots, v_n, e_n\}$, any propagator with all ends on the same edge or vertex (for example a single-ended propagator) can never cross another propagator.

\begin{eg} \label{eg:crossing}     
    For the first type of crossing, consider the generalized Wilson loop diagram in Example \ref{eg:working example}. Propagators $q$ and $r$ do not cross, however the propagator $p$ crosses both $q$ and $r$. Specifically, the ends of $q$ are $\{e_9, v_4 \}$ while the ends of $r$ are $\{e_5, e_6, e_9\}$. Since there are not a pair of ends of $q$, $\{a, b\}$ and a pair of ends of $r$, $\{c, d\}$ such that $a<c<b<d$, $q$ and $r$ do not cross. The ends of $p$ are $\{e_1, e_3, e_8\}$. We have that $e_1 < v_4 <e_8 <e_9$, showing that $p$ and $q$ cross, and $e_1 < e_5 <e_8 <e_9$ showing that $p$ and $r$ cross. 
    
    For the second type of crossing edges, consider the following diagram.
    
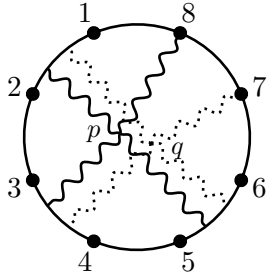
\begin{align*}
    \begin{tikzpicture}[rotate=-67.5, line width=1, scale=1.5]
    	\def \n {8}
    	\draw circle(1)
    	\foreach \v in{1,...,\n}
    	{(360*\v/\n-360/\n+180:1)circle(.4pt)circle(.8pt)circle(1.2pt)circle(1.4pt) node[anchor=360/\n*\v-360/\n-67.5]{$\v$}};
    	\foreach \x/\y in {-0.966/-0.259,0.500/-0.866,0.866/0.500,0/1.00000000000000}{\draw[dotted, decorate,decoration={snake,amplitude=0.8mm}] (0.100,0.094) -- (\x,\y);}
    	\draw (0.100,0.094) node[shift = {(.35,-.10)}] {\small $q$};
    	\foreach \x/\y in {-0.866/-0.500,0.259/-0.966,0.966/0.259,-0.707/0.707}{\draw[decorate,decoration={snake,amplitude=0.8mm}] (-0.087,-0.125) -- (\x,\y);}
    	\draw (-0.087,-0.125) node[shift = {(-.35,-.0)}] {\small $p$};
    \end{tikzpicture}
    \end{align*}
    Let $p$ be the propagator drawn with a solid line and ends $\{e_1, e_3, e_5, v_7\}$ and $q$ be the propagator drawn with a solid line and ends $\{e_1, e_3, e_5, v_7\}$. Namely,  they share three endpoints. While there is no set of ends $\{a, b\}$ of $p$ and $\{c,d\}$ of $q$ satisfying $a<c<b<d <a $ in the cyclic order of edges and vertices, these propagators are clearly crossing in any pictoral representation. 
\end{eg}

For any fixed vertex or edge, the non-crossing structure defines an (essentially) canonical order on the propagators with an end on that vertex or edge and hence on the ends at that vertex or edge. These end orders can then be concatenated to give a cyclic order on all ends. 

We need one lemma in order to define this order. Recall from the original definition of a generalized Wilson loop diagram (Definition \ref{dfn:gWLD}) that $C$ is the cyclically ordered set $\{v_1, e_1, v_2, e_2, \ldots, v_n, e_n\}$. Given $a\in C$ we will for this section write $E_a$ for the set of ends on $a$ and then $\Prop(E_a)$ is the set of propagators to which these ends belong.
\begin{lem}\label{res:non-crossing pairs}
    Let $W = (\cP, [n])$ be a non-crossing generalized Wilson loop diagram, and $a \in C$. Take $p,q\in \Prop(E_a)$. Then either for all $b \in \cE(p) \setminus E_a$ and all $c \in \cE(q) \setminus E_a$ we get $b\leq_a c$, or the same holds with $p$ and $q$ reversed.
\end{lem}

Also note that if one or both of $p$ and $q$ only has copies of $a$ as its ends, then the statement of the lemma is vacuously true.

\begin{proof}
    Suppose that the property in the statement does not hold. Without loss or generality we may suppose that in the linear order $<_a$ we can find ends $b_1 <_a c_1$ and $c_2 <_a b_2$ where $b_1$ and $b_2$ are $p$ ends and $c_1$ and $c_2$ are $q$ ends. If $b_1=c_2$ and $c_1=b_2$ then taking also two ends from $E_a$, one a $p$ end and one a $q$ end, we get a crossing by the first part of the definition of crossing. Suppose now that $b_1\neq c_2$ or $c_1\neq b_2$. In that case we have at least one of $c_2<b_1<c_1$, $b_1<c_2<b_2$, $b_1<c_1<b_2$, or $c_2<b_2<c_1$. Consider also an end in $E_a$ of whichever of $p$ or $q$ has an end in the middle of the chain of inequalities. Then we have in the cyclic order on $C$ a $q$ end, a $p$ end, a $q$ end, and then a $p$ end with all four of these ends distinct. This is also a crossing, giving in all cases a contradiction and hence proving the lemma.
\end{proof}

Before we can define an order on all ends we also need to define an ordering of the propagator sets $\Prop(E_a)$.
\begin{dfn}\label{dfn:new def of prop order at edge or vertex}
    Let $(\cP, [n])$ be a non-crossing generalized Wilson loop diagram. For each $a\in \{v_1, e_1, \ldots, v_n, e_n\}$, define an order on $\Prop(E_a)$ as follows: 
    \begin{itemize}
        \item Fix a total order for the set of propagators with all ends in $E_a$ and set these propagators to come before all other propagators in $\Prop(E_a)$. 
        \item For the remaining propagators, i.e.\ when $p, q\in \Prop(E_a)$ with $\cE(p), \cE(q) \not \subseteq  E_a$, order $p$ and $q$ in the reverse of the order that their ends outside of $E_a$ have under Lemma~\ref{res:non-crossing pairs}.
    \end{itemize}
\end{dfn}

For and example of this linear order, see Example \ref{eg:non-crossing order}.

Note that this local linear ordering of $\Prop(E_a)$ is canonical up to a choice of orderings of propagators with a single distinct end (with possible multiplicity). We may use this local linear ordering of propagators at the edges and vertices of the diagram to define a global cyclic ordering of the ends of the diagram $\cE$.

\begin{dfn}\label{dfn:global cyclic end order}
    For any propagator with multiple copies of the end $a$, fix an ordering on the copies of the element $a$ to refine the order on $\Prop(E_a)$ into an order on $E_a$. Define a cyclic order on $\cE$ by concatenating the resulting order on $E_{v_1}$ with the order on $E_{e_1}$ with the order on $E_{v_2}$, and so on following the cyclic order on $\{v_1, e_1, \ldots, v_n, e_n\}$.
\end{dfn}

Note that, like the ordering on the propagators, the ordering on the ends is canonical up to a choice of ordering of multiple ends on the same edge or vertex belonging to the same propagator.

\begin{rmk}\label{remark drawing is non-crossing}
    The point of the orders of Definition~\ref{dfn:new def of prop order at edge or vertex} is that they are the orders given by a non-crossing drawing of the generalized Wilson loop diagram, up to the ambiguities of propagators with all ends in $E_a$ (which can be positioned arbitrarily) and the order of multiple ends. Furthermore, a drawing of the diagram according to these orders will, with suitable positioning inside the circle, be non-crossing.

    This also explains why the order of the propagators with at least one end outside of $E_a$ are reversed relative to Lemma~\ref{res:non-crossing pairs}: those propagators whose other ends come first after $a$ must come later in the order on $a$ in order not to cross, see Example~\ref{eg:non-crossing order}. 

    We briefly sketch a proof that such a non-crossing drawing can always be made: Note that the cyclic order on $\cE$ of Definition~\ref{dfn:global cyclic end order} refines the pre-order on ends used in the definition of non-crossing, and, similarly to the proof of Lemma~\ref{res:non-crossing pairs}, refines it in such a way that none of what were equalities become crossings. Now that there are no equalities, the only possible kind of crossing is the kind with four distinct interleaved ends, which is the usual notion of crossing for a set partition. Therefore, the propagators give a non-crossing set partition on the ends, and it is well known that non-crossing set partitions can be drawn without crossings in the manner that we require. 
\end{rmk}

\begin{eg} \label{eg:non-crossing order}
To illustrate the global and local cyclic orderings of ends of a generalized Wilson loop diagram with non-crossing propagators, consider the following generalized Wilson loop diagram:
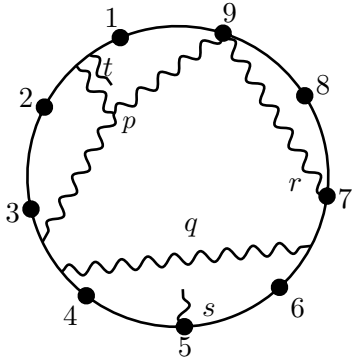
\begin{center}
\begin{tikzpicture}[rotate=-67.5, line width=1, scale=2]
	\def \n {9}
	\draw circle(1)
	\foreach \v in{1,...,\n}
	{(360*\v/\n-360/\n+180:1)circle(.4pt)circle(.8pt)circle(1.2pt)circle(1.4pt) node[anchor=360/\n*\v-360/\n-67.5]{$\v$}};
	\foreach \x/\y in {-0.940/-0.342,0.058/-0.998,-0.766/0.643}{\draw[decorate,decoration={snake,amplitude=0.8mm}] (-0.549,-0.232) -- (\x,\y);}
	\draw (-0.549,-0.232) node[shift = {(.2,-.15)}] {\small $p$};
	\draw[decorate,decoration={snake,amplitude=0.8mm}] (0.287,-0.958) -- (0.766,0.643);
	\draw[decorate,decoration={snake,amplitude=0.8mm}] (0.940,-0.342) -- ++({atan(-0.364)}:-.25 cm);
 \draw (.9,-.15) node {$s$};
	\draw[decorate,decoration={snake,amplitude=0.8mm}] (0.500,0.866) -- (-0.766,0.643);
	\draw[decorate,decoration={snake,amplitude=0.8mm}] (-0.973,-0.231) -- ++({atan(0.237)}:.25 cm);
	\draw (.35,.7) node {$r$};
    \draw (.35,-.05) node {$q$};
    \draw (-0.85,-0.15) node {$t$};
\end{tikzpicture}
\end{center}

Let $p$ be the propagator with $3$ ends, $q$ the propagators with ends $\{e_3, e_6\}$, $r$ the propagator with ends $\{v_7, v_9\}$, $s$ the single-ended propagator with end $v_5$ and $t$ the single-ended propagator with end $e_1$. The capacities of the propagators do no matter for this example, so we set them all to be $1$ for simplicity. 

From Definition \ref{dfn:global cyclic end order}, we see that the end $e_1$ of $p$ precedes $e_3$ of $q$, which precedes the end $v_7$ of $r$. 

Furthermore, there are three cases of multiple propagators supported on the same vertex or edge: $\Prop(e_3) = \{p,q\}$, $\Prop(v_9) = \{p, r\}$, $\Prop(e_1) = \{p, t\}$. Note that, starting at $e_3$, all the other ends of $p$ come weakly after all the other ends of $r$, and thus (by Definition \ref{dfn:new def of prop order at edge or vertex}) the end of $p$ precedes the end of $q$ on $e_3$. Similarly, starting at $v_9$ the other ends of $p$ precede the other ends of $r$, and thus (by Definition \ref{dfn:new def of prop order at edge or vertex}) the end of $r$ precedes the end of $p$ at $v_9$. However, at $e_1$, $t$ does not have other ends whose position one can compare with the other ends of $p$. Therefore, one has freedom to place these elements in the linear order. Definition \ref{dfn:new def of prop order at edge or vertex} chooses the convention of placing these ends first. 

\end{eg}

\begin{eg}\label{eg:manyrepscrossing}

Note that in the case of crossing propagators, one cannot impose a canonical ordering of the ends of the vertices. This is because there is no longer a canonical ordering of the propagators as given by Lemma \ref{res:non-crossing pairs}. For instance, in the following two diagrams, the dotted and the solid propagators cross and there are multiple ways to draw them in the interior of the diagram with the same minimal number of crossings.
 \begin{align*}
 \begin{tikzpicture}[rotate=-67.5, line width=1, scale=2]
	\def \n {7}
	\draw circle(1)
	\foreach \v in{1,...,\n}
	{(360*\v/\n-360/\n+180:1)circle(.4pt)circle(.8pt)circle(1.2pt)circle(1.4pt) node[anchor=360/\n*\v-360/\n-67.5]{$\v$}};
	\foreach \x/\y in {-0.901/-0.434,0.989/0.149,-0.365/0.931}{\draw[dotted, decorate,decoration={snake,amplitude=0.8mm}] (-0.092,0.215) -- (\x,\y);}
	\draw (-0.092,0.215) node[shift = {(.2,-.05)}] {\small2};
	\foreach \x/\y in {-0.223/-0.975,0.989/-0.149,-0.075/0.997}{\draw[decorate,decoration={snake,amplitude=0.8mm}] (0.230,-0.042) -- (\x,\y);}
\end{tikzpicture} \quad & \quad \begin{tikzpicture}[rotate=-67.5, line width=1, scale=2]
\def \n {7}
\draw circle(1)
\foreach \v in{1,...,\n}
{(360*\v/\n-360/\n+180:1)circle(.4pt)circle(.8pt)circle(1.2pt)circle(1.4pt) node[anchor=360/\n*\v-360/\n-67.5]{$\v$}};
\foreach \x/\y in {-0.901/-0.434,0.989/-0.149,-0.365/0.931}{\draw[dotted, decorate,decoration={snake,amplitude=0.8mm}] (-0.092,0.116) -- (\x,\y);}
\draw (-0.092,0.116) node[shift = {(.2,-.05)}] {\small 2};
\foreach \x/\y in {-0.223/-0.975,0.989/0.149,-0.075/0.997}{\draw[decorate,decoration={snake,amplitude=0.8mm}] (0.230,0.2) -- (\x,\y);}
\end{tikzpicture} 	
\end{align*}
\end{eg}

\section{Background}\label{sec:background}

\subsection{Physics motivation}\label{sec:SYM}

The original motivating question for studying Wilson loop diagrams and their boundaries is the computation of amplitudes in $\cN=4$ supersymmetric Yang--Mills (SYM). The reader who is uninterested can skip Section~\ref{sec:SYM}, but we outline this motivation as it provides richness and justification for everything we study.

Wilson loop diagrams are graphical representations of particle interactions in $\cN=4$ SYM theory. Most generally, a holomorphic Wilson loop gives the amplitude for an N${}^k$MHV interaction\footnote{Here MHV means maximal helicity violation. The MHV amplitudes are well understood. Interactions that are one less than MHV are denoted NMHV while those $k$ less than MHV are denoted N${}^k$MHV and are the ones that are currently an active area of research.} on $n$ particles. Each holomorphic Wilson loop can be decomposed in terms of integrals in twistor space, each of which, diagrammatically, correspond to an N${}^k$MHV diagram \cite{Adamo:2011pv, Boels:2007qn, Britto:2005fq}. For this paper, we are interested in a diagrammatic representation for the integrals for tree level interactions in the dual twistor space, which we call (ordinary) Wilson loop diagrams \cite{HeslopStewart,LM1}. We find this perspective insightful because of the explicitness of the connection to matroids \cite{Wilsonloop, casestudy, generalcombinatoricsI, generalcombinatoricsII,Amplituhedronsquared}.

It is well known that $\cN = 4$ SYM is a finite theory, and therefore does not suffer from the ultraviolet divergences present in many field theories. However, each of the integrands representing an N${}^k$MHV interaction is a rational function, each of the poles of which are conjectured to correspond to an ultraviolet singularity. It is further conjectured in \cite{Cachazo:2004kj,hodges:2013eliminating} that these singularities cancel when taking the sum that corresponds to the amplitude for an N${}^k$MHV interaction on $n$ particles. These singularities appear to occur in ways that can be represented by certain diagrammatic limits of the original diagrams. See, for instance, the discussion of \emph{boundary diagrams} in \cite{casestudy}. In \cite{cancellation}, one of us with other collaborators show that this cancellation occurs according to the conjecture for the degree one singularities. 

The proof of cancellation of degree one singularities comes from a correspondence between Wilson loop diagrams and subspaces of the positive Grassmannian, $\Gr(k, n)$, called positroid cells \cite{Wilsonloop, generalcombinatoricsII}. The authors of \cite{cancellation} show that the singularities lie on the boundaries of these cells, and thus cancel by appearing on shared boundaries with combined opposite magnitudes. The geometry of the positroid cells and the positive Grassmannians are well studied objects both in the physics and mathematics literature \cite{Arkani-Hamed:2013jha, Even-Zohar:2021sec,  GrBall, Postnikov}. In fact, there is a large body of literature connecting the on shell interactions of $\cN=4$ SYM theory to positroids, giving rise to a geometry called the Amplituhedron \cite{Arkani-Hamed:2013kca, galashinlam18,AmplituhedronDecomposition}. While that body of work is related to the work presented here, the N${}^k$MHV perspective, and thus the Wilson loop diagram perspective, corresponds to the off shell scattering amplitudes \cite{Cachazo:2004kj}, and the Wilson loop diagram work outlined above gives an Amplituhedron analogue in the off shell case.

From this perspective, there is a special class of ordinary Wilson loop diagrams that arises from the physics. 

\begin{dfn}\label{dfn:admissibleWLD} An ordinary Wilson loop diagram $W= (\cP, [n])$ is \emph{admissible} if \footnote{Note that this is not the standard definition of admissible Wilson loop diagrams. Rather, this corresponds to \emph{weakly admissible diagrams} from \cite{generalcombinatoricsI, generalcombinatoricsII}. However, since this is the only notion of admissibility for ordinary Wilson loop diagrams that we use in this paper, we omit the word weakly}
\begin{enumerate}
	\item no pair of propagators $p, q \in \cP$ cross, and
	\item for every subset $P \subseteq \cP$, $|V(P)|\geq |P|+3$. 
\end{enumerate}
\end{dfn}	
In particular, the last condition imposes that a propagator is always supported by exactly $4$ distinct vertices. Specifically, it never has ends on two consecutive edges. 

For admissible ordinary Wilson loop diagrams, there is a well defined mapping to positroid cells, and a well defined bijection between these same Wilson loop diagrams and tree level interactions in $\cN=4$ SYM theory. Therefore, one may associate to each tree level interaction in $\cN=4$ SYM theory to a positroid cell. The singularities of the tree level interactions all lie on the boundaries of the associated cell. However, the relationship between boundaries of positroid cells and the singularities of tree level interactions is far from simple. Furthermore, there is not a good methodology for identifying the boundaries of positroid cells from the physical picture. In order to understand the intricacies of the boundaries of the positroid cells arising from Wilson loop diagrams, one needs to move beyond the simple diagrammatics currently used in the physics literature \cite{HeslopStewart, casestudy}, however we wish to stay diagrammatic in order to maintain both physical and combinatorial control over the behaviour of the boundaries. It is with this goal in mind that we work with generalized Wilson loop diagrams. 

\subsection{Matroidal background \label{sec:MatroidsandRado}}

In this section, we present several results about matroids, transversal matroids, Rado matroids, and positroids that are needed for our results. See, for example, \cite{Oh43,OxleyMatroidBook, DevosNotes, Postnikov}, Experts on the subject may skip this exposition.

\subsubsection{Matroids \label{sec:matroids}}

This section is a review for the experts, but is included here for completeness. An interested reader is encourage to turn to textbooks on the subject, such as \cite{OxleyMatroidBook} for a more thorough discussion of the topics contained herein.

A \emph{matroid} $\mathcal{M}$ is defined by a \emph{ground set} $E$ and the data of the independence of subsets of the ground set. There are several equivalent representations of the independence of subsets that define a matroid. For this paper, we review independence systems, sets of bases, rank functions, circuit sets, and flats. 

We begin with a definition of an \emph{independence system}. The members of the independence system are called \emph{independent sets}. Any subset of $E$ not in the independence system is a \emph{dependent set}.

\begin{dfn}[(I1)-(I2) \cite{OxleyMatroidBook}] \label{dfn:indep system}
	An \emph{independence system}, $\sI$, is a set of subsets of $E$ that is downward closed. In other words, if $I \subseteq \sI$, then any subset $J \subseteq I$ is also in $\sI$.
\end{dfn}

Given an independence system (on a ground set $E$) one may define a matroid from it by adding an exchange property.
 
\begin{dfn}[(I1)-(I2) \cite{OxleyMatroidBook}] \label{dfn:indep exchange matroid}
Let $E$ be a ground set and $\sI$ an independence system of $E$. Then $\cM=(E, \sI)$ is a \emph{matroid} if and only if $\sI$ satisfies the independence exchange property: \ba \forall I,\; J \in \sI \textrm{ with } |I| > |J|, \; \exists a \in I \textrm{ such that } J \cup a \in \sI\;. \label{eq:indepexch}\ea 
\end{dfn}

Because of the downward closed property, we may define an independence system in terms of its largest independent sets. For a matroid, these are called the \emph{bases}. A matroid can be defined in terms of its set of bases, $\cB$. 

\begin{dfn}[Lemma 1.2.2 \cite{OxleyMatroidBook}]\label{dfn:basis exchange}
Let $E$ be a ground set and $\cB$ a set of subsets of $E$. Then $\cM = (E, \cB)$ is a matroid with bases $\cB$ if and only if $\cB$ is non-empty and satisfies the basis exchange property: \bas \forall A, \; B \in \cB, \quad a \in A \setminus B, \quad \exists b \in B \setminus A \textrm{ such that } (A \setminus a) \cup b \in \cB \;.\eas 
\end{dfn}

Any subset of an element of $\cB$ is an independent set. If a set is not contained in an element of $\cB$, then it is a dependent set. Every element of $\cB$ is of the same size, which can easily seen from the independence exchange property.

Each matroid may also be defined by a \emph{rank function} $\rk: 2^E \rightarrow \Z_{\geq 0}$ that assigns a non-negative integer to each subset of the ground set. The rank function is a monotonic subadditive function, bounded above by the number of elements in the subset. 

\begin{dfn}[Corollary 1.3.4 \cite{OxleyMatroidBook}]\label{dfn:rank function}
Let $E$ be a ground set and $\rk$ a function from the set of subsets of $E$ to the non-negative integers. Then $\cM = (E, \rk)$ is a matroid if and only if $\rk$ satisfies the following properties:
\begin{itemize}
	\item $\rk(\emptyset) = 0$
	\item $\forall x \in E, \quad \rk(A) \leq \rk(A \cup x ) \leq \rk(A) +1$
	\item $\forall A, B \subseteq E, \quad \rk(A \cup B) + \rk(A \cap B) \leq \rk(A) + \rk(B)$.
\end{itemize}
\end{dfn}

The definition of a matroid in terms of a rank function relates to the definition in terms of independent sets via the rule that a set is independent if and only if its rank is its cardinality.

Because of the downward closed property, one may additionally define an independence system by its set of smallest dependent sets. For a matroid, these are called the \emph{circuits}. Any dependent set of a matroid contains a circuit. Note that minimality implies that any proper subset of a circuit is independent. This gives the fourth equivalent definition of a matroid.

\begin{dfn}[Corollary 1.1.5 \cite{OxleyMatroidBook}]\label{dfn:circuit condition}
Let $E$ be a ground set and $\cC$ a set of subsets of $E$. Then $\cM = (E, \cC)$ is a matroid if and only if $\cC$ satisfies the \emph{circuit condition}: \bas \forall C_1, C_2 \in \cC, \; x \in C_1 \cap C_2, \exists C \subseteq (C_1 \cup C_2) \setminus x \textrm{ such that } C \in \cC \;.\eas 
\end{dfn}

Finally, given a set $A \subset E$, and a rank function on $E$, we define the \emph{closure} of $A$ 
$$\cl(A) = \{x \in E : \rk(A \cup x) = \rk(A)\}$$ as the set of elements that can be added to $A$ without increasing the rank. A \emph{flat} is a set $F \subset E$ such that $\cl(F) = F$. 

\begin{dfn}[Section 1.4, Exercise 11 \cite{OxleyMatroidBook}]\label{dfn:flats}  For $E$ a ground set and $\cF$ a set of subsets, $\cM = (E, \cF)$ is a matroid if and only if $\cF$ satisfies the following properties:
\begin{itemize}
        \item $E \in \cF$
	\item If $F, G \in \cF$, then $F \cap G \in \cF$
	\item For $F \in \cF$, and $x \not \in F$, then the set of elements of $\cF$ containing both $F$ and $x$ has an unique minimal element under set inclusion. 
\end{itemize} Any $\cF$ satisfying the above is the set of flats of a matroid. \end{dfn}

There is a special class of flats that we discuss in this paper called \emph{cyclic flats}.  A cyclic flat is a set which is both a flat and a union of circuits.\footnote{The word cyclic, in this situation, arises from the fact that, in matroid theory a cycle is a union of circuits. In order to avoid confusion with a contradictory use of the word cycle in graph theory, we avoid defining cycles directly in the text of this paper.}

\begin{dfn} \label{dfn:loopsandco-loops}
	There are two special classes of elements of a ground set of a matroid, $M = (E, \cB)$ that are worth noting:\begin{itemize}
		\item an element $x \in E$ is a \emph{loop} if it is not an element of any basis set: $\forall \; B \in \cB, \; x \not \in B$ [Section 1.1 \cite{OxleyMatroidBook}]  
		\item an element $x \in E$ is a \emph{co-loop} if it is an element of every basis set: $\forall \; B \in \cB, \; x \in B$. [Section 2.1 \cite{OxleyMatroidBook}]
	\end{itemize}
\end{dfn}

Here are a few further notions related to matroids that will be useful later. 

\begin{dfn}[Section 1.1, \cite{OxleyMatroidBook}] \label{dfn:representable matroid}
A matroid $\cM$ is \emph{representable} over a field $k$ if there is a matrix with elements in $k$ and a bijection between the ground set of $\cM$ and the columns of the matrix so that the rank functions agree. In this case the bases of the matroid and the column space of the matrix also agree, and likewise for the independent sets.
\end{dfn}

\begin{dfn}[Section 1.2, \cite{OxleyMatroidBook}]\label{dfn:uniformmatroid}
	A \emph{uniform matroid} of rank $k$ is a matroid over a base set $E$ where every subset of size $<k$ is independent. We denote a uniform matroid of rank $k$ by $U_n^k = (E, \cB)$ where $|E|= n$ and $\cB = \{B \subseteq E :\ |B| = k \}$. 
\end{dfn}

Note that for a uniform matroid, the circuits are the sets of size $k+1$, all flats of size less than $k$ are independent sets, and the only cyclic flat is $E$.

\begin{dfn}[Section 1.3 \cite{OxleyMatroidBook}]\label{dfn:matroid restriction} 
    Let $\cM$ be a matroid with underlying set $E$ and let $S\subseteq E$. Then the \emph{restricted matroid} $\cM|_S$ is the matroid on underlying set $S$ where a subset of $S$ is independent in $\cM|_S$ if and only if it is independent in $\cM$. 
\end{dfn}

In other words, the set of bases of a restricted matroid $\cM|_S$ is the set $\cB|_S = \{B \cap S : |B\cap S| \text{ maximal among } B \in \cB\}$ where $\cM = (E, \cB)$. 

\begin{dfn}[(3.1.14) \cite{OxleyMatroidBook}]\label{dfn:matroid contraction} 
  Let $\cM$ be a matroid with underlying set $E$ and set of bases $\cB$ and let $S\subseteq E$ be a subset of $E$ of rank $\rk(S) = r$. Then the \emph{contracted matroid} $\cM/S$ is the matroid on underlying set $E \setminus S$
with set of bases $\cB/S = \{B \setminus S : |B\cap S| \text{ maximal among } B \in \cB\}$.
\end{dfn}

\begin{dfn}[Proposition 4.2.8 \cite{OxleyMatroidBook}]\label{dfn:matroid direct sum} 
Let $\cM$ be a matroid with underlying set $E$. If $E$ can be partitioned into $E_1\sqcup E_2 = E$ such that the independent sets in $\cM$ are precisely the unions of an independent set in $\cM|_{E_1}$ with an independent set in $\cM|_{E_2}$ then $\cM$ is the \emph{direct sum} of $\cM|_{E_1}$ with $\cM|_{E_2}$, written $\cM = \cM|_{E_1} \oplus \cM|_{E_2}$.

If $\cM$ cannot be written as a direct sum of two nonempty matroids then $\cM$ is \emph{connected}.
\end{dfn}
Equivalently phrased in terms of rank functions, $\cM = \cM|_{E_1} \oplus \cM|_{E_2}$ if and only if the rank function of $\cM$ is not just subadditive, but additive across the partition $E_1\sqcup E_2$. That is for $S\subseteq E$, $\rk(S) = \rk(S\cup E_1) + \rk(S\cup E_2)$. Using the other matroid properties, it suffices to check this property only on $E$ itself: $\rk(E) = \rk(E_1) + \rk(E_2)$ [(4.2.13) \cite{OxleyMatroidBook}]

We conclude with a general fact about flats of uniform matroids that we use in Section \ref{sec:positive gWLD}.

\begin{lem}\label{res:flatcondpos}
	A matroid $M$ is uniform if and only if every connected dependent flat has full rank.
\end{lem}
Here by a flat being connected, we mean the matroid restricted to the flat is connected.
\begin{proof}
	By definition $M$ is a uniform matroid of rank $r$, if and only if every set of size $\leq r$ is independent. Therefore if $F$ is a connected dependent flat of $M$, it has size at least $r+1$ and so contains a basis and must be of full rank. 
	
	To prove the other direction, suppose $F$ is a connected dependent flat of $M$ with $\rk(F) < r$. Then $F$ contains a circuit, $U \subset F$ with $|U| \leq \rk(F) +1 \leq r$. In other words, $U$ is a dependent set with $|U| \leq r$, forcing $M$ to not be uniform.
\end{proof}

\subsubsection{Transversal and Rado matroids \label{sec:Rado}} 

Transversal matroids are a special class of matroids where the independence data can be encoded in a bipartite graph by Hall's matching theorem. In order to define this, we need to establish some notation that we will continue to use through the paper. For a more thorough discussion of the subject, see \cite{OxleyMatroidBook}, section 1.6.

Let $G = (X \sqcup Y, E)$ be a bipartite graph with vertex sets $X$ and $Y$ and edge set $E$. A \emph{matching}, $M$, is a subset of $E$ such that such that no edges in $M$ share a common vertex. A vertex in $X$ or $Y$ is covered by $M$ if it is the endpoint of an edge in $M$. We write $M\cap X$ and $M\cap Y$ to indicate the vertices in $X$ and $Y$ that are covered by $M$. For $A \subseteq X$, an $A$\emph{-perfect matching} is a matching that covers every vertex of $A$. In other words, if $M$ is a matching in $G$ with $A = M \cap X$, then $M$ is an $A$-perfect matching. By Hall's matching theorem, the subset $A$ has an $A$-perfect matching if and only if, for every subset $U \subseteq A$, the neighborhood of $U$ in $Y$ has more elements than $U$: $\forall U \subseteq A, \; |U| \leq |N(U)|$. A \emph{maximum matching} is one which contains the largest possible number of edges.

Every bipartite graph gives rise to a transversal matroid.

\begin{dfn}\label{dfn:transeverse matroid}
The bipartite graph $G = (X \sqcup Y, E)$ defines a \emph{transversal matroid} $\cM_G$ over the ground set $X$ as follows. The set $A \subseteq X$ is independent in $\cM_G$ if and only if there is an $A$-perfect matching in $G$. That is, $A \subseteq X$ is independent in $\cM_G$ if and only if $\forall U \subseteq A, \; |U| \leq |N(U)|$.
\end{dfn}

The following are useful additional facts about $\cM_G$. The set $A$ is a basis of $\cM_G$ if and only if the corresponding $A$-perfect matching is a maximum matching. If $A$ is an independent set in $\cM_G$ but not a basis, then there is an algorithm to extend any $A$-perfect matching to a maximum matching. 

A Rado matroid is a generalization of a transversal matroid, where the independence structure is derived from a bipartite graph $G = (X \sqcup Y, E)$ together with a matroid structure $\sfM_Y$ on the set $Y$. In the case that $Y$ is a discrete matroid, that is, when every subset of $Y$ is independent, then the Rado matroid is simply the transversal matroid of the bipartite graph.

Some further matching language incorporating the structure of $\sfM_Y$ will be useful, which we give following \cite{DevosNotes}.

\begin{dfn}\label{dfn:indepmatching} For $G = (X \sqcup Y , E)$ where $Y$ has a matroid structure $\mathsf{M}_Y$, a matching $M$ in $G$ is \emph{independent} if the set of vertices which is covered in $Y$ (still denoted $M \cap Y$) is an independent set in $\sfM_Y$. 
\end{dfn}

\begin{dfn} \label{dfn:saturation_closure}We say a vertex $y\in Y$ is \emph{saturated} by a matching $M$ in $G$ if and only if it is in the closure of the covered vertices: $y\in \cl(M\cap Y)$, while a vertex $x\in X$ is \emph{saturated} by $M$ if and only if it is covered by $M$, $x \in M\cap X$. \end{dfn}

Note that this means that vertices with rank $0$ in $Y$ are saturated by the empty matching.

In the case that $\sfM_Y$ is a discrete matroid then all matchings are independent matchings and the notion of saturated becomes identical to the notion of covered both in $X$ and in $Y$. Relatedly, by taking a discrete matroid structure on $X$ we could give the definition of saturation symmetrically, on both sides saying a vertex is saturated if it is in the closure of the covered vertices. This explains why, in the usual theory of matchings on bipartite graphs, covered and saturated are interchangeable terms, though we will reserve saturation for the meaning given in Definition~\ref{dfn:saturation_closure}.

Rado's theorem states that there is an $A$-perfect independent matching in $G$ if and only if, for every $U \subseteq A$, the rank of the neighborhood of $U$ exceeds the size of $U$: $\forall U \subseteq A,\; |U| \leq \rk(N(U))$. Note that if $\sfM_Y$ is the discrete matroid on $Y$, this becomes Hall's matching theorem.

We can use this to define the Rado matroid defined by a bipartite graph $G$, with matroid structure $\sfM_Y$.

\begin{dfn} \label{dfn:Rado}
The \emph{Rado matroid} $\cM_{G,\sfM_Y}$ defined by $G = (X \sqcup Y , E)$ and $\sfM_Y$ has $X$ as its ground set. A set $A \subseteq X$ is an independent set in  $\cM_{G,\sfM_Y}$ if and only if there is an $A$-perfect independent matching in $G$; that is, $\forall U \subseteq A,\; |U| \leq \rk(N(U))$.
\end{dfn}

The following lemma regarding disconnected Rado matroids will be useful later. Recall the definition of direct sum of matroids and connectivity of matroids from Definition~\ref{dfn:matroid direct sum}.
\begin{lem}\label{res:rado disconnected}
    Suppose the bipartite graph $G=(X\sqcup Y, E)$ is disconnected and write $G=G_1\sqcup G_2$ with $V(G_i) = X_i \sqcup Y_i$, $X_1\sqcup X_2 = X$, $Y_1\sqcup Y_2 = Y$. Suppose also that $\sfM_Y = \sfM_{Y_1} \oplus \sfM_{Y_2}$ where $\sfM_{Y_i} = \sfM_Y|_{Y_i}$. Then $\cM_{G, \sfM_Y} = \cM_{G_1, \sfM_{Y_1}} \oplus \cM_{G_2, \sfM_{Y_2}}$.
\end{lem}

\begin{proof}
    Suppose $S_1$ is independent in $\cM_{G_1, \sfM_{Y_1}}$ and $S_2$ is independent in $\cM_{G_2, \sfM_{Y_2}}$. By definition, for every $U_i\subseteq S_i$, $|U_i| \leq \rk_{\sfM_{Y_i}}(N_{G_i}(U_i))$. Then summing over $i=1,2$ and using the direct sum structure of $\sfM_Y$ (phrased in terms of rank) we get $|U| \leq \rk_{\sfM_Y}(N(U))$. Since every $U\subseteq S_1\sqcup S_2$ is formed in this way, $S_1\sqcup S_2$ is independent in $\cM_{G, \sfM_Y}$. On the other hand, if $S_i$ is dependent in $\cM_{G_i, \sfM_{Y_i}}$ then it is dependent in $\cM_{G, \sfM_Y}$ and so is any set containing it. Therefore $\cM_{G, \sfM_Y} = \cM_{G_1, \sfM_{Y_1}} \oplus \cM_{G_2, \sfM_{Y_2}}$.
\end{proof}

\subsubsection{Positroids \label{sec:positroid}}  

In this section, we review the definition of positroids and their geometry. This is meant to be a very minimal review of the details needed for this paper. For a fuller exposition, see \cite{Postnikov, Juggling, GrBall}.

As before $[n]$ denotes the set $\{1, \ldots, n\}$ with the standard cyclic order, and for $i \in n$, we have the associated linear orders $<_i$ as in Definition~\ref{dfn:cyclic orders}. Sometimes we write $<$ instead of $<_1$. 

\begin{dfn} \label{dfn:positroid} \cite{Oh} 
	A \emph{positroid}, $\cM$, is a matroid with a cyclically ordered ground set, $[n]$, such that it can be represented over $\R$ by a matrix with non-negative maximal minors.
\end{dfn}

There are a few things of note here. First, the positivity condition defining positroids is not in general invariant under arbitrary permutations of the ground set, only under cyclic permutations. Consequently, positroids are generically isomorphic only under cyclic permutation of the ground set while general matroids are isomorphic under arbitrary permutations of the ground set. 
Secondly, one may restate Definition \ref{dfn:positroid} to say that a positroid of rank $k$ can be represented by a point in the Grassmannian, $\Grall(k, n)$, with non-negative Pl\"{u}cker coordinates. For a more complete discussion of the Grassmannian as a manifold and Pl\"{u}cker coordinates, see sections 2.1 and 2.2 of \cite{Postnikov}. We refer to the subset of $\Grall(k, n)$ with non-negative Pl\"{u}cker coordinates as the positive Grassmannians, $\Gr(k, n)$. One may define a positroid cell, $\Sigma(\cM)$ as the set of points in $\Gr(k,n)$ that represent the positroid $\cM$. These positroid cells form a CW-complex over $\Gr(k,n)$ [Theorem 3.5 \cite{Postnikov}]. 

This CW-complex structure gives rise to the idea of a positroid cell being a boundary of another cell. Algebraically, in our context, this becomes the following definition.
\begin{dfn}[Section 17 \cite{Postnikov}]\label{dfn:positriodboundary} 
Let $\cM = ([n], \cB)$ and $\cM' = ([n], \cB')$ be two positroids. The positroid cell corresponding to $\cM'$ is a \emph{boundary cell} of the positroid cell corresponding to $\cM$ (in $\Gr(\rk(\cM), n)$) if and only if  $\cB' \subseteq \cB$.  \end{dfn}

In a slight abuse of notation, with $\cM$ and $\cM'$ as above, we call $\cM'$ a \emph{boundary} of $\cM$.

There are many equivalent combinatorial structures that are in bijection with positroids. In this section, we define one, \emph{Grassmann necklaces}, which we will make considerable use of. The notions of Le diagrams and reduced word subword pairs will briefly appear in Section \ref{sec:diagrammaticmoves}, but we will not need their definitions, which the interested reader can find in \cite{Juggling,Postnikov}.

\begin{dfn}[Definition 4 \cite{Oh43}]\label{dfn:Grassmannnecklace}   A \emph{Grassmann necklace}, $\cI$ is a sequence of $n$ sets of size $k$ \bas (I_1, \ldots , I_n )\eas with elements in $[n]$ such that either 

\begin{itemize} 
    \item $i \in I_i$ and there is some $j \in [n]$ such that $I_{i+1} = (I_i \setminus i) \cup j$ or 
    \item  $i \not \in I_i$ $I_i  = I_{i+1}$ otherwise. 
\end{itemize} 
\end{dfn}

The set of Grassmann necklaces is in bijection with the set of postitroids. However, one may associate a Grassmann necklace to any matroid via its set of bases, \cite{Oh43, Postnikov}. Specifically, for any matroid $\cM  = ([n], \cB)$, let $\cI(\cB)$ be the lexicographically minimal element of $\cB$ in the $<_i$ linear order on $[n]$. Then $\cI(\cB)$ is a Grassmann necklace. 

To calculate the set of bases of a matroid associated to a Grassmann necklace, we need the definition of a \emph{Gale ordering}. 
\begin{dfn}[Section 2.2 \cite{Oh43}]\label{dfn:Gale order}
	The \emph{Gale ordering} is a partial ordering on sets of size $k$ with elements in the cyclically ordered set $[n]$. Specifically, if $A$ and $B$ are two such sets, we write $A = \{A^{(1)}, A^{(2)}, \ldots , A^{(k)} \}$ and $B = \{B^{(1)}, B^{(2)}, \ldots , B^{(k)} \}$ with respect to the $i^{th}$ ordering, i.e.\ $A^{(j)} <_i A^{(j+1)}$. Then $A \leq_i B$ in the Gale ordering if and only if $A^{(j)} \leq_i B^{(j)}$ for all $j \in [n]$. 
\end{dfn}

Each Grassmann necklace defines a set of bases, using the Gale ordering defined above: \ba  \cB_\cI = \{ B \subseteq [n] : |B|=k \text{ and } I_i \leq_i B \; \forall i \in [n]\}  \;, \label{eq:basisfromGN} \ea and the matroid $([n], \cB_\cI)$ is a positroid \cite{Postnikov}. 

Given a matroid $([n], \cB)$ and with $\cI$ the associated Grassmann necklace, note that $\cB \subseteq \cB_\cI$ is always true. However, $([n],\cB_I)$ is the smallest positroid that contains $\cM=([n],\cB)$ and the matroid $\cM$ is a positroid if and only if $\cB_\cI = \cB$.

We are interested in Grassmann necklaces in this paper because there is an algorithm to pass from Wilson loop diagrams to the Grassmann necklaces associated to the corresponding positroid cell \cite{generalcombinatoricsII}. We generalize this result for generalized Wilson loop diagrams in Algorithm \ref{algo:non-crossgwld}.

For now we note a fact about Grassmann necklaces that will be needed to read the Grassmann necklace from the generalized Wilson loop diagram.

\begin{lem}\label{res:greedyGN}
	Let $I_i \in \cI$ be an element of the Grassmann necklace associated to the matroid $\cM$ on $[n]$. Then (using the notation from Definition~\ref{dfn:Gale order}) $I_i^{(j+1)}$ is the first element of $[n]$ in the $<_i$ order that is independent of the set $\{I_i^{(1)}, \ldots, I_i^{(j)}\}$.
\end{lem}
\begin{proof}
	This follows from the independence exchange axiom of matroids: Definition \ref{dfn:indep exchange matroid}.
	
	Suppose, for contradiction, that there is an $a \in [n]$ such that  $I_i^{(j)} <_i a <_i I_i^{(j+1)}$ and $J = \{I_i^{(1)}, \ldots, I_i^{(j)}, a \}$ is independent. Since $I_i$ is the lexicographically minimal independent set of size $\rk(\cM)$ in the $<_i$ order, $j+1 < \rk(\cM)$. Otherwise, $J$ would be the Grassmann necklace element. Since $J$ is independent, $I_i$ is lexicographically minimal, but $a <_i I_i^{(j+1)}$, so it must be true that $J$ cannot be extended to an independent set of size $\rk(\cM)$. 
 
    However, $|I_i| = \rk(M) > j+1$. Therefore, by the independence exchange property of matroids, there is a $b \in I_i$ such that $J \cup b$ is independent. By repeated extension, one can find elements of $I_i$ that allow $J$ to be extended to an independent set of size $\rk(\cM)$, which yields a contradiction.
\end{proof}

\section{Matroids of generalized Wilson loop diagrams \label{sec:allWLDtomatroid}}

In this section, we use the matroid background presented in Section \ref{sec:background} to associate matroids to generalized Wilson loop diagrams and investigate some of their properties. 

\subsection{Generalized Wilson loop diagrams and their matroids \label{sec:gWLDtomatroid}}

We first consider the case of ordinary Wilson loop diagrams as discussed in Remark~\ref{rmk:ordinary WLD} and previous sources including \cite{Wilsonloop}. The matroid structure for ordinary Wilson loop diagrams has already been established by one of us and other coauthors in \cite{Wilsonloop} using a matrix defined combinatorially from the ordinary Wilson loop diagram. This is the matrix that carries the physical information of the particle interaction \cite{Arkani-Hamed:2013jha}. As a lead in to the generalized Wilson loop diagram case, we will revisit the ordinary case using an alternate approach of transversal matroids.

\begin{dfn} \label{dfn:Wtransveral}
 	For each Wilson loop diagram $W = (\cP, n)$, define the bipartite graph $G^\cP_W$ to be the graph with one set of vertices corresponding to the propagators, the other to the set $[n]$ and an edge in $G^\cP_W$ between the propagators $p$ and the vertex $v$ if and only if $v \in V(p)$. Define the transversal matroid associated to $W$ to be the transversal matroid determined by $G^\cP_W$.
\end{dfn}

Note that the \emph{neighbors} in $G^\cP_W$ of a set $V \subseteq [n]$ are the propagators supported on those vertices, $N(V) = \Prop(V)$ and for a set of propagators $P \subseteq \cP$, the neighbors of $P$ in $G^\cP_W$ are their supporting vertices, $N(P) = V(P)$. 

\begin{eg}\label{eg:wldbipartite}
	For example, the ordinary Wilson loop diagram, $W$. has the following associated bipartite graph, $G^\cP_W$: 
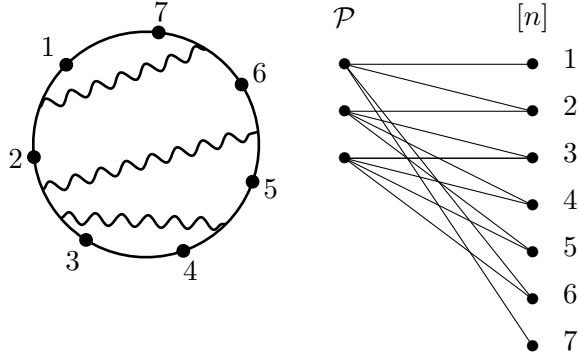
\begin{center} \begin{tikzpicture}[rotate=-45, line width=1, scale=1.5, baseline=(current bounding box.north)]
		\def \n {7}
		\draw circle(1)
        \foreach \v in{1,...,\n}
		{(360*\v/\n-360/\n+180:1)circle(.4pt)circle(.8pt)circle(1.2pt)circle(1.4pt) node[anchor=360/\n*\v-360/\n-45]{$\v$}};
		\draw[decorate,decoration={snake,amplitude=0.8mm}] (-0.901,-0.434) -- (-0.223,0.975);
		\draw[decorate,decoration={snake,amplitude=0.8mm}] (-0.365,-0.931) -- (0.623,0.782);
		\draw[decorate,decoration={snake,amplitude=0.8mm}] (-0.075,-0.997) -- (1.00000000000000,0);
	\end{tikzpicture}
\begin{tikzpicture}[scale=1.25, baseline=(current bounding box.north)]
	\labeledvertices{ , , }{1,...,7}{1/1, 1/2, 1/6, 1/7, 2/2, 2/3, 2/4, 2/5, 3/3, 3/3, 3/4, 3/5, 3/6}
	\bipartitelabels{\cP}{[n]}
\end{tikzpicture} 
\end{center}
	
Note that $N(V) = \Prop(V)$ and that $N(P) = V(P)$.
\end{eg}

We see from Hall's matching theorem that the transversal matroid of $W$ is exactly the same matroid as the matroid associated to $W$ in \cite{Wilsonloop}. Specifically, by Hall's matching theorem, a subset $V \subseteq [n]$ is independent if and only if, for all $U \subseteq V$, $|N(U)| \geq |U|$. Similarly, in \cite{Wilsonloop}, the authors show that a subset of vertices $V \subseteq [n]$ is independent if and only if, for all $U \subseteq V$, $|\text{Prop}(U)| \geq |U|$. 

In order to extend this kind of correspondence from ordinary Wilson loop diagrams to generalized Wilson loop diagrams we need to take into account that a propagator may have more or less than two ends and that a propagator may have capacity larger than $1$. For the matroid associated to the diagram this will require that we pass from transversal matroids to Rado matroids. Additionally, the bipartite graph will have the set of propagators replaced by the set of ends. This set of ends will be given an underlying matroid structure that will serve as input to the Rado matroid construction.

In order to define the matroid on the set of ends, we first show that the capacity function defined in Equation \eqref{eq:capacityfunction} is a rank function on $\cE$. 

\begin{lem} \label{res:capacityisrank}
	The capacity function defined in Equation \eqref{eq:capacityfunction} gives a rank function on $\cE$.
\end{lem}
\begin{proof}
	We check the three requirements of a rank function of a matroid 
 (see Definition~\ref{dfn:rank function}).
 
	If $E \subset \cE$ is empty, then $c(E) = |E| = 0$. Otherwise, since $c(p)$ is positive, $c(E)$ is positive.
 
	For the second point, the definition of $c(E)$ shows that for any $x \in \cE \setminus E$, $c(E) \leq c(E \cup x) \leq c(E)+ 1$.
	
    For the third point, since $c(E)$ is additive in $p$ (i.e.\ $c(E) = \sum_{p \in \cP}c(E\cap \cE(p))$), we only need to prove the inequality \bas c (E \cup F) + c(E \cap F) \leq c(E) + c(F) \eas for single propagators. So let $p$ be a propagator and assume $E, F\subset \cE(p)$. Then \bas c (E \cup F) + c(E \cap F )  = 
	\min\{|E \cup F|, c(p)\} + \min\{|E \cap F|, c(p)\} \eas 
    Without loss of generality take $|E|\leq |F|$ so we have $|E\cap F| \leq |E|\leq |F|\leq |E\cup F|$. Considering each of the five possible positions for $c(p)$ in that last chain of inequalities one can directly compute in each case (using $|E|+|F|=|E\cup F|+|E\cap F|$ as needed) that \bas \min\{|E \cup F|, c(p)\} + \min\{|E \cap F|, c(p)\} \leq \min\{|E|, c(p)\} + \min\{|F|, c(p)\} \eas
    as desired.
\end{proof}

In this manner, we may define a matroid structure on the set of ends of a generalized Wilson loop diagram.  We call this matroid $\sfM_\cE$.

\begin{cor}\label{res:endmatroid}
	The matroid defined by the rank function $c(E)$ on the ground set $\cE$ is a direct sum of the uniform matroids $U_{|\cE(p)|}^{c(p)}$, where $U_n^r$ denotes the uniform matroid of rank $r$ on a set of size $n$, as in Definition \ref{dfn:uniformmatroid}.
\end{cor}

\begin{proof}
	This follows from the fact that the capacity function is additive on the propagators, and that for any $E \subseteq \cE(P)$, $c(E) = \min\{|E  \cap \cE(p)|, c(p)\}$.
\end{proof}

Note that in this definition, we may write $\sfM_\cE = \oplus_{p \in \cP} U_{|\cE(p)|}^{c(p)}$. Each set of propagators $P \subseteq \cP$ is a flat of rank $c(P)$ in $\sfM_\cE$. 

\begin{dfn}\label{dfn:gWLDbipartite}
Let $W = (\cP, [n])$ be a generalized Wilson loop diagram. Define the bipartite graph $G^\cE_W$ as follows. One set of vertices is $\cE$ while the other is $[n]$. There is an edge in $G^\cE_W$ between the end $a \in \cE$ and the vertex $v\in [n]$ if and only if $v \in V(a)$. 
For each $p \in \cP$, the set $\cE(p)$ is given the uniform matroid structure of rank $c(p)$. The set $\cE$ is given the matroid structure $\sfM_\cE$, namely the matroid structure of a direct sum of the uniform matroids defined by each propagator.
\end{dfn}

\begin{eg}\label{eg:workingegbipartite}
	For the generalized Wilson loop diagram in Example \ref{eg:working example}, the associated bipartite graph $G^\cE_W$ is 
 
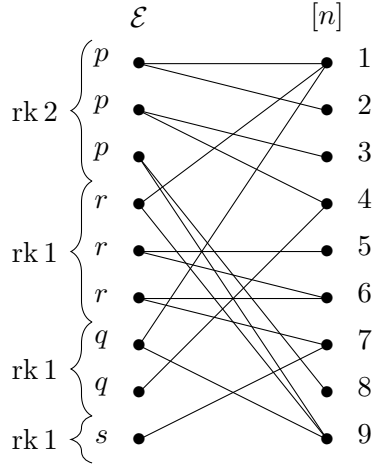
\begin{center}
	\begin{tikzpicture}[scale=1.25]
		\labeledvertices{ p, p,p, r, r, r, q,q, s }{1,...,9}{1/1, 1/2, 2/3, 2/4, 3/8, 3/9, 4/1, 4/9, 5/5, 5/6, 6/6,6/7,7/1, 7/9, 8/4, 9/7 }
		\bipartitelabels{\cE}{[n]}
		\draw [decorate,decoration={brace,amplitude=8pt, mirror},xshift=-0.5cm,yshift=0pt]
		(-1,-.25) -- (-1,-1.75) node [midway,right,xshift=-1.2cm] {$\rk 2$};
		\draw [decorate,decoration={brace,amplitude=8pt, mirror},xshift=-0.5cm,yshift=0pt]
		(-1,-1.75) -- (-1,-3.25) node [midway,right,xshift=-1.2cm] {$\rk 1$};
		\draw [decorate,decoration={brace,amplitude=8pt, mirror},xshift=-0.5cm,yshift=0pt]
		(-1,-3.25) -- (-1,-4.25) node [midway,right,xshift=-1.2cm] {$\rk 1$};
		\draw [decorate,decoration={brace,amplitude=8pt, mirror},xshift=-0.5cm,yshift=0pt]
		(-1,-4.25) -- (-1,-4.75) node [midway,right,xshift=-1.2cm] {$\rk 1$};
	\end{tikzpicture}
 \end{center}

\end{eg}

As above, the neighbors of an end $a \in \cE$ are exactly its support $N(a) = V(a)$, and the neighbors of a vertex $v \in [n]$ are exactly the ends that it supports, $N(v) = \End(v)$.

We can now define the matroid of a generalized Wilson loop diagram.

\begin{dfn}\label{dfn:RadogWLD} Let $W=(\cP, [n])$ be a generalized Wilson loop diagram. The \emph{matroid of} $W$, denoted $\cM(W)$, is the Rado matroid of $G_W^\cE$ and $\sfM_\cE$. Namely, a subset $V \subseteq [n]$ is independent if and only if, \ba  \forall \; U \subseteq V,\quad c(\End(U))  \geq |U|\;. \label{eq:Radoindep}\ea \end{dfn}

As a concrete example of this matroid structure, we revisit the diagrams with Example \ref{eg:same support and cap}, and show that each diagram corresponds to a different matroid.

\begin{eg}\label{eg:same support and cap matroid}
	The three diagrams in Example \ref{eg:same support and cap} correspond to the following bipartite graphs:
		
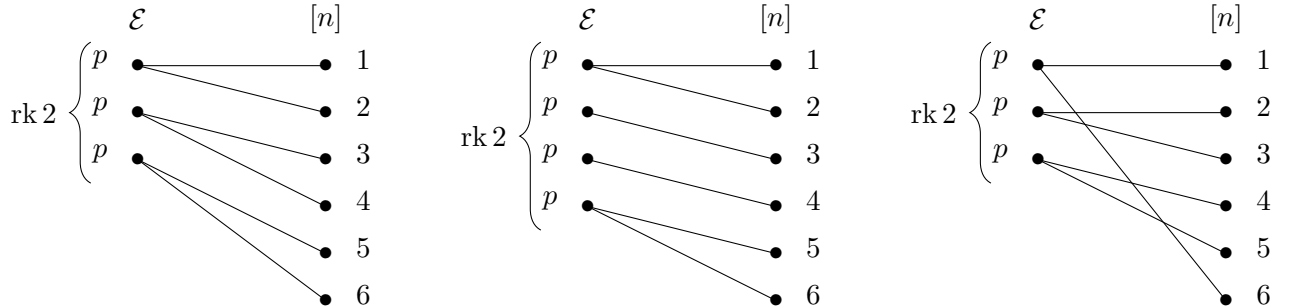
\begin{tikzpicture}[scale=1.25]
		\labeledvertices{p, p,p }{1,...,6}{1/1, 1/2, 2/3, 2/4, 3/5, 3/6}
		\bipartitelabels{\cE}{[n]}
		\draw [decorate,decoration={brace,amplitude=8pt, mirror},xshift=-0.5cm,yshift=0pt]
		(-1,-.25) -- (-1,-1.75) node [midway,right,xshift=-1.2cm] {$\rk 2$};
	\end{tikzpicture} $\qquad$
	\begin{tikzpicture}[scale=1.25]
		\labeledvertices{p, p,p,p }{1,...,6}{1/1, 1/2, 2/3, 3/4, 4/5, 4/6}
		\bipartitelabels{\cE}{[n]}
		\draw [decorate,decoration={brace,amplitude=8pt, mirror},xshift=-0.5cm,yshift=0pt]
		(-1,-.25) -- (-1,-2.25) node [midway,right,xshift=-1.2cm] {$\rk 2$};
\end{tikzpicture} $\qquad$
	\begin{tikzpicture}[scale=1.25]
		\labeledvertices{p, p,p }{1,...,6}{1/1, 1/6, 2/2, 2/3, 3/4, 3/5}
		\bipartitelabels{\cE}{[n]}
		\draw [decorate,decoration={brace,amplitude=8pt, mirror},xshift=-0.5cm,yshift=0pt]
		(-1,-.25) -- (-1,-1.75) node [midway,right,xshift=-1.2cm] {$\rk 2$};\end{tikzpicture}
	
Notice that in the first graph, the vertex sets $\{1, 2\}, \{3,4\}$ and $\{5,6\}$ are dependent. In the second the sets $\{1, 2\}, $ and $\{5,6\}$ are dependent while $\{3,4\}$ is independent. In the last, the sets $\{1, 2\}, \{3,4\}$ and $\{5,6\}$ are independent, but $\{1, 6\}, \{2,3\}$ and $\{4,5\}$ are dependent.
\end{eg}

\begin{rmk}\label{rmk:matroids on props M(W)}
    If we were to further generalize the generalized Wilson loop diagrams to have a matroid associated to each propagator as in Remark~\ref{rmk:matroids on props def}, then $\sfM_\cE$ would be the direct sum of these matroids instead of being the direct sum of uniform matroids determined by the capacity. $\cM(W)$ would then be defined exactly as above, using this more complicated $\sfM_\cE$. 
\end{rmk}

We next show that in the case where $W$ is an ordinary Wilson loop diagram, the Rado matroid associated to $G^\cE_W$ and $\sfM_\cE$ is the same as the transversal matroid associated to $G^\cP_W$. This is not a triviality since while transversal matroids are a special case of Rado matroids, here the underlying set has changed from the set of propagators to the set of ends, but the special nature of propagators of capacity 1 provides the link between them.

\begin{lem}\label{res:bothmatroidssame}
	Let $W$ be a generalized Wilson loop diagram where all propagators have capacity $1$. Then the condition that a subset $V \subseteq [n]$ is independent in $\cM(W)$ if and only if, \bas  \forall \; U \subseteq V,\quad c(\End(U))  \geq |U|\eas is equivalent to the condition that a subset $V \subseteq [n]$ is independent in $\cM(W)$ if and only if \bas  \forall \; U \subseteq V,\quad c(\Prop(U))  \geq |U|\;.\eas
\end{lem}
\begin{proof}
	The key to this proof is the fact that all propagators have capacity $1$. Therefore, for any $E \subseteq \cE$, \bas \sum_{p \in \Prop(E)} c(E \cap \cE(p)) = \sum_{p \in \Prop(E)} \min\{|E \cap \cE(p)|, c(p)\} = \sum_{p \in \Prop(E)} 1 \ = |\Prop(E)| \;.\eas Setting $E = \End(U)$ completes the proof.
\end{proof}

In particular, in the case that $W$ is an ordinary Wilson loop diagram, then the Rado matroid of $G^\cE_W$ and $\sfM_\cE$ defines the matroid of the ordinary Wilson loop diagram as given previously. 

\begin{cor}
\label{res:matroidWLDandgWLD}
	For $W$ an ordinary Wilson loop diagram, the transversal matroid of $W$ (Definition~\ref{dfn:Wtransveral}) is the same as the matroid $\cM(W)$ (Definition~\ref{dfn:RadogWLD}).
\end{cor}

\subsection{Properties of matroids associated to generalized Wilson loop diagrams\label{sec:prop}}

We want to impose some ``niceness'' conditions for generalized Wilson loop diagrams in order to more easily understand the independence structure of the associated matroid. In particular, we give a condition for when the rank of $\cM(W)$ is what would naively be expected from the capacities and a definition for when the support of end of propagators can be reduced without changing the underlying matroid structure. Finally, we give some examples of when different generalized Wilson loop diagrams give the same matroid. The resulting understanding will be useful when we move to considering which generalized Wilson loop diagrams give positroids.

\subsubsection{Minimal and capacity-ranked generalized Wilson loop diagrams \label{sec:capacityrank}} 

A minimal representation of a transversal matroid is one where the matrix has a minimal number of entries or equivalently the bipartite graph has a minimal number of edges \cite{basisshapeloci}.  This notion can then be inherited to obtain a notion of minimality for ordinary Wilson loop diagrams  for which there is a diagrammatic characterization, see Theorem 3.2 of \cite{basisshapeloci}.
Passing to the Rado case, we take the bipartite graph perspective as definitional and so  define a generalized Wilson loop diagram $W$ to be minimal if among all generalized Wilson loop diagrams with the same associated matroid, $W$ has the minimal number of edges in its bipartite graph.  We do not have a diagrammatic characterization of minimality for generalized Wilson loop diagrams, so we will instead work with the weaker condition that no edge of the bipartite graph can be removed without changing the matroid.  This we call \emph{local minimality}.

\begin{dfn} \label{dfn:local-min}   We say that any generalized Wilson loop diagram, $W$ or associated graph $G^\cE_W$ is \emph{locally minimal} if removing one of more edges of $G^\cE_W$ changes the matroid $\cM(W)$. \end{dfn} 

Later in Section \ref{sec:diagrammaticmoves} when we work with diagrammatic moves, we will define our moves to avoid causing non-minimal configurations among the active ends involved in the move, though this will not completely prevent non-minimal diagrams from appearing after our moves as non-minimalities can be caused elsewhere in the diagram. 

Diagrammatically, the situation in Definition~\ref{dfn:local-min} can take a few forms. If an end corresponds to an edge in $\{e_1, \ldots, e_n\}$ then its support has size 2 so it contributes two edges to the bipartite graph. Removing one of these edges corresponds to moving the end from this edge to one of the incident vertices. Generally this will change the matroid, but for particular configurations of other propagators and ends, it may leave the matroid unchanged, see Example~\ref{eg:min1}, Example \ref{eg: half props clear}, and Example \ref{eg:no multi-ends}. In these cases we say the original generalized Wilson loop diagram was \emph{not minimal}, as they do not satisfy Definition \ref{dfn:local-min}.

Removing all the edges attached to an end in $\cE$ in effect removes the end entirely from the diagram -- the vertex associated to this end in $G^\cE_W$ has become an isolated vertex, and thus never contributes to the rank function for $\cM(W)$. Consequently we can equally well take  the bipartite graph with the isolated vertices in $\cE$ removed, without any change to the matroid.  In view of this, whether we consider the bipartite graph with or without these vertices removed will be a matter of what is more convenient in a given situation.

A further special case of non-minimality for a generalized Wilson loop diagram $W=(\cP, [n])$ is the case where \emph{decreasing the capacity of a propagator preserves the matroid} $\cM(W)$. In the special case where we decrease the capacity all the way to 0 this is another way to look at removing a propagator. For an example of this, consider the the discussion in Example \ref{eg:contraction} in Section \ref{sec:contractions} on contractions of these matroids.

\begin{eg}\label{eg:min1}
    Let $W=(\{p,q\}, [5])$ with $p$ the propagator with only end $e_1$ and with $q$ the propagator with ends $e_1$ and $v_4$, as illustrated below. 
    
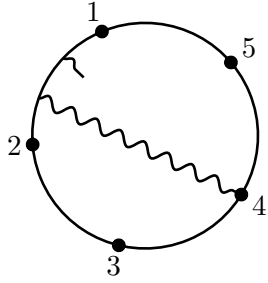
\begin{align*}
    \begin{tikzpicture}[rotate=-67.5, line width=1, scale=1.5]
    	\def \n {5}
    	\draw circle(1)
    	\foreach \v in{1,...,\n}
    	{(360*\v/\n-360/\n+180:1)circle(.4pt)circle(.8pt)circle(1.2pt)circle(1.4pt) node[anchor=360/\n*\v-360/\n-67.5]{$\v$}};
    	\draw[decorate,decoration={snake,amplitude=0.8mm}] (-0.914,-0.407) -- ++({atan(0.445)}:.25 cm);
    	\draw[decorate,decoration={snake,amplitude=0.8mm}] (-0.669,-0.743) -- (0.809,0.588);
    \end{tikzpicture}
\end{align*}
     The bases of $\cM(W)$ are $\{\{1,2\}, \{1,4\}, \{2,4\} \}$. Removing the edge between the end $e_1$ of $q$ and vertex $1$ in $G^\cE_W$ is equivalent to changing that end of $q$ from $e_1$ to $v_2$, as illustrated below. 
     
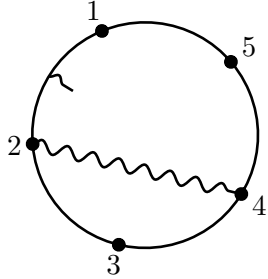
\begin{align*}
     \begin{tikzpicture}[rotate=-67.5, line width=1, scale=1.5]
     	\def \n {5}
     	\draw circle(1)
     	\foreach \v in{1,...,\n}
     	{(360*\v/\n-360/\n+180:1)circle(.4pt)circle(.8pt)circle(1.2pt)circle(1.4pt) node[anchor=360/\n*\v-360/\n-67.5]{$\v$}};
     	\draw[decorate,decoration={snake,amplitude=0.8mm}] (-0.809,-0.588) -- ++({atan(0.727)}:.25 cm);
     	\draw[decorate,decoration={snake,amplitude=0.8mm}] (-0.309,-0.951) -- (0.809,0.588);
     \end{tikzpicture}
 \end{align*}
     The bases remain the same and so the matroid remains unchanged.
\end{eg}

\begin{lem}\label{res:non min and co-loops}  
	Let $W$ be a locally minimal (Definition~\ref{dfn:local-min}) generalized Wilson loop diagram. If $x$ is a co-loop of $\cM(W)$, there is a unique end $a$ supported only on the vertex $x$: $\End(x) = a$ and $V(a) = x$.  Furthermore, if $a$ is not a single ended propagator, then the propagatar that $a$ is an end of has capacity greater than $1$ and there is another diagram corresponding to the same matroid formed by replacing the propagator containing $a$ with two propagators: a single-ended propagator with end $a$ and the rest of propagator which originally contained $a$ with capacity decreased by 1.
\end{lem} 
\begin{proof}
	To see this, first note that if we remove one or more edges of $G^{\cE}_W$ to obtain some $G'$ corresponding to a generalized Wilson loop diagram $W'$, then any independent matching of $G'$ is still an independent matching of $G^{\cE}_W$, so either $\cM(W')$ has strictly lower rank than $\cM(W)$ or the set of bases of $\cM(W')$ is a subset of the set of bases of $\cM(W)$.
	
	Now suppose $W$ is locally minimal (Definition~\ref{dfn:local-min}) and $x$ is a co-loop of $\cM(W)$. Let $a$ be an end supported on $x$ and let $e$ be the edge between $x$ and $a$ in $G^{\cE}_W$. If there is a basis of $\cM(W)$ that has a corresponding independent matching of $G^{\cE}_W$ that does not use $e$, then this would be an independent matching of $G'$ as well, where $G'$ is the bipartite graph without $e$, and hence $\rk(\cM(W))=\rk(\cM(W'))$. Consequently $x$ is in every basis of $\cM(W')$ and so remains a co-loop in $\cM(W')$. Take any basis $B$ of $\cM(W)$ and let $I=B\setminus x$. The neighborhood of $I$ is unchanged between $G^{\cE}_W$ and $G'$ so $I$ remains independent in $\cM(W')$, but adding a co-loop to an independent set gives an independent set, so $\cM(W')$ and $\cM(W)$ have the same bases and hence are the same matroid. This contradicts local minimality condition on $W$, so $e$ must be in every maximum independent matching of $G^{\cE}_W$. Then local minimality condition also implies that $a$ cannot be supported on any other vertex and that no other end can be supported on $x$ since the corresponding edges in the bipartite graph can never be used in any matching, as originally claimed.

    For the proof the final sentence of the statement, note that since $a$ is always matched, the described change can be made without changing the possible independent matchings. Note that, by local minimality, the propagator that $a$ belongs to must have capacity greater than one. Otherwise one could remove the edges in $G_W^\cE$ connecting the other ends to their supporting vertices. 
\end{proof}

A second key notion of niceness for a generalized Wilson loop diagram $W$ is when the rank of the matroid $\mathcal{M}(W)$ is the same as the sum of the capacities of the propagators, as captured in the following definition.
\begin{dfn} \label{dfn:capranked}
	A generalized Wilson loop diagram is \emph{capacity-ranked} if and only if  $\rk\left(\cM(W)\right) = c(\cP)$.
\end{dfn}

Next, we give an example of a diagram that is not capacity-ranked. 

\begin{eg}\label{eg:not-capranked}
	Suppose $W$ is a generalized Wilson loop diagram that has single-ended propagators, $p$, $q$ and $r$, supported on the vertices $v$, the edge defined by $v$ and $v+1$ and the vertex $v+1$ respectively. Since these are single-ended propagators, we equate the propagator with its unique end. Suppose for simplicity that $\Prop(\{v, v+1\}) = \End(\{v, v+1\}) =\{p, q, r\}$. 
	\begin{center}
	\begin{tikzpicture}[rotate=-67.5, line width=1, scale=1.25]
		\def \n {6}
		\draw circle(1)
		\foreach \v in{1,...,\n}
		{(360*\v/\n-360/\n+180:1) node[anchor=360/\n*\v-360/\n]{}};
		\draw (360*2/\n-360/\n+180:1)circle(.4pt)circle(.8pt)circle(1.2pt)circle(1.4pt) node[anchor=360/\n*2-360/\n-67.5]{$v$};
		\draw (360*3/\n-360/\n+180:1)circle(.4pt)circle(.8pt)circle(1.2pt)circle(1.4pt) node[anchor=360/\n*3-360/\n-67.5]{$v+1$};
		\draw[decorate,decoration={snake,amplitude=0.8mm}] (-0.500,-0.866) -- ++({atan(1.73)}:.25 cm);
		\draw[decorate,decoration={snake,amplitude=0.8mm}] (0,-1.00000000000000) -- ++({atan(3.14)}:.25 cm);
		\draw[decorate,decoration={snake,amplitude=0.8mm}] (0.500,-0.866) -- ++({atan(-1.73)}:-.25 cm);
	\end{tikzpicture}
 \end{center}
 
    The only basis of $\cM(W)$ is $\{v,v+1\}$ so $\rk (\cM(W))=2$. But $c(\cP)=3$ so $\cM(W)$ is not capacity-ranked and removing any one of $p$, $q$ or $r$ in the diagram does not change $\cM(W)$.
	
\end{eg}

Note that we can have diagrams that are capacity-ranked but not minimal such as Example \ref{eg:min1}. 

In the remainder of this paper, we frequently assume diagrams are capacity-ranked, however we will not make use of an explicit condition for capacity-rankedness. We can give such a condition at the cost of a digression into an auxiliary matroid, and while it is not necessary for the subsequent work, the interested reader can find these results in Lemma \ref{res:capranked} which gives a concrete condition for what it means for a set of propagators to be independent in the auxiliary matroid and Proposition \ref{res:propindepcorrect} which shows that this explicit condition is enough to check that the entire diagram is capacity-ranked. The reader uninterested in this condition can skip directly to the discussion of multiple representations beginning at Section~\ref{sec:manyreps}.

\begin{dfn}\label{dfn:propindep}
For a generalized Wilson loop diagram $W = (\cP, [n])$, let $\widetilde{\cP}$ be the multiset of propagators where each propagator $p$ appears with multiplicity $c(p)$. Define the bipartite graph $G_W^{\widetilde\cP}$ with vertex bipartition $\widetilde{\cP} \sqcup [n]$ and with an edge from any copy of the propagator $p$ to a vertex $v \in [n]$ if and only if $v \in V(p)$. Give $[n]$ the matroid structure of $\cM(W)$. Define $\widetilde M(W)$ to be the Rado matroid associated to $G_W^{\widetilde\cP}$ and $\cM(W)$.
\end{dfn}

Note that we've swapped which half of the bipartite graph plays which role in the Rado matroid between the definition of $\cM(W)$ and $\widetilde{M}(W)$.

Given a set of propagators $P\subseteq \cP$, we will write $\widetilde{P}$ for the subset of $\widetilde \cP$ consisting of all copies of the propagators in $p$. We will abuse notation lightly by saying that $P$ is \emph{independent in $\widetilde{M}(W)$} whenever $\widetilde{P}$ is an independent set in $\widetilde{M}(W)$.

\begin{lem}\label{res:capranked}
	For a generalized Wilson loop diagram $W = (\cP, [n])$, a propagator set $P \subseteq \cP$ is independent in $\widetilde{M}(W)$ if and only if 
	\bas \sum_{p \in P} c(p) \leq \rk_{\cM(W)}\left(V(P)\right) \;.\eas
\end{lem}

\begin{proof}
	Suppose that $P \subset \cP$ is not independent in $\widetilde M(W)$. 
 Note that in $G^{\widetilde \cP}_W$, the neighbors of a set of propagators is the set of vertices that support those propagators $N(P) = V(P)$.
	
	Then by definition of a Rado matroid, Definition \ref{dfn:Rado}, $P \subset \cP$ is not an independent set in $\widetilde M(W)$ if and only if there is a subset $\widetilde{Q} \subseteq \widetilde P$ such that $|\widetilde Q| > \rk_{\cM(W)}\left(V(P)\right)$. For all $p \in \widetilde Q$, write $m(p)$ to indicate the multiplicity of $p$ in $\widetilde Q$, with zero indicating non-inclusion. In other words, writing $|\widetilde Q| = \sum_{p \in P} m(p)$, the set $P \subset \cP$ is independent if and only if for any set of multiplicities $0 \leq m(p) \leq c(p)$, \bas \sum_{p \in P} c(p) \leq \rk_{\cM(W)}\left (V(P)\right) \;. \eas Therefore, we must always have that $\sum_{p \in P} c(p) \leq \rk_{\cM(W)}\left (V(P)\right )$.
\end{proof}

\begin{prop}\label{res:propindepcorrect} Let $W = (\cP, [n])$ be a generalized Wilson loop diagram. $W$ is capacity-ranked if and only if all $P\subseteq \cP$ are independent in $\widetilde{M}(W)$.
\end{prop}
\begin{proof} 
	First observe that $\rk(\cM(W)) \leq c(\cP)$. To see this note that by definition of a Rado matroid, Definition \ref{dfn:Rado}, $V\subseteq [n]$ is an independent in $\cM(W)$ if and only if $\forall U \subseteq V$, $|U| \leq c(\End(U))$, in which case $\rk(V) = |V|$. The largest possible value of $c(\End(V))$ is $c(\cP)$ so $\rk(\cM(W)) \leq c(\cP)$.
	
	Suppose that $P\subseteq \cP$ is not independent in $\widetilde M(W)$. Then, there is some subset $\widetilde Q \subseteq \widetilde P$ where each $p \in P$ has multiplicity $0 \leq m(p) \leq c(p)$ such that $\rk_{\cM(W)}(V(P)) < \sum_{p\in P} m(p) \leq c(P)$. 

	Take a subset of $V(P)$ which is independent and of maximal rank in $\cM(W)$ and extend it to a basis $V$ of $\cM(W)$. Suppose now for a contradiction that $V$ has size $c(\cP)$. By construction $|V \cap V(P)| = \rk_{\cM(W)}(V(P))$ and so using the previous paragraph \bas |V \setminus V(P)| \geq c(\cP) - \rk_{\cM(W)}(V(P)) > c(\cP) - c(P) \;.\eas  However, the set of propagators supported by $V \setminus V(P)$ is a subset of the complement of $P$ therefore, $|V \setminus V(P)| > c(\cP) - c(P) \geq c(\Prop(V \setminus V(P))) \geq c( \End(V \setminus V(P)))$. This contradicts $V$ being independent in $\cM(W)$. Therefore, the basis $V$ of $\cM(W)$ cannot be of size $c(\cP)$ and hence $\rk (\cM(W)) <c(\cP)$. 
	
	In the other direction, if $\rk (\cM(W)) < c(\cP)$ then $\rk_{\cM(W)}(V(\cP)) < c(\cP)$. Then, by Lemma \ref{res:capranked}, $\cP$ is not independent in $\widetilde M(W)$.
	
\end{proof}

Using the two previous results we see that $W=(\cP, [n])$ is not capacity-ranked if and only if some $P\subseteq \cP$ is dependent in $\widetilde{M}(W)$ which happens if and only if there are some $0\leq m(p)\leq c(p)$ with $\rk_{\cM(W)}(V(P)) < \sum_{p\in P}m(p)$. Since capacity is additive on propagators this last sum is less than or equal to $c(P)$. Conversely if there is some $P\subseteq \cP$ with $\rk_{\cM(W)}(V(P)) < c(P)$ then $m(p)$ can be selected so as to witness that $P$ is dependent in $\widetilde{M}(W)$ and hence $W$ is not capacity-ranked.

\subsubsection{Generalized Wilson loop diagrams representing the same matroid \label{sec:manyreps}}
While there is a unique Rado matroid associated to each generalized Wilson loop diagram, there is not a unique generalized Wilson loop diagram associated to each matroid of the form $\cM(W)$. There are several examples of this phenomenon that deserve special attention. We address these in this section, and establish some conventions for the diagrams and calculations considered in this paper. We have also seen some examples in Section \ref{sec:capacityrank} in the discussion of minimality. Non-uniqueness of representations for ordinary Wilson loop diagrams was studied in \cite{Wilsonloop, generalcombinatoricsI} and is exactly characterized in the admissible case.

The following example shows that if a generalized Wilson loop diagram has a propagator with the same number of ends as its capacity, it can be replaced by the same number of single-ended propagators.

\begin{eg} \label{eg:cap equal end num}
Consider the situation when a generalized Wilson loop diagram $W$ has a propagator $p$ with capacity equal to its number of ends. Then the diagram formed by replacing $p$ with the set of single-ended propagators corresponding to the ends of $p$, each necessarily with capacity 1, gives rise to the same matroid. This is because the propagator $p$ corresponds to the submatroid $U_{c(p)}^{c(p)}$ of $\cE$. Since $U_{c(p)}^{c(p)} = \oplus_{i= 1}^{c(p)} U_{1}^{1}$, we may replace $p$ with $c(p)$ single-ended propagators with the same ends. For instance, in the diagram below, since a uniform matroid of rank $4$ on $4$ elements is the same as the sum of four rank 1 uniform matroids on one element each, the two generalized Wilson loop diagrams on the left and right correspond to the same matroid. 

	
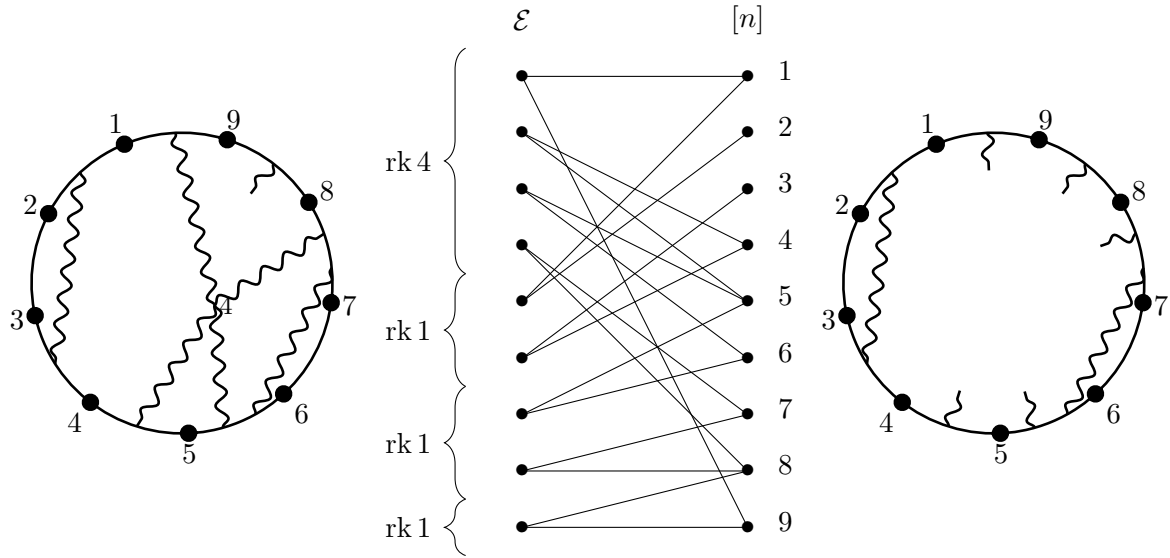
\begin{tikzpicture}[rotate=-67.5, line width=1, scale=2,baseline=0] 
		\def \n {9}
		\draw circle(1)
		\foreach \v in{1,...,\n}
		{(360*\v/\n-360/\n+180:1)circle(.4pt)circle(.8pt)circle(1.2pt)circle(1.4pt) node[anchor=360/\n*\v-360/\n-67.5]{$\v$}};
		\draw[decorate,decoration={snake,amplitude=0.8mm}] (-0.940,-0.342) -- (0.174,-0.985);
		\foreach \x/\y in {0.766/-0.643,0.993/-0.116,0.058/0.998,-0.940/0.342}{\draw[decorate,decoration={snake,amplitude=0.8mm}] (0.219,0.145) -- (\x,\y);}
		\draw (0.219,0.145) node[shift = {(.15,0)}] {\small4};
		\draw[decorate,decoration={snake,amplitude=0.8mm}] (0.993,0.116) -- (0.287,0.958);
		\draw[decorate,decoration={snake,amplitude=0.8mm}] (-0.500,0.866) -- ++({atan(-1.73)}:.25 cm);
	\end{tikzpicture} 
	\begin{tikzpicture}[scale=1.5,baseline=-3.5cm]
		\labeledvertices{ , ,,, , , ,,  }{1,...,9}{1/1, 1/9, 2/4, 2/5, 3/5, 3/6, 4/7, 4/8, 5/1, 5/2, 6/3,6/4,7/5, 7/6, 8/7, 8/8,9/8,9/9}
		\bipartitelabels{\cE}{[n]}
		\draw [decorate,decoration={brace,amplitude=8pt, mirror},xshift=-0.5cm,yshift=0pt]
		(-1,-.25) -- (-1,-2.25) node [midway,right,xshift=-1.2cm] {$\rk 4$};
		\draw [decorate,decoration={brace,amplitude=8pt, mirror},xshift=-0.5cm,yshift=0pt]
		(-1,-2.25) -- (-1,-3.25) node [midway,right,xshift=-1.2cm] {$\rk 1$};
		\draw [decorate,decoration={brace,amplitude=8pt, mirror},xshift=-0.5cm,yshift=0pt]
		(-1,-3.25) -- (-1,-4.25) node [midway,right,xshift=-1.2cm] {$\rk 1$};
		\draw [decorate,decoration={brace,amplitude=8pt, mirror},xshift=-0.5cm,yshift=0pt]
		(-1,-4.25) -- (-1,-4.75) node [midway,right,xshift=-1.2cm] {$\rk 1$};
	\end{tikzpicture} 
	\begin{tikzpicture}[rotate=-67.5, line width=1, scale=2,baseline=0]
		\def \n {9}
		\draw circle(1)
		\foreach \v in{1,...,\n}
		{(360*\v/\n-360/\n+180:1)circle(.4pt)circle(.8pt)circle(1.2pt)circle(1.4pt) node[anchor=360/\n*\v-360/\n-67.5]{$\v$}};
		\draw[decorate,decoration={snake,amplitude=0.8mm}] (-0.940,-0.342) -- (0.174,-0.985);
		\draw[decorate,decoration={snake,amplitude=0.8mm}] (0.766,-0.643) -- ++({atan(-0.839)}:-.25 cm);
		\draw[decorate,decoration={snake,amplitude=0.8mm}] (0.993,-0.116) -- ++({atan(-0.117)}:-.25 cm);
		\draw[decorate,decoration={snake,amplitude=0.8mm}] (0.993,0.116) -- (0.287,0.958);
		\draw[decorate,decoration={snake,amplitude=0.8mm}] (0.058,0.998) -- ++({atan(17.2)}:-.25 cm);
		\draw[decorate,decoration={snake,amplitude=0.8mm}] (-0.500,0.866) -- ++({atan(-1.73)}:.25 cm);
		\draw[decorate,decoration={snake,amplitude=0.8mm}] (-0.940,0.342) -- ++({atan(-0.364)}:.25 cm);
	\end{tikzpicture}

\end{eg}

Another useful equivalence that we use for the computations in Section \ref{sec:diagrammaticmoves} concerns single-ended propagators. Namely, if a single-ended propagator ends on an edge, $e_i$, an equivalent diagram has only the single-ended propagator on that edge, while all other ends that used to be on that edge are moved to one of the adjacent vertices. If the single-ended propagator ends on a vertex $v_i$, then all other ends on that vertex may be dropped and all remaining ends supported on that vertex may be pushed to their other support vertex without affecting the matroid structure.  We saw one special case of this in Lemma~\ref{res:non min and co-loops}.

\begin{eg} \label{eg: half props clear}
    Let $W = (\cP, [n])$ be a generalized Wilson loop diagram with single-ended propagators. Let $G_W^\cE$ be the associated bipartite graph.  Note that the single-ended propagators are co-loops in the matroid structure in the matroid $\sfM_\cE$ on the end set $\cE$ in $G_W^\cE$. This induces a special structure on these propagators in the Rado matroid $M(W)$.
    Let $p \in \cP$ be a single-ended propagator.
    
    In $G_W^\cE$, the single-ended propagator $p$ is connected to the set $V(p) \subseteq [n]$, where $V(p)$ is either a single vertex or a pair of adjacent vertices. 
    Define the set of ends that share support with $p$: \bas E_p = \{q \in \cP | V(q) = V(p)\}\;. \eas  Define $G_{W'}^\cE$ to be any graph derived from $G_W^\cE$ by removing one edge incident to $a$ (and an element of $V(p)$) from each end $a$ in $E_p \setminus p$. (Note, this removes all edges incident to $a$ if $|V(p)| = 1$.) 
    We claim that $G_W^\cE$ and every such $G_{W'}^\cE$ define the same matroid.

    Use $N_W$ and $N_{W'}$ for the set of neighbors in $G^\cE_W$ and $G^\cE_{W'}$ respectively. Whether $V(p) = \{i, i+1\}$, or $V(p) = \{i\}$, for any $U \subseteq [n]$ such that $|V(p) \cap U| = 0$ , $N_W(U) = N_{W'}(U)$ and $\rk_{\cE'}(N_W(U)) = \rk_{cE}(N_{W'}(U))$, $U$ is independent in $M(W')$ if and only if $U$ is independent in $M(W')$.
    
    If $|V(p)\cap U| = 2$ we are in the case where $|V(p)| =2$ and $V(p) \subseteq U$. In this case, $N_W(U) = N_{W'}(U)$. Since the rank function on the ends doesn't change, such a $U$ is independent in $M(W')$ if and only if $U$ is independent in $M(W')$. 
    
    Finally if  $|V(p) \cap U| = 1$, write $V(p)\cap U =\{i\}$. Then by the argument for when $|V(p) \cap U| = 0$, $U \setminus i$ is independent in $M(W')$ if and only if it is independent in $M(W')$. Moreover, $p$ is not a neighbor of $U \setminus i$ in either $G_{W'}^\cE$ or $G_{W}^\cE$. Since $p$ is a loop in $\cE$ and $\cE'$, adding $p$ to the neighbors of $U\setminus i$ in either $W$ or $W'$ increases the rank of both sets. Therefore, $U$ is independent in $\cM(W)$ or $\cM(W')$ if and only if $U \setminus i$ is independent in $\cM(W)$ or $\cM(W')$, and $U \setminus i$ is independent in $\cM(W)$ if and only if it is independent in $\cM(W')$. Therefore, $U$ is independent in $M(W')$ if and only if $U$ is independent in $M(W')$. 

\end{eg}

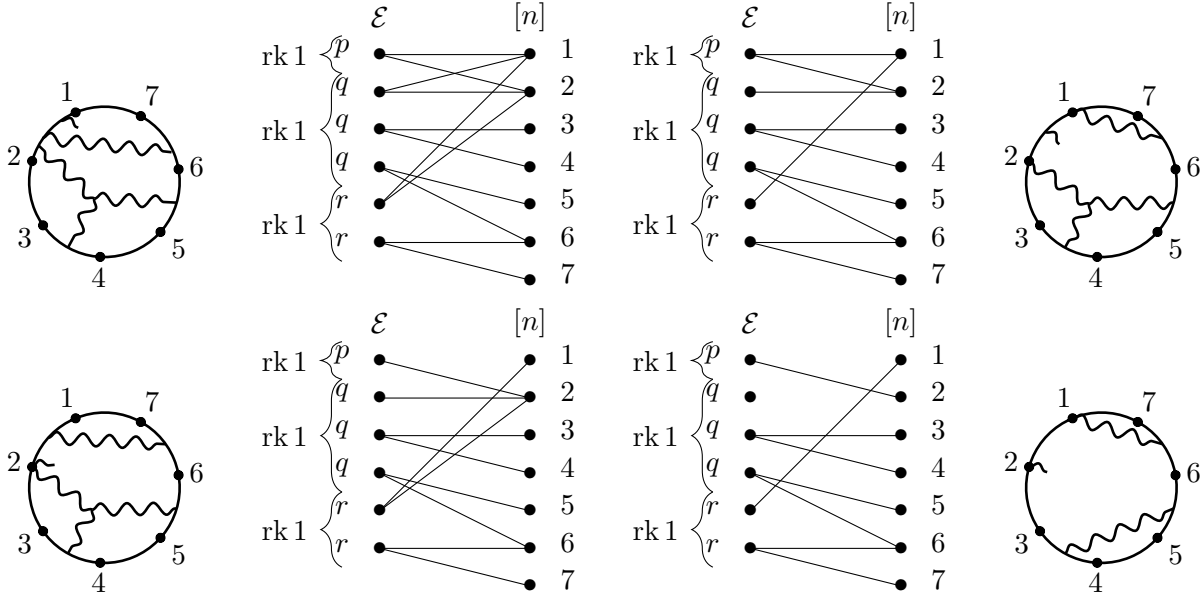
\begin{figure}
\begin{tikzpicture}[rotate=-67.5, line width=1, scale=1]
	\def \n {7}
	\draw circle(1)
	\foreach \v in{1,...,\n}
	{(360*\v/\n-360/\n+180:1)circle(.4pt)circle(.8pt)circle(1.2pt)circle(1.4pt) node[anchor=360/\n*\v-360/\n-67.5]{$\v$}};
	\draw[decorate,decoration={snake,amplitude=0.8mm}] (-0.975,-0.223) -- ++({atan(1)}:.25 cm);
	\foreach \x/\y in {-0.782/-0.623,0.623/-0.782,0.623/0.782}{\draw[decorate,decoration={snake,amplitude=0.8mm}] (0.155,-0.208) -- (\x,\y);}
	\draw[decorate,decoration={snake,amplitude=0.8mm}] (-0.855,-0.545) -- (-0.023,0.975);   
\end{tikzpicture}  \quad 
\begin{tikzpicture}[scale=1]
		\labeledvertices{ p, q, q, q,r, r }{1,...,7}{1/1, 1/2, 2/1, 2/2, 3/3, 3/4, 4/5, 4/6, 5/1, 5/2, 6/6, 6/7}
		\bipartitelabels{\cE}{[n]}
		\draw [decorate,decoration={brace,amplitude=8pt, mirror},xshift=-0.5cm,yshift=0pt]
		(-1,-.25) -- (-1,-.75) node [midway,right,xshift=-1.2cm] {$\rk 1$};
		\draw [decorate,decoration={brace,amplitude=8pt, mirror},xshift=-0.5cm,yshift=0pt]
		(-1,-.75) -- (-1,-2.25) node [midway,right,xshift=-1.2cm] {$\rk 1$};
		\draw [decorate,decoration={brace,amplitude=8pt, mirror},xshift=-0.5cm,yshift=0pt]
		(-1,-2.25) -- (-1,-3.25) node [midway,right,xshift=-1.2cm] {$\rk 1$};
	\end{tikzpicture}  \quad 
\begin{tikzpicture}[scale=1]
		\labeledvertices{ p, q, q, q,r, r }{1,...,7}{1/1, 1/2, 2/2, 3/3, 3/4, 4/5, 4/6, 5/1, 6/6, 6/7}
		\bipartitelabels{\cE}{[n]}
		\draw [decorate,decoration={brace,amplitude=8pt, mirror},xshift=-0.5cm,yshift=0pt]
		(-1,-.25) -- (-1,-.75) node [midway,right,xshift=-1.2cm] {$\rk 1$};
		\draw [decorate,decoration={brace,amplitude=8pt, mirror},xshift=-0.5cm,yshift=0pt]
		(-1,-.75) -- (-1,-2.25) node [midway,right,xshift=-1.2cm] {$\rk 1$};
		\draw [decorate,decoration={brace,amplitude=8pt, mirror},xshift=-0.5cm,yshift=0pt]
		(-1,-2.25) -- (-1,-3.25) node [midway,right,xshift=-1.2cm] {$\rk 1$};
	\end{tikzpicture} \quad
\begin{tikzpicture}[rotate=-67.5, line width=1, scale=1]
 \def \n {7}
 \draw circle(1)
 \foreach \v in{1,...,\n}
 {(360*\v/\n-360/\n+180:1)circle(.4pt)circle(.8pt)circle(1.2pt)circle(1.4pt) node[anchor=360/\n*\v-360/\n-67.5]{$\v$}};
 \draw[decorate,decoration={snake,amplitude=0.8mm}] (-0.901,-0.434) -- ++({atan(0.482)}:.25 cm);
 \foreach \x/\y in {-0.623/-0.782,0.623/-0.782,0.623/0.782}{\draw[decorate,decoration={snake,amplitude=0.8mm}] (0.208,-0.261) -- (\x,\y);}
 \draw[decorate,decoration={snake,amplitude=0.8mm}] (-1.00000000000000,0) -- (-0.223,0.975);
\end{tikzpicture} \\
\begin{tikzpicture}[rotate=-67.5, line width=1, scale=1]
	\def \n {7}
	\draw circle(1)
	\foreach \v in{1,...,\n}
	{(360*\v/\n-360/\n+180:1)circle(.4pt)circle(.8pt)circle(1.2pt)circle(1.4pt) node[anchor=360/\n*\v-360/\n-67.5]{$\v$}};
	\draw[decorate,decoration={snake,amplitude=0.8mm}] (-0.623,-0.782) -- ++({atan(3.26)}:.3 cm);
	\foreach \x/\y in {-0.623/-0.782,0.623/-0.782,0.623/0.782}{\draw[decorate,decoration={snake,amplitude=0.8mm}] (0.208,-0.261) -- (\x,\y);}
	\draw[decorate,decoration={snake,amplitude=0.8mm}] (-0.901,-0.434) -- (-0.223,0.975);   
\end{tikzpicture} \quad 
\begin{tikzpicture}[scale=1]
		\labeledvertices{ p, q, q, q,r, r }{1,...,7}{1/2, 2/2, 3/3, 3/4, 4/5, 4/6, 5/1, 5/2, 6/6, 6/7}
		\bipartitelabels{\cE}{[n]}
		\draw [decorate,decoration={brace,amplitude=8pt, mirror},xshift=-0.5cm,yshift=0pt]
		(-1,-.25) -- (-1,-.75) node [midway,right,xshift=-1.2cm] {$\rk 1$};
		\draw [decorate,decoration={brace,amplitude=8pt, mirror},xshift=-0.5cm,yshift=0pt]
		(-1,-.75) -- (-1,-2.25) node [midway,right,xshift=-1.2cm] {$\rk 1$};
		\draw [decorate,decoration={brace,amplitude=8pt, mirror},xshift=-0.5cm,yshift=0pt]
		(-1,-2.25) -- (-1,-3.25) node [midway,right,xshift=-1.2cm] {$\rk 1$};
	\end{tikzpicture}  \quad 
\begin{tikzpicture}[scale=1]
		\labeledvertices{ p, q, q, q,r, r }{1,...,7}{1/2, 3/3, 3/4, 4/5, 4/6, 5/1, 6/6, 6/7}
		\bipartitelabels{\cE}{[n]}
		\draw [decorate,decoration={brace,amplitude=8pt, mirror},xshift=-0.5cm,yshift=0pt]
		(-1,-.25) -- (-1,-.75) node [midway,right,xshift=-1.2cm] {$\rk 1$};
		\draw [decorate,decoration={brace,amplitude=8pt, mirror},xshift=-0.5cm,yshift=0pt]
		(-1,-.75) -- (-1,-2.25) node [midway,right,xshift=-1.2cm] {$\rk 1$};
		\draw [decorate,decoration={brace,amplitude=8pt, mirror},xshift=-0.5cm,yshift=0pt]
		(-1,-2.25) -- (-1,-3.25) node [midway,right,xshift=-1.2cm] {$\rk 1$};
	\end{tikzpicture} \quad
\begin{tikzpicture}[rotate=-67.5, line width=1, scale=1]
 \def \n {7}
 \draw circle(1)
 \foreach \v in{1,...,\n}
 {(360*\v/\n-360/\n+180:1)circle(.4pt)circle(.8pt)circle(1.2pt)circle(1.4pt) node[anchor=360/\n*\v-360/\n-67.5]{$\v$}};
 \draw[decorate,decoration={snake,amplitude=0.8mm}] (-0.623,-0.782) -- ++({atan(1.26)}:.25 cm);
 \draw[decorate,decoration={snake,amplitude=0.8mm}] (0.623,-0.782) -- (0.623,0.782);     
 \draw[decorate,decoration={snake,amplitude=0.8mm}] (-1.00000000000000,0) -- (-0.223,0.975);
\end{tikzpicture}

\caption{As described in Example \ref{eg: half props clear}, if a single-ended propagator has support $V(p) = \{i, i+1\}$, one may change the support of any ends sharing the support to either be $i$ or $i+1$ without changing the associated Rado matroid. If the single-ended propagator has support $i$, then one may remove $i$ from any end with $i$ in its support without changing the associated Rado matroid.} \label{fig:half props clear}
\end{figure}

\begin{rmk}\label{rmk:drawing conventions} 
	Because of the previous two examples, unless otherwise specified, we use the convention that a propagator with multiple ends always has capacity strictly less than its number of ends and that if a single-ended propagator is on an edge then no other end in $\cE$ has support that has $V(p)$ as a subset.
\end{rmk}

Another useful equivalence of representations considers the case when a propagator is a multiset  of elements of $\{v_1, e_1, \ldots, v_n, e_n\}$ that is not a set. In this case, we may replace multiple ends on $e_i$ with a single end on $v_i$ and another on $v_{i+1}$ and we may replace the multiple ends on $v_i$ with a single end on $v_i$. It is still useful to allow multiple ends on a given edge or vertex in our formalism, partly because our boundary moves of Section \ref{sec:diagrammaticmoves} can generate them and partly because resolving multiple ends on a propagator can result in crossings depending on the surrounding propagators.

\begin{eg} \label{eg:no multi-ends}  Specifically, let $W=(\cP, [n])$ be a capacity-ranked generalized Wilson loop diagram with a propagators $p$ with multiple ends on either an edge $e_i$ or a vertex $v_i$, and $G_W^\cE$ the corresponding bipartite graph. Define $W'$ to be the generalized Wilson loop diagram obtained by replacing propagator $p$ with propagator $p'$ where if $p$ has multiple ends on $v_i$ then $p'$ is $p$ except with only a single end on $v_i$, and if $p$ has multiple ends on $e_i$ then $p'$ is obtained from $p$ by removing the ends on $e_i$ and replacing them with one end on $v_i$ and one end on $v_{i+1}$. The capacity of $p'$ is unchanged from $p$.

The equivalence between $W$ and $W'$ can be seen by considering the matchings on the corresponding bipartite graphs.  At most one end can be matched to a vertex, so if $p$ has multiple ends on $v_i$ then at most one of them is ever used in any matching, so no change to matroid is caused by moving to $p'$.  Likewise at most two ends can be matched to the supporting vertices of an edge, so if $p$ has multiple ends on $e_i$ then at most two of them are ever used in any matching, so no change to the matroid is caused by moving to $p'$.

This also justifies that the capacity of $p'$ can remain unchanged from $p$, since $W$ was capacity-ranked, capacity-many ends must get matched in every maximal independent matching, and so this remains true for $p'$.

\end{eg}

Note that in this situation, if there are other propagators in $W$ that have ends on $e_i$, this move may lead to a diagram with crossing propagators. We will discuss this further in Section \ref{sec:gWLDtopositroid}.

\subsubsection{Restrictions of generalized Wilson loop diagrams}\label{sec:restrictions}

It is diagrammatically very natural to restrict the underlying set $[n]$ of a generalized Wilson loop diagram. Fortunately, this corresponds to a restriction operation on the associated the matroid and to a restriction on one side of the bipartite graph. For $V \subseteq [n]$ denote the induced subgraph on $V\sqcup N(V)$ by $G^\cE_W |_V$. 

Note that restricting an ordinary Wilson loop diagram will often result in generalized Wilson loop diagrams that are not ordinary Wilson loop diagrams since propagators may get restricted to have fewer than two ends or fewer than four supporting vertices.

Additionally, note that in the restrictions, $V$ is still a cyclically ordered set, with its ordering induced from $[n]$. The matroid structure on the vertices $N(V)$ in $G^\cE_W|_V$ is the restricted matroid $(\sf{M}_\cE)|_{N(V)}$. We define capacities on the restricted graph in the following way.

\begin{dfn} \label{dfn:restrictedcap}
	Let $p \in \cP$ be a propagator of a generalized Wilson loop diagram $W = (\cP, [n])$. For $V \subseteq [n]$, define $c_V(p)$ to be the smaller of its capacity and the number of ends the corresponding propagator has in $W|_V$: \bas c_V(p)  = \min \{c(p), |\{\End(V) \cap \cE(p)\}| \} \;.\eas
\end{dfn}

The precise relationship between the Rado matroid of a generalized Wilson loop diagram and its restriction is given as follows:

\begin{lem} \label{res:restrictedmatroid}
For $W = (\cP, [n])$ a generalized Wilson loop diagram, and $V \subset [n]$, the Rado matroid defined by $G^\cE_W|_V$ (with the matroid structure on $V$ given by taking $c_V$ as a rank function) is $\cM(W)|_V$.
\end{lem}	
\begin{proof}
	By definition of the matroid restriction, a set $I$ is independent in $\cM(W)|_V$ if and only if it can be written as $I = J \cap V$ for some independent set $J$ of $\cM(W)$. This is exactly the independence structure of the Rado matroid of $G^\cE_W|_V$ induced by this restriction of the graph.
\end{proof}

We may also define a Wilson loop-diagrammatic restriction to the set $V$, and denote it $W|_V$. If $p \in \cP$ has an end $a \in \cE(p)$, such that $V(a) \cap V =  \emptyset$, the corresponding propagator in $W|_V$ does not contain the end $a$. If $a$ has support intersecting, but not contained in $V$, then $a$ must be an edge $e_i$ and the intersection of its support with $V$ is either $v_i$ or $v_{i+1}$. In this case the corresponding propagator in $W|_V$ will have an end corresponding to $a$, but this end will be $v_i$ or $v_{i+1}$, whichever is in $v$. Note that this means that the propagators of $W|_V$, $\cP|_V$, need not be a subset of the propagators, $\cP$, of $W$. We may also split propagators which newly have number of ends equal to capacity in the manner of Remark~\ref{rmk:drawing conventions}, but this is not necessary. In this way we obtain $\cM(W|_V) = \cM(W)|_V$.

\begin{eg}\label{eg:restriction}
Consider the generalized Wilson loop diagram from Example \ref{eg:working example}. The associated bipartite graph $G_W^\cE$ from Example \ref{eg:workingegbipartite}. In this example, we take the restriction of this diagram to the vertex set $V  = [n] \setminus \{v_4, v_7, v_8\}$. Let $S = \{v_4, v_7, v_8\}$ indicate the set of vertices not in the restriction. The restricted graph, $G^\cE_W|_V$, can be built as follows: remove the vertices $S$ from $[n]$ and all edges adjacent to $S$, then remove any isolated vertices from $\cE$ (see below). 

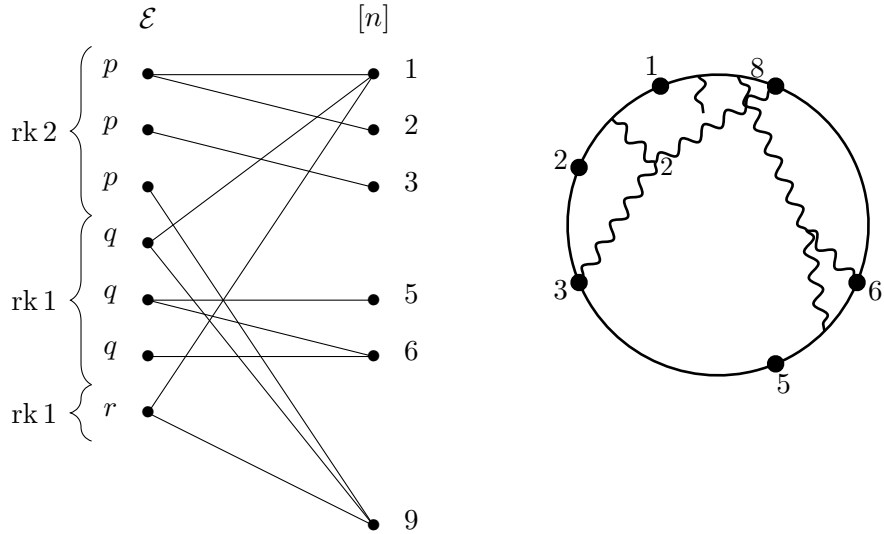
\begin{align*}
\begin{tikzpicture}[scale=1.5,baseline=-2.75cm]
	\labeledvertices{ p, p,p, q, q, q, r }{1,...,3,5,6,9}{1/1, 1/2, 2/3,  3/9,  4/1, 4/9, 5/5, 5/6, 6/6,7/1, 7/9}
	\bipartitelabels{\cE}{[n]} 
	\draw [decorate,decoration={brace,amplitude=8pt, mirror},xshift=-0.5cm,yshift=0pt]
	(-1,-.25) -- (-1,-1.75) node [midway,right,xshift=-1.2cm] {$\rk 2$};
	\draw [decorate,decoration={brace,amplitude=8pt, mirror},xshift=-0.5cm,yshift=0pt]
	(-1,-1.75) -- (-1,-3.25) node [midway,right,xshift=-1.2cm] {$\rk 1$};
	\draw [decorate,decoration={brace,amplitude=8pt, mirror},xshift=-0.5cm,yshift=0pt]
	(-1,-3.25) -- (-1,-3.75) node [midway,right,xshift=-1.2cm] {$\rk 1$};
\end{tikzpicture} 
\qquad & \qquad 
\begin{tikzpicture}[rotate=-67.5, line width=1, scale=2,baseline=0]
	\def \n {8}
	\draw circle(1)
	\foreach \v in{1,2,3,5,6}
	{(360*\v/\n-360/\n+180:1)circle(.4pt)circle(.8pt)circle(1.2pt)circle(1.4pt) node[anchor=360/\n*\v-360/\n-67.5]{$\v$}};
	\draw (360 - 360/8 + 180:1)circle(.4pt)circle(.8pt)circle(1.2pt)circle(1.4pt) node[anchor=360-360/8]{$8$};
	\foreach \x/\y in {-0.924/-0.383,0/-1.00000000000000,-0.707/0.707}{\draw[decorate,decoration={snake,amplitude=0.8mm}] (-0.544,-0.225) -- (\x,\y);}
	\draw (-0.544,-0.225) node[shift = {(.15,-.05)}] {\small2};
	\draw[decorate,decoration={snake,amplitude=0.8mm}] (-0.966,0.259) -- ++({atan(-0.268)}:.25 cm);
	\foreach \x/\y in {0.924/0.383,0.707/0.707,-0.866/0.500}{\draw[decorate,decoration={snake,amplitude=0.8mm}] (0.255,0.530) -- (\x,\y);}
\end{tikzpicture}
\end{align*}

To get the restricted generalized Wilson loop diagram, remove the vertices in $S$ from the diagram, and change the ends such that instead of being supported on $V(e)$, it is supported on $V(e) \setminus S$. If $V(e) \subseteq S$, then this end is removed entirely. In this manner, the propagator $r$ is changed and the propagator $s$ is removed. 
\end{eg}

 \subsubsection{Contractions of generalized Wilson loop diagrams}\label{sec:contractions}

Contraction of the matroid winds up being diagrammatically nice in the case where the restriction to the same set is capacity-ranked, which fortunately, is the case that will be useful later. We will first consider a more general contraction operation on the bipartite graph and then make the connection with the matroids.

\begin{dfn}\label{dfn:contraction graph}
	Let $W$ be a generalized Wilson loop diagram with associated bipartite graph $G_W^\cE$ and matroid structure $\sfM_\cE$ on the ends. For any $S \subseteq [n]$, set $c_S(p)$ to be the capacity of the propagator $p$ restricted to the set $S$ (Definition~\ref{dfn:restrictedcap}). Then define the graph $G_{W/S}^{\cE}$ to be the induced bipartite graph on $([n] \setminus S) \sqcup \cE$ with the matroid structure $\sfM_{\cE \setminus S}$ obtained by giving the set of ends of a propagator $p$ a uniform matroid structure of rank $c(p)- c_S(p)$. 
\end{dfn}

On the level of the diagrams this corresponds to the following steps. First,  remove the vertices $S$ from the outer boundary circle (but otherwise leaving the cyclic order the same). Secondly, for each $a \in \cE$, remove or adjust its support as done in the case of restriction. Namely, if $a$'s support is contained in $S$ remove $a$, if $a$'s support intersects $S$ without being contained, adjust the end of $a$ to be supported on the single supporting vertex not in $S$. If $V(a)$ is disjoint from $S$, do nothing to the end. Finally, change the capacities of each propagator in $W$ to $c(p) - c_S(p)$. Any propagators with $0$ capacity are removed. We call this diagram $W/S$.

\begin{eg} \label{eg:contraction}
Consider again the generalized Wilson loop diagram from Example \ref{eg:working example} and its associated bipartite graph $G_W^\cE$ from Example \ref{eg:workingegbipartite}. In this example, we let $S = \{v_2, v_7\}$. Then the bipartite graph $G_{W/S}^{\cE}$ and generalized Wilson loop diagram $W/S$ are as follows:
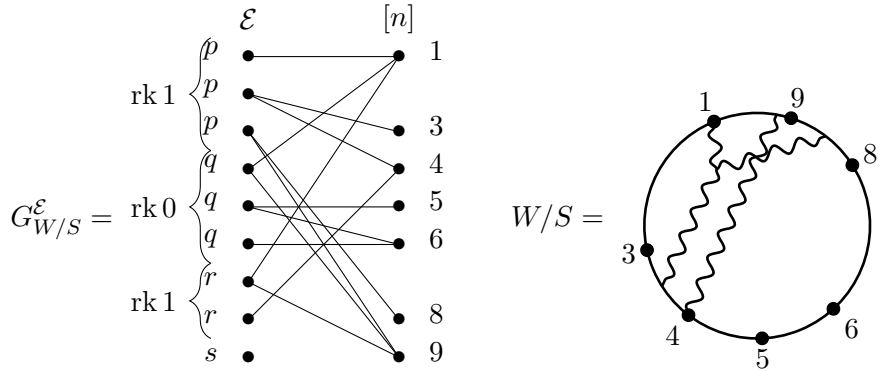
\begin{align*}
	G_{W/S}^{\cE} = \begin{tikzpicture}[scale=1,baseline=-2.75cm]
		\labeledvertices{ p, p,p, q, q, q, r, r, s}{1,3,4,5,6,8,9}{1/1, 2/3, 2/4, 3/8, 3/9,  4/1, 4/9, 5/5, 5/6, 6/6,7/1, 7/9, 8/4}
		\bipartitelabels{\cE}{[n]} 
		\draw [decorate,decoration={brace,amplitude=8pt, mirror},xshift=-0.5cm,yshift=0pt]
		(-1,-.25) -- (-1,-1.75) node [midway,right,xshift=-1.2cm] {$\rk 1$};
		\draw [decorate,decoration={brace,amplitude=8pt, mirror},xshift=-0.5cm,yshift=0pt]
		(-1,-1.75) -- (-1,-3.25) node [midway,right,xshift=-1.2cm] {$\rk 0$};
		\draw [decorate,decoration={brace,amplitude=8pt, mirror},xshift=-0.5cm,yshift=0pt]
		(-1,-3.25) -- (-1,-4.25) node [midway,right,xshift=-1.2cm] {$\rk 1$};
	\end{tikzpicture} 
	\qquad  W/S = \begin{tikzpicture}[rotate=-67.5, line width=1, scale=1.5,baseline = 0cm]
		\def \n {9}
		\draw circle(1)
		\foreach \v in{1,3,4,5,6,8,9}
		{(360*\v/\n-360/\n+180:1)circle(.4pt)circle(.8pt)circle(1.2pt)circle(1.4pt) node[anchor=360/\n*\v-360/\n-67.5]{$\v$}};
		\foreach \x/\y in {-1.00000000000000/0,0.174/-0.985,-0.500/0.866}{\draw[decorate,decoration={snake,amplitude=0.8mm}] (-0.6,-0.140) -- (\x,\y);}
		\draw[decorate,decoration={snake,amplitude=0.8mm}] (0.500,-0.866) -- (-0.850,0.542);
	\end{tikzpicture} 
\end{align*} Note that since the propagators $q$ and $s$ have their capacities reduced to $0$, they do not appear in the diagram $W/S$. 
\end{eg}
Write $\cM(W/S)$ to indicate the Rado matroid associated to the diagram $W/S$ and the bipartitite graph $G_{W/S}^{\cE}$. Then we have the following result. 

\begin{lem}\label{res:contraction matroid}
	Let $W$ be a capacity-ranked generalized Wilson loop diagram. The equality $\cM(W/S) = \cM(W)/S$ holds if and only if $c_S(\cP) = \rk_{\cM(W)}(S)$.
\end{lem}
\begin{proof}	
Write $S^c = [n]\setminus S$ for the complement of $S$.

By definition of the rank function and the capacity of the propagators, one always has that $\rk_{\cM(W)}(S) \leq \min\{|S|, c_S(\cP)\}$. In other words, we only need to check when $\rk_{\cM(W)}(S) = c_S(\cP)$ and when $\rk_{\cM(W)}(S) < c_S(\cP)$.

When $\rk_{\cM(W)}(S) = c_S(\cP)$, let $B$ be a basis of $\cM(W)$ with $\rk_{\cM(W)}(S) = | B \cap S|$. By construction, there is a $B$-perfect independent matching in $G_W^\cE$. In particular, there is $B \cap S^c$-perfect independent matching in $G_W^\cE$. By construction of $G_{W/S}^\cE$, this is also a matching on $B \cap S^c$ in $G_{W/S}^\cE$
It remains to check that it is also independent in $G_{W/S}^\cE$.

Since $W$ is capacity-ranked, for any $B$-perfect independent matching in $G_W^\cE$, the set $B \cap S^c$ is matched to exactly $c(p) - c_S(p)$ ends of $p$. Therefore, restricting the matching on $G_W^\cE$ to $G_{W/S}^\cE$, it remains independent. 

Next consider when $\rk_{\cM(W)}(S) < c_S(\cP)$. Again, let $B$ be a basis of $\cM(W)$ with $\rk_{\cM(W)}(S) = | B \cap F|< c_S(\cP)$. Since $W$ is capacity-ranked, there is a propagator $p \in \Prop(F)$ such that any $B$-perfect independent matching in $G_W^\cE$ connects the vertices of $B \cap S^c$ to strictly more than $c(p) - c_S(p)$ ends of $p$. Note that the rank of the ends of $p$ is $c(p) - c_S(p)$ in $G_{W/S}^\cE$:  $\rk_{\cE/S}(\cE(p)) = c(p) - c_S(p)$. Therefore, we have identified a subset of $B \cap S^c$, namely those vertices matched to ends of $p$ where the size of the set is strictly greater than the rank in $\cM(W/S)$. Therefore, this set is a dependent set in $\cM(W/S)$ but an independent set in $\cM(W)/S$. 
\end{proof}

\begin{eg} \label{eg:contraction condition}
In Example \ref{eg:contraction}, note that $\Prop(S)= \{p, q, s\}$ with $c_S(\cP) = 3$. Meanwhile, the set $S$ is independent in $\cM(W)$, therefore $\rk_{\cM(W)}(W) = 2$. In other words, in this example, one does not have $\rk_{\cM(W)}(S) = c_S(\cP)$. Therefore, we may not have that $\cM(W/S) = \cM(W)/S$. Specifically, the matroid $\cM(W)/S$ has rank $3$ while $\cM(W/S)$ has rank $2$. 
\end{eg}

The condition that $c_S(\cP) = \rk_{\cM(W)}(S)$ is particularly useful to us because all cyclic flats of $\cM(W)$ satisfy this condition. We state this as a lemma as we use it in the discussion of positivity of $\cM(W)$ in Section~\ref{sec:positive gWLD}.

\begin{lem}\label{res:cyclic flat prop ranked} 
	Let $W$ be a generalized Wilson loop diagram and $F$ a cyclic flat of $\cM(W)$. Then $c_F(\cP) = \rk_{\cM(W)}(F)$.
\end{lem}

\begin{proof}
	By definition of the rank function and the capacity of the propagators, one always has that $\rk_{\cM(W)}(F) \leq \min\{|F|, c_F(\cP)\}$. Therefore, we only need to check when $\rk_{\cM(W)}(F) < c_F(\cP)$. 
	
        Suppose now that $\rk_{\cM(W)}(F) < c_F(\cP)$. Let $\rk_{\cM(W)}(F) = r$ and let $M$ be a maximal independent matching in $G_W^\cE$ that matches $r$ elements of $F$: $|M \cap F | = r$. Let $U=M\cap F$ be the set of these elements. Since $r< c_F(\cP)$, we know that for any such $M$, there is a propagator $p \in \Prop(F)$ supported by $F$ such that $M$ does not match $c_F(p)$ vertices of $F$ to $c_F(p)$ ends of $p$. Let $a$ be an end of $p$ supported by $F$ that is not covered by $M$: $a \not \in M\cap \cE$; $a \in \End(F)\cap \cE(p)$. Let $x$ be a vertex in $F$ supporting $a$ that is covered by $M$: $x \in V(a) \cap F$ (if no such $x$ existed then $M$ could be extended contradicting maximality). Let $b$ be the end matching to $x$ by $M$.

        Consider now $\cl(U\setminus x)$. Note that since $U$ is independent and has closure $F$ we know that $x\not\in \cl(U\setminus x)$ and $\rk(\cl(U\setminus x))=r-1$.

        Take $y\in V, y\not\in \cl(x)$. If $y\neq x$ then $\{y,x\}$ is an independent set because of the following observations: $N(y)$ is nonempty since $y\not\in \cl(U\setminus x)$ and $F$ is a union of circuits, $|N(x)|>2$ by construction, and furthermore since $a$ is unsaturated by $M$, $\rk_\cE(N(x))\geq 2$.
        Additionally, for such a $y$, since $y\not\in\cl(U\setminus x)$, there must be a neighbor of $y$ which is either $b$ or is a neighbor of $x$ unsaturated by $M$. In either case we can make a new larger independent matching as follows. If $y$ has a neighbor which is a neighbor of $x$ unsaturated in $M$ then without loss of generality this neighbor is $a$ and then matching $a$ and $y$ gives a larger independent matching than $M$. On the other hand if $y$ has $b$ as a neighbor, then we can modify $M$ by matching $b$ to $y$ and $a$ to $x$ again giving a larger independent matching than $M$. In both cases we contradict the fact that $M$ is a maximum independent matching.

        Therefore, $x$ is the only element of $F$ which is not in $\cl(U\setminus x)$ and hence \bas \rk_{\cM(W)}(F) = \rk_{\cM(W)}(F\setminus x) + \rk_{\cM(W)}(x) \eas which shows that $x$ is not an element of any circuit contained in $F$ since $x$ is not itself a loop. Therefore, $F$ is not a cyclic flat.
\end{proof}

\section{Positroids and generalized Wilson loop diagrams\label{sec:gWLDtopositroid}}

We are now ready to consider which capacity-ranked generalized Wilson loop diagrams (i.e.\ those satisfying $\rk(\cM(W)) = c(\cP)$) correspond to positroids. Towards this end we give an algorithm for associating a Grassmann necklace to each generalized Wilson loop diagram. This algorithm is a direct generalization of the process of passing from an admissible ordinary Wilson loop diagram to a Grassmann necklace in the case where there are no crossing propagators, see Algorithm 2 \cite{reversingOh}. Then we give a diagrammatic classification of which generalized Wilson loop diagrams give rise to positroids. 

\subsection{Grassmann necklaces of Rado matroids \label{sec:GN of Rado}}

First we show that we can get the Grassmann necklace of the matroid $\cM(W)$ for any generalized Wilson loop diagram $W=(\cP, [n])$ by consistently choosing the next unsaturated element of $[n]$ that is connected to an unsaturated vertex of $\cE$. The key to showing this works is Lemma \ref{res:greedyGN}.

This algorithm will build each set $I_i$ in the Grassmann necklace by, for each $i$, incrementing $v\in [n]$ starting from $i$ and working cyclically.

\begin{algorithm}\label{algo:gengwld}
	Let $W=(\cP, [n])$ be a generalized Wilson loop diagram with associated matroid $\cM(W)$ and bipartite graph $G_W^\cE$. We define a sequence $I(W) = \{I_1, \ldots, I_n\}$ of $\rk(\cM(W))$-tuples of $[n]$ as follows:

    For each $i$:
	\begin{enumerate}
		\item To construct the element $I_i$, start at the vertex $i$ in $W$. Initialize the following parameters:
		\begin{enumerate}
			\item Set $I_i = \emptyset$.
			\item Set $v = i$. 
		\end{enumerate}  
		\item \label{step 2}Consider $G^\cE_W$. 
		\begin{enumerate} 
			\item If there exists an independent matching in $G^\cE_W$ such that $I_i\cup v$ is saturated then set $I_i := I_i \cup v$.
			\item If $v +1 = i$, return $I_i$. Otherwise, set $v = v+1$ (with indices interpreted cyclically as always).
			\item Repeat step \ref{step 2}.
		\end{enumerate}
	\end{enumerate}
\end{algorithm}

\begin{prop}\label{res:algogenwldgivesGN}
    The $I(W)$ constructed by Algorithm~\ref{algo:gengwld} is the Grassmann necklace of $\cM(W)$.
\end{prop}

\begin{proof}
    Lemma~\ref{res:greedyGN} tells us that we can build the Grassmann necklace element $I_i$ for $\cM(W)$ by repeatedly taking the first element in the $<_i$ order which is independent of what we have so far. Once we reach as many elements in $I_i$ as the rank of the matroid, no more can be taken in this manner, so we are free to stop the algorithm when the index cycles back to $i$, as the algorithm does.
\end{proof}

In the case of a generalized Wilson loop diagram $W$ that both has no crossings and another technical property we define in Definition \ref{dfn:admissiblegWLD}, this algorithm simplifies to a process that can be read directly and efficiently off the diagram in a manner which generalizes the algorithm for admissible ordinary Wilson loop diagrams. It is this simplified version of the algorithm (Algorithm \ref{algo:non-crossgwld}) that we implement in the work discussed in Section~\ref{sec:diagrammaticmoves}. %

Recall from Definitions \ref{dfn:new def of prop order at edge or vertex} and \ref{dfn:global cyclic end order} that if there are no crossing propagators then there is a natural cyclic ordering of the ends of a generalized Wilson loop diagram. When the diagram is drawn with no internal intersecting propagators, this cyclic ordering corresponds to ordering the ends in a counterclockwise manner (Remark~\ref{remark drawing is non-crossing}).

We will need to \emph{mark} ends when they have been used in the algorithm or are no longer available to be used on account of the capacity.

\begin{algorithm}\label{algo:non-crossgwld}
	Let $W=(\cP, [n])$ be a generalized Wilson loop diagram with associated matroid $\cM(W)$. Assume that $W$ has no crossing propagators. We define a sequence $I(W) = \{I_1, \ldots, I_n\}$ of subsets of $[n]$ as follows:

    For each $i$:
	\begin{enumerate}
		\item To construct the element $I_i$, start at the vertex $i$ in $W$. Initialize the following parameters:
		\begin{enumerate}
			\item Set $I_i = \emptyset$. 
			\item Set $v = i$. 
		\end{enumerate}  
		\item \label{loop step of alg2} If $v$ supports an unmarked end,  let $a$ be the smallest such end in the linear ordering given by restricting the cyclic ordering of  Definition \ref{dfn:global cyclic end order} to the edge before $v$, $v$ itself, and the edge after $v$:
		\begin{enumerate}
			\item If there is a single-ended propagator $r$ supported only on the vertex $v$, set $a$ to be $r$. (If there is more than one such $r$, pick one.) 
                \item Set $I_i = I_i \cup v$.
                \item Mark the end $a$.
                \item If any propagator $p$ has $c(p)$ marked ends, mark all ends of $p$.
			\item If $v+1 = i$, return $I_i$ and say that $W$ is not admissible if $|I_i|\neq c(\cP)$. \linebreak Otherwise, repeat step \ref{loop step of alg2} with $v = v+1$. 
		\end{enumerate}	
        Else ($v$ does not support an unmarked end):
		\begin{enumerate}
			\item Do not change $I_i$.
			\item If $v+1 = i$, return $I_i$ and say that $W$ is not admissible if $|I_i| \neq c(\cP)$. \linebreak Otherwise, repeat step \ref{loop step of alg2} with $v = v+1$. 
		\end{enumerate}
		
	\end{enumerate} 
\end{algorithm}

Note that this algorithm only works for diagrams with non-crossing propagators, as otherwise, there is not a canonical ordering that allows one to choose the clockwise-most end, as illustrated in Figure \ref{eg:manyrepscrossing}.

\begin{dfn}\label{dfn:admissiblegWLD}
Let $W=(\cP, [n])$ be a non-crossing generalized Wilson loop diagram. If all the $I_i$ returned by Algorithm~\ref{algo:non-crossgwld} have size $|I_i| = c(\cP)$ for all $i \in [n]$, then $W$ is \emph{admissible}.
\end{dfn}

Note that by this definition, if a generalized Wilson loop diagram is not capacity-ranked, it is not admissible. 

A useful alternate characterization of admissibility is as follows.
\begin{lem}\label{res:admissible local def}
    Let $W=(\cP, [n])$ be a non-crossing generalized Wilson loop diagram. $W$ is admissible if and only if at each step in the construction of each $I_i$ in Algorithm~\ref{algo:non-crossgwld}, any propagator $p$ supported entirely in the cyclic interval $[i, v]$ has had all of its ends marked by step $v$ in the construction of $I_i$.
\end{lem}

\begin{proof}
    Say a propagator \emph{contributes} an end to the Grassmann necklace element $I_i$ when that end is the $a$ marked in the first part of Step~\ref{loop step of alg2}.

    Suppose the condition of the lemma holds, i.e.\ any propagator $p$ supported entirely in the cyclic interval $[i, v]$ has had all of its ends marked by step $v$ of the construction of $I_i$. Then for each $I_i$ by the end of the construction of all propagators will have all ends marked. So all propagators will have contributed capacity-many elements to $I_i$ and hence $W$ will be admissible. Suppose the condition does not hold. Then in the construction of some $I_i$ at some step all the supporting vertices of some propagator will have passed without having contributed capacity-many ends. Since no further ends of this propagator can be encountered by the algorithm while constructing $I_i$, this propagator cannot contribute capacity-many ends to the final $I_i$. Each propagator can contribute at most capacity-many ends, so as a whole $I_i$ must have fewer than $c(\cP)$ elements. Hence $W$ is not admissible. 
\end{proof}

Admissibility is a condition that fails when there is a high density of ends that are badly positioned. In that case the second part of the equivalence in Lemma~\ref{res:admissible local def} will fail to hold and hence the diagram will not be admissible. Unfortunately, we do not have a nice non-algorithmic characterization of admissibility, though one special case is that when the diagram as a whole is not capacity-ranked then it is not admissible. The choice of the word admissible was deliberate since while admissibility in ordinary Wilson loop diagrams has a definition of a very different character (as an explicit density condition in part (4) of Definition \ref{dfn:admissibleWLD}), its role in assuring the correct behaviour of the Grassmann necklace algorithm is the same.  In particular, note that an admissible ordinary Wilson loop diagram is also admissible in this more general sense.

\begin{eg}\label{eg:admissiblity}
To understand the difference between the two algorithms presented here, consider the two generalized Wilson loop diagrams and their associated bipartite graphs below. 

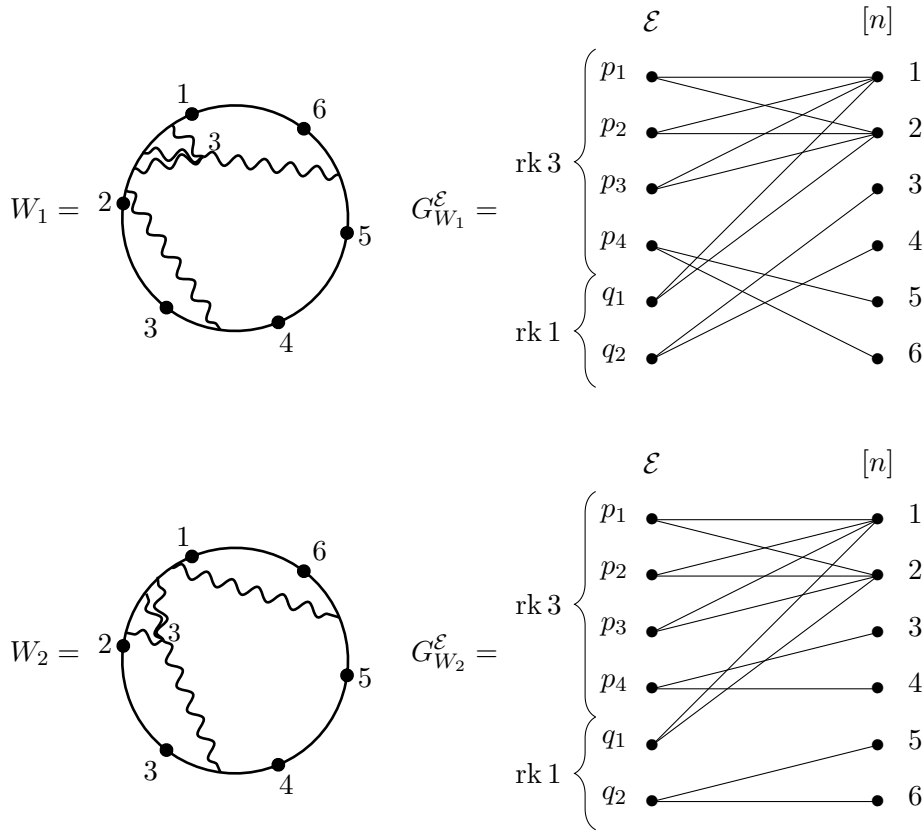
\bas	W_1= \begin{tikzpicture}[rotate=-67.5, line width=1, scale=1.5,baseline=0]
	\def \n {6}
	\draw circle(1)
	\foreach \v in{1,...,\n}
	{(360*\v/\n-360/\n+180:1)circle(.4pt)circle(.8pt)circle(1.2pt)circle(1.4pt) node[anchor=360/\n*\v-360/\n-67.5]{$\v$}};
	\foreach \x/\y in {-0.975/-0.2,-0.850/-0.542,-0.750/-0.666,0/1.00000000000000}{\draw[decorate,decoration={snake,amplitude=0.8mm}] (-0.610,-0.052) -- (\x,\y);}
	\draw (-0.610,-0.052) node[shift = {(.15,.20)}] {\small3};
	\draw[decorate,decoration={snake,amplitude=0.8mm}] (-0.600,-0.8) -- (0.866,-0.500);
\end{tikzpicture} \quad G_{W_1}^\cE = \begin{tikzpicture}[scale=1.5,baseline=-2.625cm]
			\labeledvertices{ p_1, p_2, p_3,p_4,q_1,q_2}{1,...,6}{1/1, 1/2, 2/1, 2/2, 3/1, 3/2, 4/5, 4/6, 5/1, 5/2, 6/3, 6/4}
			\bipartitelabels{\cE}{[n]}
			\draw [decorate,decoration={brace,amplitude=8pt, mirror},xshift=-0.5cm,yshift=0pt]
			(-1,-.25) -- (-1,-2.25) node [midway,right,xshift=-1.2cm] {$\rk 3$};
			\draw [decorate,decoration={brace,amplitude=8pt, mirror},xshift=-0.5cm,yshift=0pt]
			(-1,-2.25) -- (-1,-3.25) node [midway,right,xshift=-1.2cm] {$\rk 1$};
					\end{tikzpicture} \eas 
\bas 	W_2= \begin{tikzpicture}[rotate=-67.5, line width=1, scale=1.5,baseline=0]
	\def \n {6}
	\draw circle(1)
	\foreach \v in{1,...,\n}
	{(360*\v/\n-360/\n+180:1)circle(.4pt)circle(.8pt)circle(1.2pt)circle(1.4pt) node[anchor=360/\n*\v-360/\n-67.5]{$\v$}};
	\foreach \x/\y in {-0.950/-0.35,-0.850/-0.5, -0.600/-0.8, 0.866/-0.500}{\draw[decorate,decoration={snake,amplitude=0.8mm}] (-0.410,-0.52) -- (\x,\y);}
	\draw (-0.410,-0.52) node[shift = {(.15,.10)}] {\small3};
	\draw[decorate,decoration={snake,amplitude=0.8mm}] (-0.975, -0.2) -- (0.0,1.000);
\end{tikzpicture} \quad G_{W_2}^\cE = \begin{tikzpicture}[scale=1.5,baseline=-2.625cm]
	\labeledvertices{ p_1, p_2, p_3,p_4,q_1,q_2}{1,...,6}{1/1, 1/2, 2/1, 2/2, 3/1, 3/2, 4/3, 4/4, 5/1, 5/2, 6/5, 6/6}
	\bipartitelabels{\cE}{[n]}
	\draw [decorate,decoration={brace,amplitude=8pt, mirror},xshift=-0.5cm,yshift=0pt]
	(-1,-.25) -- (-1,-2.25) node [midway,right,xshift=-1.2cm] {$\rk 3$};
	\draw [decorate,decoration={brace,amplitude=8pt, mirror},xshift=-0.5cm,yshift=0pt]
	(-1,-2.25) -- (-1,-3.25) node [midway,right,xshift=-1.2cm] {$\rk 1$};
\end{tikzpicture}\eas
	Observe that both these diagrams correspond to the same matroid. Indeed, by Example \ref{eg:cap equal end num}, Example \ref{eg: half props clear}, and Example \ref{eg:no multi-ends} both $W_1$ and $W_2$ represent the same matroid as the following diagram.
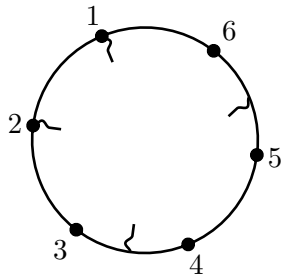
\bas
 \begin{tikzpicture}[rotate=-67.5, line width=1, scale=1.5]
		\def \n {6}
		\draw circle(1)
		\foreach \v in{1,...,\n}
		{(360*\v/\n-360/\n+180:1)circle(.4pt)circle(.8pt)circle(1.2pt)circle(1.4pt) node[anchor=360/\n*\v-360/\n-67.5]{$\v$}};
		\draw[decorate,decoration={snake,amplitude=0.8mm}] (-1.00000000000000,0) -- ++({atan(0)}:.25 cm);
		\draw[decorate,decoration={snake,amplitude=0.8mm}] (-0.500,-0.866) -- ++({atan(1.73)}:.25 cm);
		\draw[decorate,decoration={snake,amplitude=0.8mm}] (0.866,-0.500) -- ++({atan(-0.577)}:-.25 cm);
		\draw[decorate,decoration={snake,amplitude=0.8mm}] (0,1.00000000000000) -- ++({atan(-2.7)}:.25 cm);
	\end{tikzpicture} \eas

In both $W_1$ and $W_2$, the ends $p_1, p_2, p_3, q_1$ are ends on the edge $e_1$. In $W_1$ these ends appear in the linear order $p_1, p_2, p_3, q_1$ on $e_1$. In $W_2$ they appear in the linear order $q_1, p_1, p_2, p_3$ on $e_1$. For both diagrams, Algorithm \ref{algo:gengwld} gives the same Grassmann necklace: \bas \cI(W_1) = \cI(W_2) = \{ 1235, 2351, 3512, 4512, 5123, 6123\}\;. \eas

To understand Algorithm \ref{algo:non-crossgwld}, when calculating $I_1$ for $W_1$, when $v = 1$, the edge $p_1$ is marked. At $v=2$, $p_2$ is marked, at $v = 3$, $q_2$ is marked, then one notes that $q$, which has capacity $c(q) = 1$, has one marked end, and therefore all ends of $q$ are marked. Namely, $q_1$ is also marked. At $v=4$ there are no unmarked ends. At $v=5$, $p_4$ is marked, which means that, with $c(p)= 3$, there are the same number of marked ends as capacity, so all ends of $p$ are marked. Namely, $p_3$ is also marked. At $v = 6$, there are no remaining unmarked ends in the diagram at all, so there are not unmarked ends supported on $v$. Therefore $I_1 = 1235$.

For $W_2$, $q_1$, and thus $q_2$ is marked at $v=1$, $p_1$ is marked at $v=2$, $p_4$ is marked at $v = 3$ and there are no unmarked ends supported on $v = 4,5,$ or $6$. Therefore, $I_1 = 123$. 

This shows that Algorithm \ref{algo:non-crossgwld} produces different results for $I_1$ for $W_1$ and $W_2$. One may further check that the algorithms agree for all other $i$. Therefore, $W_1$ is admissible while $W_2$ is not. 
\end{eg}

We are now ready to prove that Algorithm \ref{algo:non-crossgwld} is equivalent to Algorithm \ref{algo:gengwld} in the case when the diagram contains no crossing propagators and is admissible. 

\begin{lem}\label{res:algo-noncross admissible gives GN}
	When a generalized Wilson loop diagram $W$ has no crossing propagators and is admissible, then Algorithm \ref{algo:non-crossgwld} produces the Grassmann necklace of $\cM(W)$.
\end{lem}

\begin{proof}
We prove this result by showing that Algorithm \ref{algo:non-crossgwld} is equivalent to Algorithm \ref{algo:gengwld} under the non-crossing and admissibility conditions. The key to this proof lies in the fact that with these assumptions there is an unmarked end of a propagator supported on a vertex $v$ in Algorithm \ref{algo:non-crossgwld} if and only if there is an independent $I_i \cup v$ matching at that point in the algorithm. 

Suppose that both algorithms start at the vertex $i$. That is, in both cases, one is constructing the Grassmann necklace element $I_i$. 

At each step of both algorithms we have have a matching saturating the $I_i$ that has been constructed so far. In Algorithm~\ref{algo:gengwld} this matching exists by assumption from the previous iteration. In Algorithm~\ref{algo:non-crossgwld} this matching is the one which, from each previous iteration, matches the $a$ of that iteration to the $v$ of that iteration.

We proceed by induction on $v$. For the base case, in both algorithms, $I_i^{(1)}$ is the first vertex in the $<_i$ ordering that supports any ends of $W$, so the algorithms agree on $v=i$ up to $v = I_i^{(1)}$.

Now proceed with the inductive case. 
Note that in Algorithm \ref{algo:non-crossgwld}, at each step, after the end $a$ is marked, the closure of the entire marked set (in $\sfM_\cE$) is also marked. Assume that, for some $i <_i v$, both algorithms match the same subset $I_i = \{I_i^{(1)} \ldots I_i^{(j)}<v \} \subset [n]$ on the steps up to but not including step $v$ and now we are at step $v$. This means that $I_i$ is independent in $\cM(W)$. Let $E$ be the set of marked ends in $W$ by Algorithm \ref{algo:non-crossgwld} when it reaches $v$. 

Suppose that there is an unmarked neighbor of $v$ at step $v$ in Algorithm~\ref{algo:non-crossgwld}. Because the set $E$ is a flat in $M_\cE$, any unmarked neighbor of $v$ is independent of $N(I_i)$ in the matroid $M_\cE$ at that point. Specifically, the end $a$ identified in Algorithm \ref{algo:non-crossgwld} is independent of $N(I_i)$. Therefore, we may define an independent $I_i \cup v$-perfect matching, $M$, in $G_W^\cE$ where each vertex is matched to the end initially marked at that stage of the algorithm. Therefore, since a matching exists on $I_i \cup v$, $v$ is the next element in the Grassmann necklace according to Algorithm \ref{algo:gengwld}.

Now suppose that at step $v$ in Algorithm~\ref{algo:non-crossgwld} there are no unmarked ends supported at $v$. 
There may be some unmarked ends supported on $[i, v-1]$. By admissibility, the propagators with these ends will have to be picked up later in the algorithm. Let $j$ be the first vertex of $[i,v-1]$ which does not support an unmarked end. 

Consider the interval $[j,v]$. By construction all ends supported on $[j,v]$ are marked. Suppose $b$ is an end supported on $[j,v]$ that was marked but not matched. Let $b'$ be the first end of the same propagator as $b$ that was matched in constructing $I_i$. We claim that $b'$ is supported on $[j,v-1]$. Suppose not. Then an unmarked end would lie between $b'$ and $b$, but the propagator containing the unmarked end must, by admissibility, have at least one end that gets marked later in the algorithm and so lies after $b$. This contradicts non-crossingness.

Let $U = (I_i\cap [j,v-1]) \cup \{v\}$. By the observation above, the rank of the ends supported on $U$ is $|I_i\cap [j,v-1]|$ since every propagator with an unmatched end in $[j,v]$ has capacity-many matched ends in $I_i\cap [j,v-1]$. Therefore $U$ is dependent in $\cM(W)$ and so $I_i\cup \{v\}$ is dependent in $\cM(W)$. Thus Algorithm~\ref{algo:gengwld} also does not add $v$ at this step of the algorithm. 

Both cases together show that Algorithm~\ref{algo:non-crossgwld} adds $v$ to $I_i$ if and only if 
Algorithm~\ref{algo:gengwld} adds $v$ to $I_i$. By induction both algorithms generate the same $I_i$.

\end{proof}

\subsection{Positive generalized Wilson loop diagrams \label{sec:positive gWLD}} 
We show that if a capacity-ranked generalized Wilson loop diagram has all non-crossing propagators, then the associated matroid is a positroid. 

In this section, we assume that all propagators with $c(q) = |\End(q)|$ have a single end, which we can do without loss of generality by Example~\ref{eg:cap equal end num}. We also restrict ourselves to capacity-ranked Wilson loop diagrams.

Before we prove results about how the positivity of $\cM(W)$ corresponds to the crossing structure of its positroids, we give a different result about the positivity of $\cM(W)$ in terms of dependent connected flats of $\cM(W)$. We begin with some useful lemmas and definitions.

Recall the definition of direct sum and connectivity of matroids from Definition~\ref{dfn:matroid direct sum}.

\begin{lem} \label{res:disjointprops}
	Let $W = (\cP, [n])$ be a generalized Wilson loop diagram and suppose $S \subseteq [n]$ can be partitioned $S = S_1 \sqcup S_2$ such that the supported propagators are also partitioned, $\Prop(S) = \Prop(S_1) \sqcup \Prop(S_2)$ (either of which may be empty). Then $\cM(W)|_S$ is a disconnected matroid: $\cM(W)|_S =  \cM(W)|_{S_1} \oplus \cM(W)|_{S_2}$. 
\end{lem}
\begin{proof}
The bipartite graph $G^\cE_W|_S$ (as defined in Section~\ref{sec:restrictions}) is disconnected under the vertex partition $(S_1\sqcup N(S_1)) \sqcup (S_2 \sqcup N(S_2))$ by the assumptions on the propagators. The matroid $\sfM_\cE|_{N(S)}$ is then a direct sum of $\sfM_\cE|_{N(S_1)}$ and $\sfM_\cE|_{N(S_2)}$ since the capacity is additive on propagators. Lemma~\ref{res:rado disconnected} then gives the result.

\end{proof}

\begin{rmk}
	Note that Lemma \ref{res:disjointprops} does not generalize to allow for a partition of ends. If the ends are partitioned then the bipartite graph is disconnected in the appropriate way, but the matroid $\sfM_Y$ is not necessarily a direct sum and so $\cM(W)|_S$ will not typically be a direct sum in this circumstance. 
	
	For example, for the diagram in Example \ref{eg:working example}, we may set $S = [9]$ and $S_1 = [8,2]$ and $S_2 = [3,7]$. In this case, the sets $S_1$ and $S_2$ form a partition of $[n]$ and the set $\End(S_1)$ and $\End(S_2)$ form a partition of $\cE(\cP)$. However, notice that $c_{S_1}(p) = 2$, $c_{S_1}(q) = 1$, $c_{S_1}(r) = 1$, $c_{S_1}(s) = 0$ while $c_{S_2}(p) = 1$, $c_{S_2}(q) = 1$, $c_{S_2}(r) = 1$, $c_{S_2}(s) = 1$.
\end{rmk}

In the sequel, we make use of the fact that removing vertices that do not support a propagator is the same as removing loops from the ground set of a matroid, so it doesn't change the independence structure of the remaining vertices, as the next result shows. The generalized Wilson loop diagrams where all vertices support at least one propagator gives us some nice flats to work with and the restriction to such diagrams is crucial for reducing the calculation of Section~\ref{sec:diagrammaticmoves} to a finite computation. 

\begin{lem}\label{res:ignorenon-supporting} 
	Let $W$ be a generalized Wilson loop diagram with a set of vertices, $F$, that do not support any propagators. Then the matroid structure of $W$ is given by $\cM(W) = M(W|_{F^c}) \oplus M(W|_F)$ where $M(W|_F)$ is a matroid of rank $0$.
\end{lem}
\begin{proof}
	Note that if $W$ has any vertices that do not support any propagators, these vertices have rank $0$; they are not contained in any basis of $\cM$. That is, they are loops. Therefore, the set of bases of a generalized Wilson loop diagram with vertices of rank $0$ is the same as the set of bases of the same matroid restricted to the vertices of rank $1$. 
\end{proof}

Specifically, excluding vertices which do not support any propagators does not change the associated Grassmann necklace, and therefore does not change whether or not the the matroid is a positroid. Namely, a generalized Wilson loop diagram gives rise to a positroid if and only if the same generalized Wilson loop diagram with vertices not supporting any propagators removed gives rise to a positroid.

The following set of flats and their properties play a central role in our main result.

\begin{dfn} \label{dfn:condepflat} 
	Let $W  = (\cP, n)$ be a capacity-ranked generalized Wilson loop diagram such that every vertex supports a propagator. Let \bas \cF= \{ F \subset [n] : \rk(F)<\rk(\cM(W)); \; F \textrm{ connected dependent flat of } \cM(W) \} \eas be the set of connected dependent flats of $\cM(W)$ which are not of full rank.
\end{dfn}

\begin{rmk}\label{rmk:cFconnectedcyclic} 
	Note that the flats in $\cF$ are all connected cyclic flats. To see this, for $F \in \cF$, let $C \subseteq F$ be the largest union of circuits (with respect to number of elements) contained in $F$, which is non-empty since $F$ is dependent. Then $F \setminus C$ is independent (otherwise any circuit in $F \setminus C$ would be in $C$), and any $y \in F\setminus C$ satisfies $\rk(C \cup y) = \rk(C) +1$ (otherwise, $y$ would be contained in $C$). Together, this means that $\rk(F \setminus C) + \rk(C) = \rk(F)$, violating connectivity, unless $C = F$. 
\end{rmk}

Note that if $W$ has vertices that do not support any propagator, these vertices are of rank $0$ and hence are elements of every flat. Therefore, any flat in $\cM(W)$ is disconnected. This explains the reason for restricting to generalized Wilson loop diagrams where all vertices support at least one propagator in the statements below.
 
One nice property of the set $\cF$ is that one can associate to each connected cyclic flat $F \in \cF$, a non-empty set of propagators $Q_F$ whose capacities when restricted to $F$ are strictly less than their capacities in the full matroid.

\begin{dfn} \label{dfn:QF}
	For each $F \in \cF$, let $Q_F$ be the set of propagators of $W$ whose capacities change when restricted to $F$: \bas Q_F = \{q \in \cP : c_F(q) < c(q) \}\;.\eas 
\end{dfn}
Note that a propagator not supported by the vertices in $F$ is automatically in $Q_F$ because the capacity drops to $0$ in this case.
\begin{lem}\label{res:QFnonempty}
	For a capacity-ranked generalized Wilson loop diagram $W = (\cP, [n])$ where all the vertices support a propagator, fix any $F \in \cF$. The associated set $Q_F$ is non-empty.
\end{lem}
\begin{proof}
	There are two cases to consider: when the vertices in $F$ support all the propagators of $W$ ($\Prop(F) = \cP$) and when they do not ($\Prop(F) \neq \cP$).
	
	In the second case, there is a propagator $q \in \cP  \setminus \Prop(F)$. The capacity of $q$ when restricted to $F$ is $0$, but by definition, the propagator has at least capacity $1$ in the associated matroid. Therefore $c_F(q) < c(q)$. In fact, the set $Q_F$ always contains all propagators without any supporting vertices in the set $F$.
	
	In the first case, we have that $\Prop(F) = \cP$. Since $\Prop(F) = \cP$ and $W$ is capacity-ranked, we know that $\rk([n]) = c(\cP) = \sum_{q \in \cP} c(q)$. Since $F$ does not have full rank, we have $\rk(F) < c(\cP)$, so 
 \bas \sum_{q \in \cP} c_F(q) = c_F(\cP) = \rk(F) < c(\cP) = \sum_{q \in \cP} c(q)\;. \eas Since $\sum_{q \in \cP} c_F(q) < \sum_{q \in \cP} c(q)$ is a strict inequality, there is some $q \in \cP$ such that $c_F(q) < c(q)$. This propagator makes $Q_F$ non-empty.
\end{proof}

Recall that Lemma \ref{res:cyclic flat prop ranked} applies in this case, so that $c_F(\cP) = \rk_{\cM(W)}(F)$ holds for each $F \in \cF$.

For $q\in Q_F$, note that, since only single-ended propagators satisfy $c(q) = |\End(q)|$, for multi-ended propagators, $c_F(q) < c(q ) < |\End(q)|$. In other words, if $q\in Q_F$ is multi-ended, it has at least $2$ ends, $q_i$ and $q_j$ with support disjoint from $F$: $V(q_i) \cap F = \emptyset$ and $V(q_j) \cap F = \emptyset$. 

The following result is a useful corollary of Lemma \ref{res:non min and co-loops}.

\begin{lem}\label{res:no multi-ended q co-loops}
	If $W$ is locally minimal and $x$ is a co-loop in $\cM(W)$, then $x$ only supports a single-ended propagator, which is in $Q_F$.
\end{lem}
\begin{proof}
	Since $x$ is a co-loop and $F$ is a connected dependent flat, $x \not \in F$.
	By Lemma \ref{res:non min and co-loops}, we know that $x$ supports a single propagator, which is a single-ended propagator with support only on $x$. Since this propagator does not have support in $F$, it is in $Q_F$. 
\end{proof}

\begin{dfn} \label{dfn:Intset}
	For each $F \in \cF$, define $\Int_F$ to be the set of cyclic intervals $[a,b] \subseteq [n]$ that contain $F$.
\end{dfn} 

The point here is that $F$ is a collection of cyclic intervals and we are interested in those cyclic intervals which contain all the intervals of $F$.

We state our main result. For examples, see Example \ref{eg:exampleFQ_F}.

\begin{thm} \label{res:posivityonflats}
Let $W = (\cP, [n])$ be a locally minimal, capacity-ranked generalized Wilson loop diagram, where every vertex supports a propagator. 

The matroid $\cM(W)$ is a positroid if and only if, for all $F \in \cF$ and $q \in Q_F$, there is a cyclic interval $[a,b] \in \Int_F$ such that the ends of $q$ are supported only by the vertices of $F$ in $[a,b]$: \ba \cE(q) \cap \End([a,b] \setminus F) \subseteq    \emptyset\;. \label{eq:supportonlyinF}\ea 
\end{thm}

We will prove the two directions separately. In fact only one of the directions will require local minimality in hypothesis, and the direction which does not require it is the direction which we use in Section~\ref{sec:diagrammaticmoves}. First a small lemma regarding what ends are possible for a $q$ in $Q_F$.
\begin{lem}\label{res:q ends in or out}
    Let $W$ be a capacity-ranked generalized Wilson loop diagram such that every vertex supports a propagator and let $F\in \cF$. For every $q\in Q_F$ every end of $q$ is either entirely supported in $F$ or is entirely supported outside of $F$.
\end{lem}

\begin{proof}
    Suppose otherwise, so $q$ has an end $a$ on an edge with one vertex in $F$ and one vertex not in $F$. Let $v$ be the vertex in $F$. Since the capacity of $q$ drops in $F$, $q$ has at most capacity-many ends supported on $F$, and hence the ends of $q$ are independent when restricted to $F$. Now we claim $\rk(F) = \rk(F\setminus \{v\}) + \rk(v)$. For this, note that in going from the bipartite graph $G^\cE_W |_{F\setminus \{v\}}$ to the bipartite graph $G^\cE_W|_{F}$ we add $v$ on one side of the bipartition and add $a$ (and possibly other ends) on the other side of the bipartition. Since the ends of $q$ restricted to $F$ are independent, we can extend any independent matching on $G^\cE_W|_{F\setminus \{v\}}$ to $G^\cE_W|_F$ by adding an edge between $a$ and $v$. This gives the desired additivity of the ranks and implies that $F$ is disconnected. This contradicts $F\in \cF$, proving the lemma.
\end{proof}

\begin{lem}\label{res:cond5 implies positroid}
    Let $W$ be a capacity-ranked generalized Wilson loop diagram with every vertex supporting a propagator. Furthermore assume that the condition \eqref{eq:supportonlyinF} holds for all $F\in \cF$ and all $q\in Q_F$. Then $\cM(W)$ is a positroid.
\end{lem}

Note that this lemma does not assume that $W$ is locally minimal.

\begin{proof}
Write $\cM(W) = ([n], \cB)$, with $\cB$ the set of bases. Write $\rk(\cM(W)) = k$. Since $W$ is capacity-ranked, we also have that $c(\cP) = k$. Let $\cI_W$ be the Grassmann necklace of $\cM(W)$, (which can be obtained through Algorithm~\ref{algo:gengwld}), and let $\cB_{\cI_W}$ the set of bases defined by $\cI_W$. Recall that the matroid $\cM(W)$ is a positroid if and only if $\cB_{\cI_W} = \cB$, where $\cB$ is the set of bases defining $\cM(W)$ \cite{Oh, Postnikov}. Furthermore, as $\cB_{\cI_W} \supseteq \cB$ always holds, $\cM(W)$ is a positroid if and only if $\cB_{\cI_W} \setminus \cB$ is empty. So, to show that $\cM(W)$ is a positroid, we will show that for any dependent set $B$ of size $k = \rk(\cM(W))$, there is an index $v \in [n]$ such that $I_v \not \leq_v B$. That is, if $B\not \in \cB$, then it is not an element of the set of bases defined by the Grassmann necklace $\cB_{\cI_W}$. Therefore, $\cM(W)$ is a positroid.

Now we proceed with the proof. Let $B$ be a set of size $k$ that is not in $\cB$. Let $U$ be a circuit in $B$, and $F = \cl(U)$ a connected cyclic flat. Note that $U$ being a circuit forces $F$ to be connnected. $F$ is not full rank as it is the closure of a circuit of a dependent set that is too small to contain a basis. Therefore, $F \in \cF$. Fix $q \in Q_F$. By hypothesis, there exists an interval $[a, b]$ containing $F$ such that the condition in \eqref{eq:supportonlyinF} holds. Without loss of generality, we may assume that $a, b \in F$, that is, the interval $[a, b]$ does not contain a smaller interval satisfying \eqref{eq:supportonlyinF}. There are two cases to consider: when $F = [a, b]$ and when $F \subsetneq [a, b]$. 

Recall from Theorem \ref{res:greedyGN} and Algorithm \ref{algo:gengwld} that the elements of the Grassmann necklace $\cI_W$ are formed by moving through the set $[n]$ starting at $a$, and including a vertex in the set if and only if the new vertex is independent of the set constructed so far. Therefore, in the first case, if $F = [a,b]$, the first $\rk([a,b]) = |U|-1$ elements of $I_a$ lie in the interval $[a,b]$. By construction, $B$ contains at least $|U|$ elements in $[a,b]$. Therefore, in the $<_a$ order, $ B^{(|U|)} \leq_a b <_a  I_a^{(|U|)}$ (using the notation of Definition \ref{dfn:Gale order}). Therefore, $I_a \not \leq_a B$ in the Gale ordering.

In the second case, $F \subsetneq [a,b]$. In this case, write $F$ as the partition of cyclic intervals: $F  = F_1 \sqcup \ldots \sqcup F_d$, with $F_i = [a_i, b_i]$. Similarly, denote the cyclic intervals in the complement as $F_i^c = [b_i+1, a_{i+1}-1]$. That is, one alternates $F_i, F_i^c, F_{i+1}$, etc. Note that $a_1 = a$ and $b_d = b$. Furthermore, we may have $a_i = b_i$, but necessarily $b_i + 1 < a_{i+1}$. Since $d \geq 2$, and $F$ is a flat, every element of $[a,b] \setminus F$, supports a propagator $p$, such that $c_F(p) < c(p)$. (Otherwise, for $x \in [a,b] \setminus F$, we could not have $\rk (F \cup x) > \rk (F)$.) 

Since $F$ is a cyclic flat, by Lemma \ref{res:contraction matroid} and Lemma \ref{res:cyclic flat prop ranked}, we may write $\cM(W)/F = \cM(W/F)$. We claim that since condition \eqref{eq:supportonlyinF} of Theorem \ref{res:posivityonflats} holds, the cyclic intervals in the complement of $F$ are independent of each other. That is, $\cM(W/F) = \bigoplus_{i = 1}^d \cM(W/F)|_{F_i^c} $. To see this, note the set of propagators in the contracted diagram, $W/F$, is exactly $Q_F$, with the capacity of $q$ in $W/F$ being $\tilde{c}(q) = c(q) - c_F(q) > 0$, since if $q \not \in Q_F$ then the capacity is reduced to $0$ in $W/F$, which is equivalent to removing the propagator and some vertices from the diagram. 

For any two sets, $F_j^c$ and $F_l^c$, any propagators supported by both are exactly those where $c_F(p) = c(p)$. However, these propagators are removed in the diagram $W/F$. Therefore, by Lemma \ref{res:disjointprops}, $\cM(W/F)|_{(F_j^c \cup F_l^c)} = \cM(W/F)_{F_j^c} \oplus \cM(W/F)_{F_l^c}$. Since the $F_j^c$ are pairwise independent in $\cM(W/F)$, we have that \bas \cM(W/F) = \bigoplus_{i = 1}^d \cM(W/F)|_{F_i^c} \eas as desired.

Let $r_i =  \rk(\cM(W/F)|_{F_i^c})$. We may write that $\rk(\cM(W)) = \rk_{\cM(W)}(F) + \sum_{i=1}^d r_i$. Furthermore, we may think of the $r_i$ as the ``excess rank'' of $F_i^c$ compared to $F$. That is, $r_i = \rk_{\cM(W)}(F \cup F_i^c) - \rk_{\cM(W)}(F)$. In other words, each $I_a \in \cI_W$ has at least $r_i$ elements in $F_i^c$. On the other hand $B$ contains at least $|U| = \rk(F)+1$ elements of of $F$ and so contains at most $k-\rk(F)-1$ elements of $F^c$. Therefore at least one of the intervals of $F^c$, contains fewer than $r_i$ elements of $B$. Let this interval be $F_i^c$.

Consider now $I_{a_{i+1}}$. The set $B$ has fewer elements in $[b_i+1, a_{i+1}-1]$ than $I_{a_{i+1}}$ does. Let $\ell$ and $m$ be the number of these elements respectively with $\ell<m$. These are the last $\ell$ and $m$ elements in $B$ and $I_{a_{i+1}}$ respectively in the $<_{a_{i+1}}$ order.
Therefore, in the notation of  \ref{dfn:Gale order}, $B^{(k-m +1)} <_{a_{i+1}} b_i+1 \leq_{a_{i+1}} I_{a_{i+1}}^{(k-m +1)}$. Therefore $I_{a_{i+1}} \not\leq_{a_{i+1}} B$ in the Gale ordering.
\end{proof}

\begin{lem}\label{res:positroid implies cond5}
    With notation and assumptions as in Theorem~\ref{res:posivityonflats}, assume that condition \eqref{eq:supportonlyinF} does not hold for some $F\in \cF$ and $q\in Q_F$. Then $\cM(W)$ is not a  positroid. 
\end{lem} 

\begin{proof}
    As in the proof of Lemma~\ref{res:cond5 implies positroid}, write $\cM(W) = ([n], \cB)$, with $\cB$ the set of bases. Write $\rk(\cM(W)) = k = c(\cP)$. Let $\cI_W$ be the Grassmann necklace of $\cM(W)$, and $\cB_{\cI_W}$ the set of bases defined by $\cI_W$. Recall again that $\cM(W)$ is not a positroid if and only if $\cB_{\cI_W} \setminus \cB$ is not empty.

    Now we proceed with the proof. Concretely, we construct a dependent set $D$ of size $k$ such that $I_a \leq_a D$. In other words, the set $D$ is a dependent set in $\cM(W) = ([n], \cB)$ and is also in the set of bases defined by the Grassmann necklace $\cI_W$, or equivalently that $\cB_{\cI_W} \supsetneq \cB$, implying that $\cM(W)$ is not a positroid.
    
    Take $F\in \cF$ and $q\in Q_F$ be such that \eqref{eq:supportonlyinF} does not hold. Then $F$ must consist of at least two disjoint cyclic intervals and $q$ must be supported on non-co-loops on at least two disjoint cyclic intervals of $F^c$.
    
    By Lemma~\ref{res:non min and co-loops} if $q_k$ is an end of $q$ supported on a co-loop $x_k$, then $q_k$ must be matched to that co-loop in any independent matching. Therefore, we do not change the matroid by replacing $q$ by $q-q_k$ with capacity decreased by $1$ along with a single-ended propagator at $x_k$. This also does not change that the generalized Wilson loop diagram is capacity-ranked and locally minimal, or that $q$ has more ends than propagators, nor does it change condition \eqref{eq:supportonlyinF}, so without loss of generality we may suppose $q$ is not supported on any co-loops.

    Choose a circuit $C\subseteq F$ with the property that there exists $y_1,y_2\in C$ such that $q$ has at least one end supported on $F^c\cap [y_1+1, y_2-1]$ and $F^c\cap [y_2+1, y_1-1]$. Such a $C$ must exist since, for any two ends $q_1$ and $q_2$ of $q$ supported on disjoint cyclic intervals of $F^c$, we have that $q_1$ and $q_2$ partition $F$ into two parts: the part, $F_1$, from $q_1$ to $q_2$ and the part, $F_2$, from $q_2$ to $q_1$. Since $F$ is connected there is a circuit $C$ such that $C$ has elements in both $F_1$ and $F_2$ (otherwise every circuit in $F$ lies solely in $F_1$ or solely in $F_2$ and so $F_1$ and $F_2$ are disconnected).

    Note that $C\setminus y_1$ and $C\setminus y_2$ are both independent.
    By the independence exchange axiom, Equation \eqref{eq:indepexch}, we can extend $C\setminus y_1$ to a set $A$ with $\cl(A)=F$.
    
     We want to build a basis $B$ of $\cM(W)$ and an independent matching $M$ on $G^\cE_W$ corresponding to $B$ with the following properties:
     \begin{itemize}
     \item $A\subseteq B$,
     \item there is one end $q_1$ of $q$ outside $F$ which is matched to a vertex $x_1$ by $M$,
     \item there is another end $q_2$ of $q$ outside of $F$ that is not matched by $M$ and which has a supporting vertex $x_2$ that is not matched by $M$, \item one of $q_1$ and $q_2$ is supported in $F^c_1 = F^c\cap [y_1+1, y_2-1]$ and the other is supported in $F^c_2 = [y_2+1, y_1-1]$.
     \end{itemize}
     If $B$ has these properties then say $B$ is \emph{well-configured}.

    For any well-configured $B$, observe that by Lemma~\ref{res:contraction matroid} and Lemma~\ref{res:q ends in or out} since $B$ is full rank inside $F$, every propagator with ends both inside and outside of $F$ will have all its ends inside $F$ matched by $M$ to vertices inside $F$. Thus $A\subseteq B$ can be replaced by any other independent set with closure $F$ without changing that $B$ is a basis and without changing the part of $M$ matching vertices outside of $F$. Consequently $(B\setminus y_1) \cup y_2$ is also a well-configured basis with a matching that agrees with $M$ outside of $F$.
    
    Now we will build a well-configured $B$. To start take a basis $B$ of $\cM(W)$ which extends $A$ and an independent matching $M$ on $G^\cE_W$ corresponding to $B$. Since $W$ is capacity-ranked, not all ends of $q$ outside of $F$ are matched in $M$. By hypothesis there are ends of $q$ in $F^c_1$ and in $F^c_2$.
    If $B$ is well-configured then we are done. If not then since $q$ does have ends in both $F^c_1$ and $F^c_2$ and at least one of these ends is matched and at least one not matched, then there must still be a matched end $q_i$ and an unmatched end $q_j$ with one of $q_i, q_j$ in $F^c_1$ and the other in $F^c_2$. But since $B$ is not well-configured, all the supporting vertices of $q_j$ must be matched by $M$. In this situation call $q_j$ the blocked end and $q_i$ the opposite end. Let $x_j$ be a supporting vertex of $q_j$; by hypothesis $x_j\in B$. 

    Let us take stock of what we know about the structure of $W$, $q$, $B$ and $M$: $W$ is capacity-ranked, $A\subseteq B$, $q$ is not supported on any co-loops, $M$ is a maximum independent matching of $G^\cE_W$, $M$ agrees with $A$ on $F$, and there is a blocked end and an opposite end as described above.

    Let $W'$ be the generalized Wilson loop diagram resulting from removing the blocked end $q_j$ while keeping the capacity of $q$ unchanged and let $G'$ the corresponding bipartite graph. $M$ is still a maximum independent matching of $G'$ and so $\cM(W')$ and $\rk_{\cM(W')}(F)$ are both unchanged from $W$ and $W'$ is still capacity-ranked. Furthermore if $x$ is not a co-loop in $\cM(W)$ then it is still not a co-loop in $\cM(W')$ because given any circuit $C$ that $x$ belongs to in $\cM(W)$, the neighbors of $C$ in $G'$ have weakly smaller rank than in $G^\cE_W$ and hence $C$ remains a circuit in $\cM(W')$. Therefore, using basis exchange inside $\cM(W')$ we can remove $x_j$ from $B$ obtaining a new basis $B'$. We still have $A\subseteq B'$. We know there is an independent matching $M'$ in $G'$ associated to $B'$ and since $B$ was full rank on $F$, we know that the new element added to $B'$ was not in $F$.
    
    If $B'$ is well-configured as a basis of $W$ then we are done. If not, then given that $M'$ does not match $x_j$ or $q_j$ but by capacity-rankedness must match some end of $q$ outside of $F$, then it must be the case that the end of $q$ outside of $F$ that it matches, call it $q_i'$, is in the same $F^c_i$ as $q_j$, and all the ends of $q$ in other $F^c_j$ (of which there is at least one call it $q_j'$) are unmatched but have all their supporting vertices matched.

   Therefore with $q_j'$ as the blocked end and $q_i'$ as the opposite end we are in the same situation as before and with the same overall rank and with $F^c_1$ and $F^c_2$ both each containing ends of $q$. Continue likewise by removing $q_j'$. Following this procedure at each step we either obtain a $B$ which is well-configured as a basis of $W$, or we obtain new blocked and opposite ends while maintaining he same overall rank. Since $q$ has finitely many ends, this process must eventually result in a well-configured basis $B$ of $W$.

   Let $B_1$ now be a well-configured basis of $W$ with $x_1$, $x_2$, $q_1$, $q_2$ as in the definition of well-configured. Without loss of generality say $x_1\in F^c_2$ and $x_2\in F^c_1$. By construction $B_1' = (B\setminus x_1) \cup x_2$ is also a basis and has a corresponding matching obtained by removing the edge between $x_1$ and $q_1$ in $M$ and replacing it by the edge between $x_2$ and $q_2$. As observed above $B_2 = (B\setminus y_1) \cup y_2$ is also well-configured with the same matching on $F^c$ and so $B_2' = (B\setminus \{y_1, x_1\}) \cup \{y_2, x_2\}$ is also a basis for the same reason as $B_1'$.

   Using this notation, we build a set $D$ that is dependent in $\cM(W)$ and show that it is an independent set in $\cB_{\cI_W}$, showing that $\cM(W)$ is not a positroid. Let $D=(B\setminus x_1)\cup y_2$. Then $C\subseteq D$ so $D$ is dependent in $\cM(W)$. 

    The rest of the proof follows from the observation that for any independent set $B \in \cM(W)$, $I_\alpha <_\alpha B$ in the Gale ordering for all $\alpha \in [n]$. Note that in the cyclically ordered set $[n]$, we have that $x_i < y_i < x_j < y_j$. Therefore, for any $\alpha \in [n]$, we may find a valid situation where a $y$ vertex comes after and $x$ vertex. Therefore, we have that \bas I_\alpha <_\alpha \begin{cases} B_i <_\alpha D = (B_1 \setminus x_1) \cup y_2 & \textrm{if } x_1 <_\alpha y_1 \\    
	B'_2 <_\alpha D = (B'_2 \setminus x_2) \cup y_1 & \textrm{if } x_1 <_\alpha y_2 \\  
\end{cases} \;. \eas In the first case, $\alpha \in [y_2+1 , x_1]$. In the second case, $\alpha \in [y_1+1 , x_2]$, where, by construction, $y_1 < x_1 < y_2 <x_2$. Therefore, every $\alpha \in [n]$ is covered by one of the cases. Note that since in every case, one removes a smaller element in the appropriate linear order with a larger element, the Gale ordering holds. 

Therefore, we have found a dependent set $D$ that is in the set of bases defined by the Grassmann necklace of $\cM(W)$, $\cB_{\cI_W}$, making $\cM(W)$ not a positroid. 
   \end{proof}

\begin{proof}[Proof of Theorem~\ref{res:posivityonflats}]
The proof of Theorem~\ref{res:posivityonflats} then follows from Lemmas~\ref{res:cond5 implies positroid} and \ref{res:positroid implies cond5}.  

Note that the lemmas vacuously apply even when $\cF$ is empty, however in this case there is also an elementary argument: if $\cF$ is empty, every connected dependent flat of $\cM(W)$ has full rank, so by Lemma \ref{res:flatcondpos}, $\cM(W)$ is uniform, and thus a positroid.
\end{proof}

\begin{eg}\label{eg:exampleFQ_F}
Let $W$ be the generalized Wilson Loop diagram from Example \ref{eg:working example}. Note that it is not locally minimal. Namely, one can remove the end of $q$ on the edge $e_6$ without changing the independence structure of $\cM(W)$ (see Example~\ref{eg: half props clear}). One may check that the resulting diagram after making this change is locally minimal. 
	
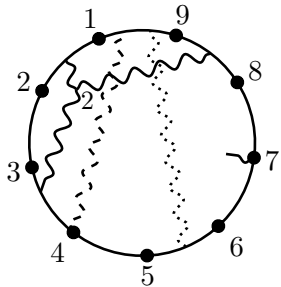
\[\begin{tikzpicture}[rotate=-67.5, line width=1, scale=1.5]
		\def \n {9}
		\draw circle(1)
		\foreach \v in{1,...,\n}
		{(360*\v/\n-360/\n+180:1)circle(.4pt)circle(.8pt)circle(1.2pt)circle(1.4pt) node[anchor=360/\n*\v-360/\n-67.5]{$\v$}};
		\foreach \x/\y in {-0.940/-0.342,0.058/-0.998,-0.500/0.866}{\draw[decorate,decoration={snake,amplitude=0.8mm}] (-0.661,-0.358) -- (\x,\y);}
		\draw (-0.661,-0.358) node[shift = {(.15,-.15)}] {\small2};
		\draw[dashed, decorate,decoration={snake,amplitude=0.8mm}] (0.500,-0.866) -- (-0.973,0.231);
		\draw[dotted, decorate,decoration={snake,amplitude=0.8mm}] (1.0,0) -- (-0.894,0.449);
		\draw[decorate,decoration={snake,amplitude=0.8mm}] (0.500,0.866) -- ++({atan(1.73)}:-.25 cm);
	\end{tikzpicture}\]
	One may further check that the associated matroid structure remains unchanged.

We present a few flats and sets of propagators $Q_F$ for the modified diagram above, with $p$ and $s$ draw with solid lines ($s$ being a single-ended propagator), $q$ with a dashed line, and $r$ with a dotted line. One may also check that $W$ is capacity-ranked.

Let $F = \{4,5,6,8,9\}$. Note that these vertices support the propagators $\Prop(F) = \{p, q, r\}$. However, notice that $F$ only supports one end of $p$. In fact, $0< c_F(p) < c(p)$. First, we show that $F$ is a flat by checking that $\rk(F \cup x) > \rk(F)$ for every $x \in F^c$. The element $7$ supports only $s$, and is the only support of $s$. Since $W$ is capacity-ranked, $7$ is a co-loop, and is independent of $F$. Similarly, the vertices in $[1,3]$ all support an end of $q$ that is distinct from the end supported in $F$. Therefore, $\rk(F\cup x) = \rk(F) + 1$. Note that the subset $\{5, 6\}$ is a circuit, so $F$ is a dependent flat. It remains to check that $F$ is connected, which can be done by exhaustively checking that no subset $S \subsetneq F$ of $S$ does not satisfy $\rk(S) + \rk(F \setminus S) = \rk(F)$. Furthermore, $F$ contains none of the vertices supporting $s$ ($c_F(s) = 0$). Therefore, $Q_F = \{p, s\}$. The support of $s$ is in a single cyclic interval in the complement of $F$, and the support of the two ends of $p$ not supported by $F$ are in another. 

Let $F' = \{ 2, 3, 4, 8\}$ which is a flat supporting the propagators $\Prop(F') = \{p, q\}$. To check that $F'$ is a flat, we check that adding any vertex outside of $F'$ to $F'$ increases its rank. As before $7$ is a co-loop, and therefore is independent of $F'$. The vertices $\{1, 5, 6, 9\}$ support $r$, which is not in $\Prop(F')$. Therefore, these vertices are independent of $F'$. The propagators $r$ and $s$ are not at all supported by these vertices, and therefore $Q_{F'} = \{ r, s\}$. In this case, there are two connected components of $F' = \{[2,4], \{8\}\}$. However, note that any cyclic interval containing $F'$ contains vertices of $r$ outside of $F'$. For example, consider the intervals $A = [8, 4]$ and $B = [2,8]$. The vertex set of $q$, $V(q) = \{1, 5, 6, 9\}$ intersects both $A \setminus F = \{ 1, 9\}$ and $B \setminus F = \{ 5, 6\}$. Therefore, by Theorem \ref{res:posivityonflats}, this diagram does not correspond to a positroid. 

Finally, to illustrate the argument given in the proof of Lemma \ref{res:positroid implies cond5}, when condition \eqref{eq:supportonlyinF} does not hold, let $F'_1 = [2,4]$ and $F'_2 = \{8\}$. Then $v = 5 = x$ and $y = 3$, with $I_5 = 57891$ and $B = 78913$. For all $\alpha$ where $x <_\alpha y$, i.e.\ for $\alpha \in [4,5]$, one can check that $I_\alpha <_\alpha I_v <_\alpha B$. The possibilities for $z$ are $\{1, 9\}$. For all $\alpha$ where $\alpha <_\alpha z <_\alpha y <_\alpha x$, i.e.\ $\alpha \in [6, 1]$, $I_\alpha^{(i)} \leq_\alpha I_v^{(i)}$ while $I_v^{(i)} \leq_\alpha 1$. The next element of $I_\alpha$ is $z$. The insertion of this index allows the subsequent indices in $I_\alpha$ to not be greater than the corresponding indices of $B$. Finally, for $\alpha \in [2, 3]$, where $z <_\alpha \alpha <_\alpha y <_\alpha x$, $I_\alpha$ $I_\alpha^{(i)} \leq_\alpha I_v^{(i)}$ while $I_v^{(i)} \leq_\alpha 3$. The next index in $I_\alpha$ is then $y$, which allows the subsequent indices in $I_\alpha$ to not be greater than the corresponding indices of $B$.
\end{eg}

We can use this to give useful conditions on which generalized Wilson loop diagrams give rise to positroids.

\begin{cor} \label{res:non-crosspos}
	If $W$ is a capacity-ranked generalized Wilson loop diagram with no crossing propagators, then $\cM(W)$ is a positroid.
\end{cor}
\begin{proof}
    By Lemma \ref{res:ignorenon-supporting}, if $W$ has a non-empty set of vertices, $V$, that do not support any propagators, we may write $\cM(W) = \cM(W)|_V \oplus \cM(W)|_{V^c}$. Therefore, without loss of generality, we may pass to $W|_{V^c}$, where all the vertices support propagators, and call this $W$. 

    Suppose that $\cM(W)$ is not a positroid, so by Lemma~\ref{res:cond5 implies positroid} there is some $F\in \cF$ and $q\in Q_F$ which do not satisfy condition \eqref{eq:supportonlyinF}. Therefore $q$ has at least two ends $a_1$ and $a_2$ supported in different cyclic intervals of the complement of $F$. By Lemma~\ref{res:q ends in or out} $a_1$ and $a_2$ must not also be supported in $F$. Therefore $F$ is partitioned into two disjoint parts, the part strictly between $a_1$ and $a_2$ in the cyclic order and the part strictly between $a_2$ and $a_1$ in the cyclic order. Since $F$ is connected, some propagator $p$ must have an end in each of these two parts and we can choose $p\neq q$ since $q$ restricted to $F$ is independent so does not lead to the connectivity of $F$. Let $b_1$ and $b_2$ be these two ends. Then $a_1, b_1, a_2, b_2$ give a crossing of $p$ and $q$.

    Taking the contrapositive, if $W$ has no crossing propagators then $\cM(W)$ is a positroid.

\end{proof}

The converse is partially true.

\begin{cor} \label{res:nonposflatcross}
	Let $W$ be a locally minimal capacity-ranked generalized Wilson loop diagram. Suppose $W$ has two crossing propagators, $p$ and $q$. Suppose further that there is $F\in\cF$ such that $q \in Q_F$ but $p \not \in Q_F$. Specifically, let $q$ have two ends, $q_1$ and $q_2$ not in $\End(F)$. If $\End(F)$ contains an end of $p$ that falls between $q_1$ and $q_2$ in the cyclic order on ends (Definition \ref{dfn:global cyclic end order}) and another end of $p$ falls between $q_2$ and $q_1$, then $\cM(W)$ is not a positroid. 
\end{cor} 
\begin{proof}
	
	Since $\End(F)$ contains ends of $p$ both between the ends $q_1$ and $q_2$ and between $q_2$ and $q_1$, any interval $[a,b]$ containing $F$ will contain at least one vertex supporting either $q_1$ or $q_2$. However, by construction, these vertices are not in $F$. Therefore, $\cM(W)$ is not a positroid.
\end{proof}

\begin{eg}\label{eg:nonposflatcross}
	A simple example of the statement of Corollary \ref{res:nonposflatcross} is when all ends of $p$ are in $\cF$. Then one is guaranteed that two cyclic intervals containing vertices supporting $q$ are in any cyclic interval containing $F$. For example, consider the following diagram: 
	
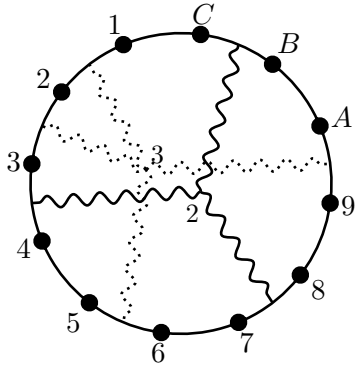
\begin{center}
	\begin{tikzpicture}[rotate=-67.5, line width=1, scale=2]
		\def \n {12}
		\draw circle(1)
		\foreach[count = \vi] \v in{1,...,9, A, B, C}
		{(360*\vi/\n-360/\n+180:1)circle(.4pt)circle(.8pt)circle(1.2pt)circle(1.4pt) node[anchor=360/\n*\vi-360/\n-67.5]{$\v$}};
		\foreach \x/\y in {-0.966/-0.259,-0.707/-0.707,0.707/-0.707,0.259/0.966}{\draw[dotted, decorate,decoration={snake,amplitude=0.8mm}] (-0.177,-0.177) -- (\x,\y);}   
		\draw (-0.177,-0.177) node[shift = {(.15,.2)}] {\small3};
		\foreach \x/\y in {-0.259/-0.966,0.966/0.259,-0.707/0.707}{\draw[decorate,decoration={snake,amplitude=0.8mm}] (.1,.1) -- (\x,\y);}
		\draw (0,0) node[shift = {(.15,-.4)}] {\small2};
	\end{tikzpicture}
	\end{center}
	
	with $p$ the propagator with $3$ ends and $q$ the propagator with $4$. Let $F = \{3, 4, 7, 8, B, C\}$. In this case $V(p) = F$ and $q \in Q_{F}$. Any cyclic interval containing $F$ contains either the vertices $\{1,2\}$, $\{5,6\}$ or $\{9, A\}$ all of which support ends of $q$ and are not in $F$. 
	
\end{eg}
	
Note that this does not preclude diagrams with crossing propagators from corresponding to positroids. Indeed, in the case of ordinary Wilson loop diagrams, if there was subdiagram with crossing propagators that corresponded to a uniform matroid, then the original Wilson loop diagram would still correspond to a positroid \cite{generalcombinatoricsI}. While a complete characterization of this in the case of generalized Wilson loop diagram is difficult, we give a few examples of diagrams with crossing propagators which still define a positroid.

\begin{eg}\label{eg:crossing_pos} 
For example, we give a diagram for which the condition of Corollary \ref{res:non-crosspos} is not satisfied, that is for every connected cyclic flat, $F$, the propagators in $Q_F$ have ends supported in exactly one cyclic interval in the complement. Furthermore, one may check by hand that this diagram gives rise to a positroid. The diagram in question is the following:
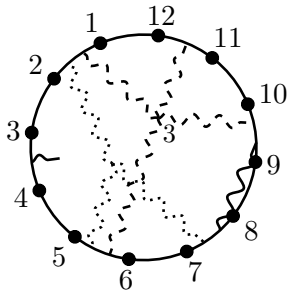
\begin{align*}
\begin{tikzpicture}[rotate=-67.5, line width=1, scale=1.5]
	\def \n {12}
	\draw circle(1)
	\foreach \v in{1,...,\n}
	{(360*\v/\n-360/\n+180:1)circle(.4pt)circle(.8pt)circle(1.2pt)circle(1.4pt) node[anchor=360/\n*\v-360/\n-67.5]{$\v$}};
	\foreach \x/\y in {-0.985/-0.174,0.766/-0.643,0.174/0.985,-0.707/0.707}{\draw[dashed, decorate,decoration={snake,amplitude=0.8mm}] (-0.188,0.219) -- (\x,\y);}
	\draw (-0.188,0.219) node[shift = {(.15,-.2)}] {\small3};
	\foreach \x/\y in {-0.940/-0.342,0.643/-0.766,0.985/0.174}{\draw[dotted, decorate,decoration={snake,amplitude=0.8mm}] (0.029,-0.311) -- (\x,\y);}
	\draw[decorate,decoration={snake,amplitude=0.8mm}] (-0.259,-0.966) -- ++({atan(3.73)}:.25 cm);
	\draw[decorate,decoration={snake,amplitude=0.8mm}] (0.940,0.342) -- (0.342,0.940);
\end{tikzpicture}
\end{align*}
The only cyclic flat that is not a cyclic interval is $F = \{ 1,2,5, 6\}$. Let $q$ be the propagator indicated by the dashed line and $p$ the propagator indicated by a dotted line. Then $q \in Q_F$, and $F$ supports two ends of $p$. However, all the ends of $q$ not supported by vertices in $F$ are in one cyclic interval in the complement of $F$, namely, $[7,12]$. Therefore, by Lemma \ref{res:cond5 implies positroid}, $\cM(W)$ is a positroid, even though $W$ has crossing propagators.
\end{eg}

We conclude by noting that there is a recent closely related and independently derived result characterizing when a matroid is a positroid. We restate it here in slightly different language than in the original.

\begin{thm} [Theorem 3.1, \cite{bonin23}]\label{res:Boninposcharacterization}
	Let $M$ be a matroid on the cyclically ordered ground set $[n]$ with no loops. Then $M$ is a positroid if and only if, for each proper connected flat $F$ with $|F| \geq 2$ and each connected component $K$ of the contraction $M/F$, the ground set of $K$ is a subset of a cyclic interval that is disjoint from $F$.
\end{thm}

We bring the reader's attention to this result because it helps illuminate the otherwise abstruse condition in the definition of $Q_F$.

\begin{rmk} 
	Theorem \ref{res:posivityonflats} is a diagrammatic version of Theorem \ref{res:Boninposcharacterization} for generalized Wilson loop diagrams. To see this, first note that in both statements of the two theorems, one restricts to matroids without loops, or, equivalently by Lemma \ref{res:ignorenon-supporting}, generalized Wilson loop diagrams where every vertex supports a propagator.
	
	Secondly, note that the condition in Theorem \ref{res:Boninposcharacterization} that $F$ be a proper connected flat with $|F| \geq 2$ forces $F$ to be a proper connected cyclic flat, i.e.\ $\cF$ in Theorem \ref{res:posivityonflats}. Specifically, if $|F| = 1$ then it is an independent flat, and thus excluded from consideration in both cases.
	
	Finally, in both theorems, one is interested in certain subsets of the ground set of $M$ to be contained in a single cyclic interval disjoint from $F$. In Theorem \ref{res:Boninposcharacterization}, these sets are the ground sets of the connected components of $M/F$. In Theorem~\ref{res:posivityonflats} these are the supports of the vertices of $q \in Q_F$ that lie outside of $F$. However, by Definition \ref{dfn:contraction graph} note that the propagator set $Q_F$ is exactly those propagators remaining in the diagram corresponding to the contracting by $F$: $W/F$. Since $F$ is a connected cyclic flat, Lemma \ref{res:cyclic flat prop ranked} shows that $c_F(\cP) = \rk_{\cM(W)}(F)$ and Lemma \ref{res:contraction matroid} shows that the corresponding matroid structure is indeed the contraction $\cM(W/F) = \cM(W)/F$.
	
	The connected components of $\cM(W/F)$ can either be single vertices or connected cyclic flats. Single vertices by default are in a single cyclic interval in the complement of $F$, and Theorem \ref{res:posivityonflats} demands that the connected cyclic flat also be in a single cyclic interval in the complement of $F$.
\end{rmk}

\section{Diagrammatic boundary moves}
\label{sec:diagrammaticmoves} 

As discussed in Sections \ref{sec:SYM} and \ref{sec:positroid} we have both mathematical and physical motivations for understanding the boundaries of positroid cells associated to admissible Wilson loop diagrams. It was generally understood by practitioners in the area that some boundaries could be obtained from the diagram by sliding an end of a propagator to an adjacent vertex, or by sliding two adjacent ends of two propagators together \cite{Wilsonloop, Amplituhedronsquared, HeslopStewart}. However, these moves are far from giving even all codimension one boundaries \cite{casestudy, cancellation}, and they could not be readily iterated since they moved outside the framework of admissible ordinary Wilson loop diagrams.

In this section, we advance the understanding of boundaries of Wilson loop diagrams by giving a diagrammatic calculus to compute boundaries of generalized non-crossing Wilson loop diagrams.  Our main results are that our boundary calculus is sufficiently powerful to, when iterated, obtain all boundaries of all codimensions of admissible ordinary Wilson loop diagrams with two propagators (Theorem \ref{thm:boundarymovesgenerate}), and to obtain all boundaries of codimension one of admissible ordinary Wilson loop diagrams with three propagators (Section \ref{sec:boundaries of more than 2 props}).  This vastly exceeds what was previously known.  

Having suitable objects on which to represent boundaries and on which to iterate boundary moves was our original motivation for introducing generalized Wilson loop diagrams. However, the boundary calculus we have developed so far remains insufficient for obtaining all boundaries of all codimensions even of admissible ordinary Wilson loop diagrams, and we discuss potential extensions in Remark \ref{rmk:matroids on props def} and Examples \ref{eg:missingcodim1boundary1} and \ref{eg:missingcodim1boundary2}, with further consideration of ways forward in the conclusion.

We begin by giving the family of diagrammatic moves that form our diagrammatic boundary calculus. 

To be precise, we have 5 classes of moves, which we call the slide move, the merge move, the merge and split move, the X-merge and split move, and the retraction move. These are laid out in Section \ref{sec:moves}. In Section \ref{sec:movesrealizeboundaries}, we show that if we start with any generalized Wilson loop diagram with non-crossing propagators, then the moves will result in a boundary of the positroid one began with. In Section \ref{sec:diagrammaticboundaries}, we prove that for any admissible ordinary Wilson loop diagram with two propagators, our boundary calculus gives all the boundaries of all codimensions.

Table \ref{table:moves} summarizes the different types of moves, the collective application of which we refer to as the graphical calculus.

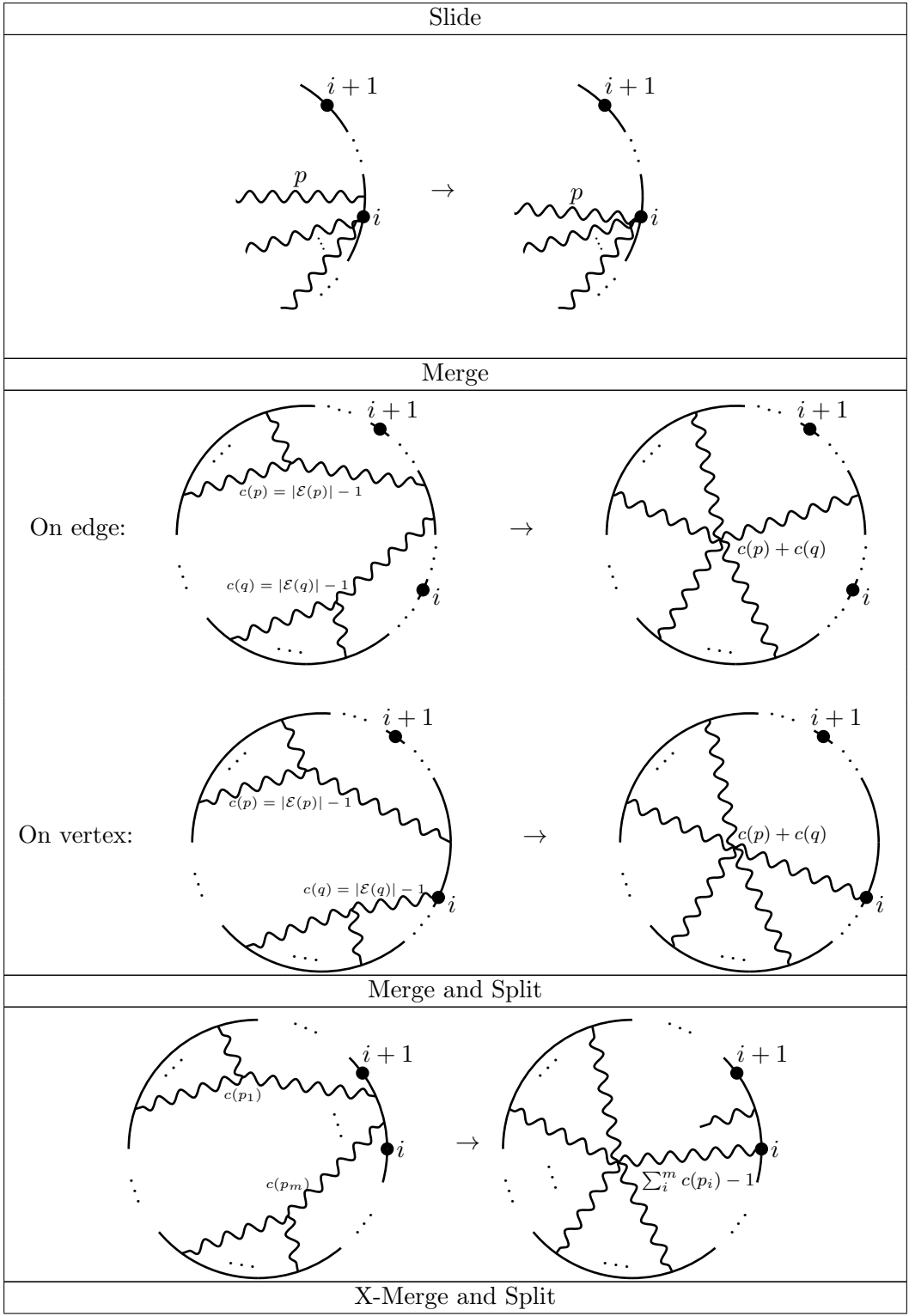
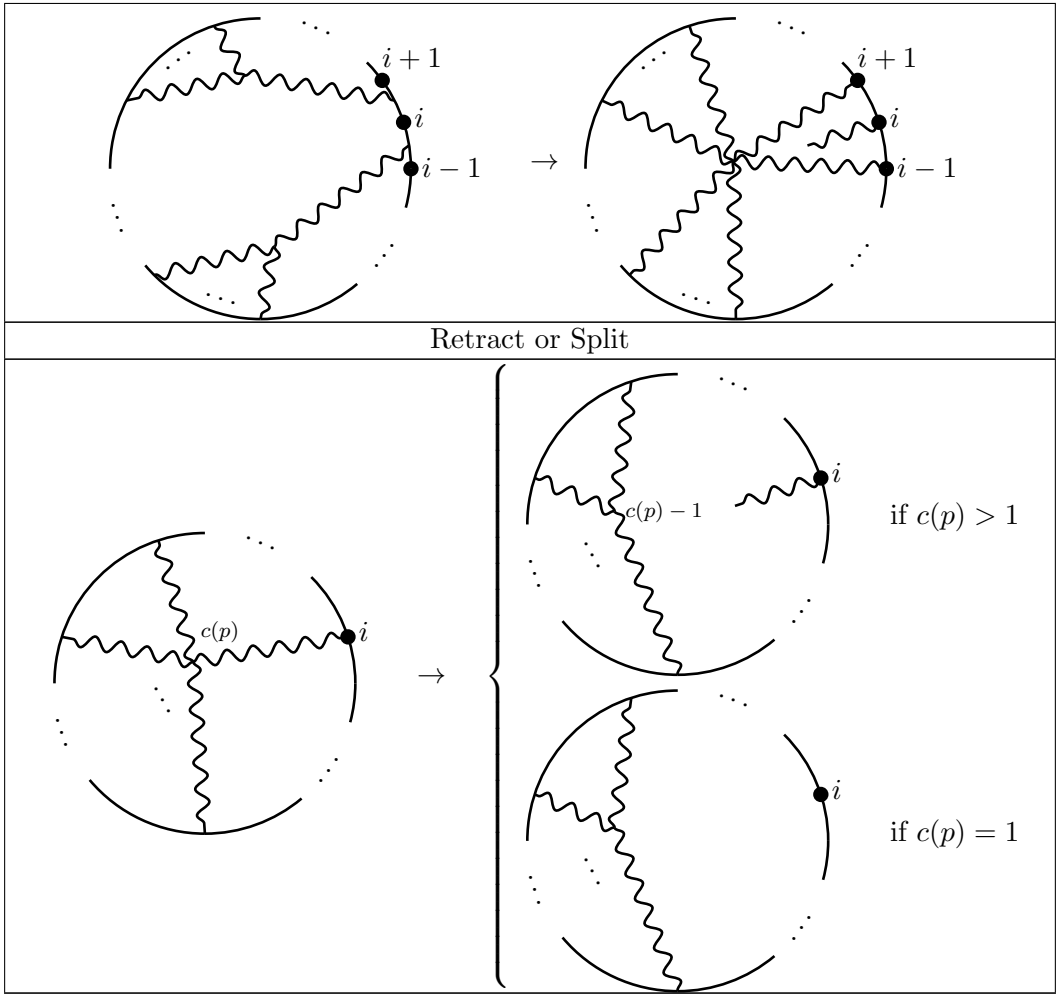
\begin{longtable}{|c|}
	\caption{Glossary of boundary moves	\label{table:moves}} \\
	
	\hline
	Slide \\
	\hline
	\begin{tikzpicture}[line width=1, scale=2,baseline=0]
		\clip (1.25,0)+(90:1.25cm) arc[start angle = 90, end angle =270,radius=1.25cm] --(1.25,0);
		\draw (0,0)+(30:1cm) arc[start angle = 30, end angle =60,radius=1cm];
		\draw (1,0) arc (0:10:1);
		\draw (1,0) arc (0:-30:1);
		\def \n {10}
		\draw foreach \v in{5.75,7.25}{(360*\v/\n-360/\n+180:1)circle(.4pt)circle(.8pt)circle(1.2pt)node[anchor=360/\n*\v-360/\n] {}};
		\draw (360*5.75/\n-360/\n+180:1) node[anchor=360/\n*6-360/\n]{$i$};
		\draw (360*7.25/\n-360/\n+180:1) node[anchor=360/\n*7-360/\n]{$i+1$};
		\draw (360*6.55/\n-360/\n+180:1) node[rotate=360*6.55/\n-360/\n+90] {$\cdots$};
		\draw (360*8/\n-360/\n+180:1) node[rotate=360*8/\n-360/\n+90] {$\cdots$};
		\draw (360*4.75/\n-360/\n+180:1) node[rotate=360*4.75/\n-360/\n+90] {$\cdots$};
		\draw (.25,.45)+(360*4.25/\n-360/\n+180:.9) node[rotate=360*4.25/\n-360/\n] {$\scriptstyle \cdots$};
		\draw[decorate,decoration={snake}] (-1,0) -- (360*6/\n-360/\n+180:1) node[pos=.75,above]{$p$};
		\draw[decorate,decoration={snake}] (-.776,-.643) -- (.985,-.174);
		\draw[decorate,decoration={snake}] (.342,-.940) -- (.985,-.174);
	\end{tikzpicture}
	\quad $\rightarrow$ \qquad
	\begin{tikzpicture}[line width=1, scale=2,baseline=0]
		\clip (1.25,0)+(90:1.25cm) arc[start angle = 90, end angle =270,radius=1.25cm] --(1.25,0);
		\draw (0,0)+(30:1cm) arc[start angle = 30, end angle =60,radius=1cm];
		\draw (1,0) arc (0:10:1);
		\draw (1,0) arc (0:-30:1);
		\def \n {10}
		\draw foreach \v in{5.75,7.25}{(360*\v/\n-360/\n+180:1)circle(.4pt)circle(.8pt)circle(1.2pt)node[anchor=360/\n*\v-360/\n] {}};
		\draw (360*5.75/\n-360/\n+180:1) node[anchor=360/\n*6-360/\n]{$i$};
		\draw (360*7.25/\n-360/\n+180:1) node[anchor=360/\n*7-360/\n]{$i+1$};
		\draw (360*6.55/\n-360/\n+180:1) node[rotate=360*6.55/\n-360/\n+90] {$\cdots$};
		\draw (360*8/\n-360/\n+180:1) node[rotate=360*8/\n-360/\n+90] {$\cdots$};
		\draw (360*4.75/\n-360/\n+180:1) node[rotate=360*4.75/\n-360/\n+90] {$\cdots$};
		\draw (.25,.45)+(360*4.25/\n-360/\n+180:.9) node[rotate=360*4.25/\n-360/\n] {$\scriptstyle \cdots$};
		\draw[decorate,decoration={snake}] (-1,0) -- (.985,-.174) node[pos=.75,above]{$p$};
		\draw[decorate,decoration={snake}] (-.776,-.643) -- (.985,-.174);
		\draw[decorate,decoration={snake}] (.342,-.940) -- (.985,-.174);
	\end{tikzpicture} \\
	\hline
	Merge \\
	\hline
	On edge: \quad \begin{tikzpicture}[line width=1, scale=2,baseline=0]
		\draw (-1,0) arc (180:86:1);
		\draw (1,0) arc (0:60:1);
		\draw (1,0) arc (0:-30:1);
		\draw (0,-1) arc (270:310:1);
		\draw (0,-1) arc (270:220:1);
		
		\def \n {10}
		\draw foreach \v in{5.3,7.53}{(360*\v/\n-360/\n+180:1)circle(.4pt)circle(.8pt)circle(1.2pt)node[anchor=360/\n*\v-360/\n] {}};
		\draw (360*5.3/\n-360/\n+180:1) node[right, anchor=(360*5.3/\n-360/\n]{$i$};
		\draw (360*7.53/\n-360/\n+180:1) node[anchor=360*7.53/\n-360/\n]{$i+1$};
		\draw (360*1.5/\n-360/\n+180:1) node[rotate=360*1.5/\n-360/\n+90] {$\cdots$};
		\draw (360*8.0625/\n-360/\n+180:1) node[rotate=360*8.0625/\n-360/\n+90] {$\cdots$};
		\draw (360*4.875/\n-360/\n+180:1) node[rotate=360*4.875/\n-360/\n+90] {$\cdots$};
		
		\draw (360*9.75/\n-360/\n+180:.9) node[rotate=360*9.75/\n-360/\n+90] {$\cdots$};
		\draw (360*3.25/\n-360/\n+180:.9) node[rotate=360*3.25/\n-360/\n+90] {$\cdots$};
		\foreach \x/\y in {0.924/0.383,-0.309/0.951,-0.951/0.309}{\draw[decorate,decoration={snake}] (-0.115,0.556) -- (\x,\y);}
		\foreach \x/\y in {-0.588/-0.809,0.309/-0.951,.991/0.131}{\draw[decorate,decoration={snake}] (0.233,-0.517) -- (\x,\y);}
		\draw (-0.115,0.556) node[shift={(.15,-0.45)}] {\tiny$c(p)= |\cE(p)| -1$};
		\draw (0.233,-0.517) node[shift={(-.75,0.25)}] {\tiny$c(q)= |\cE(q)| -1$};
		\draw[white, line width=5] (.866,.5) arc (30:50:1);
		\draw[white, line width=5] (1,0) arc (0:-20:1);
		\draw (360*5.6875/\n-360/\n+180:1) node[rotate=360*5.6875/\n-360/\n+90] {$\cdots$};
		\draw (360*7.0625/\n-360/\n+180:1) node[rotate=360*7.0625/\n-360/\n+90] {$\cdots$};
	\end{tikzpicture}
	\qquad $\rightarrow$ \qquad
	\begin{tikzpicture}[line width=1, scale=2,baseline=0]
		\draw (-1,0) arc (180:86:1);
		\draw (1,0) arc (0:60:1);
		\draw (1,0) arc (0:-30:1);
		\draw (0,-1) arc (270:310:1);
		\draw (0,-1) arc (270:220:1);
		\def \n {10}
		\draw foreach \v in{5.3,7.53}{(360*\v/\n-360/\n+180:1)circle(.4pt)circle(.8pt)circle(1.2pt)node[anchor=360/\n*\v-360/\n] {}};
		\draw (360*5.3/\n-360/\n+180:1) node[right, anchor=(360*5.3/\n-360/\n]{$i$};
		\draw (360*7.53/\n-360/\n+180:1) node[anchor=360*7.53/\n-360/\n]{$i+1$};
		\draw (360*1.5/\n-360/\n+180:1) node[rotate=360*1.5/\n-360/\n+90] {$\cdots$};
		\draw (360*8.0625/\n-360/\n+180:1) node[rotate=360*8.0625/\n-360/\n+90] {$\cdots$};
		\draw (360*4.875/\n-360/\n+180:1) node[rotate=360*4.875/\n-360/\n+90] {$\cdots$};
		\draw (360*9.75/\n-360/\n+180:.9) node[rotate=360*9.75/\n-360/\n+90] {$\cdots$};
		\draw (360*3.25/\n-360/\n+180:.9) node[rotate=360*3.25/\n-360/\n+90] {$\cdots$};
		\foreach \x/\y in {-0.588/-0.809,0.309/-0.951,0.951/0.309,-0.309/0.951,-0.951/0.309}{\draw[decorate,decoration={snake}] (-0.118,-0.038) -- (\x,\y);}
		\draw (-0.118,-0.038) node[shift={(.95,-.15)}] {\scriptsize$c(p)+c(q)$};
		\draw[white, line width=5] (.866,.5) arc (30:50:1);
		\draw[white, line width=5] (1,0) arc (0:-20:1);
		\draw (360*5.6875/\n-360/\n+180:1) node[rotate=360*5.6875/\n-360/\n+90] {$\cdots$};
		\draw (360*7.0625/\n-360/\n+180:1) node[rotate=360*7.0625/\n-360/\n+90] {$\cdots$};
	\end{tikzpicture}\\ \\
	On vertex: $\quad$ \begin{tikzpicture}[line width=1, scale=2,baseline=0]
		\draw (-1,0) arc (180:86:1);
		\draw (1,0) arc (0:60:1);
		\draw (1,0) arc (0:-30:1);
		\draw (0,-1) arc (270:310:1);
		\draw (0,-1) arc (270:220:1);
		\def \n {10}
		\draw foreach \v in{5.3,7.53}{(360*\v/\n-360/\n+180:1)circle(.4pt)circle(.8pt)circle(1.2pt)node[anchor=360/\n*\v-360/\n] {}};
		\draw (360*5.3/\n-360/\n+180:1) node[right, anchor=(360*5.3/\n-360/\n]{$i$};
		\draw (360*7.53/\n-360/\n+180:1) node[anchor=360*7.53/\n-360/\n]{$i+1$};
		\draw (360*1.5/\n-360/\n+180:1) node[rotate=360*1.5/\n-360/\n+90] {$\cdots$};
		\draw (360*8.0625/\n-360/\n+180:1) node[rotate=360*8.0625/\n-360/\n+90] {$\cdots$};
		\draw (360*4.875/\n-360/\n+180:1) node[rotate=360*4.875/\n-360/\n+90] {$\cdots$};
		\draw (360*9.75/\n-360/\n+180:.9) node[rotate=360*9.75/\n-360/\n+90] {$\cdots$};
		\draw (360*3.25/\n-360/\n+180:.9) node[rotate=360*3.25/\n-360/\n+90] {$\cdots$};
		\foreach \x/\y in {1/0,-0.309/0.951,-0.951/0.309}{\draw[decorate,decoration={snake}] (-0.115,0.556) -- (\x,\y);}
		\foreach \x/\y in {-0.588/-0.809,0.309/-0.951,.906/-0.423}{\draw[decorate,decoration={snake}] (0.233,-0.517) -- (\x,\y);}
		\draw (-0.115,0.556) node [shift={(-.25,-0.5)}] {\tiny $c(p)= |\cE(p)| -1$};
		\draw (0.233,-0.517) node [shift={(.2,0.3)}] {\tiny $c(q)= |\cE(q)| -1$};
		\draw[white, line width=5] (.866,.5) arc (30:50:1);
		\draw (360*7.0625/\n-360/\n+180:1) node[rotate=360*7.0625/\n-360/\n+90] {$\cdots$};
	\end{tikzpicture} $\qquad \rightarrow \qquad$
	\begin{tikzpicture}[line width=1, scale=2,baseline=0]
		\draw (-1,0) arc (180:86:1);
		\draw (1,0) arc (0:60:1);
		\draw (1,0) arc (0:-30:1);
		\draw (0,-1) arc (270:310:1);
		\draw (0,-1) arc (270:220:1);
		\def \n {10}
		\draw foreach \v in{5.3,7.53}{(360*\v/\n-360/\n+180:1)circle(.4pt)circle(.8pt)circle(1.2pt)node[anchor=360/\n*\v-360/\n] {}};
		\draw (360*5.3/\n-360/\n+180:1) node[right, anchor=(360*5.3/\n-360/\n]{$i$};
		\draw (360*7.53/\n-360/\n+180:1) node[anchor=360*7.53/\n-360/\n]{$i+1$};
		\draw (360*1.5/\n-360/\n+180:1) node[rotate=360*1.5/\n-360/\n+90] {$\cdots$};
		\draw (360*8.0625/\n-360/\n+180:1) node[rotate=360*8.0625/\n-360/\n+90] {$\cdots$};
		\draw (360*4.875/\n-360/\n+180:1) node[rotate=360*4.875/\n-360/\n+90] {$\cdots$};
		\draw (360*9.75/\n-360/\n+180:.9) node[rotate=360*9.75/\n-360/\n+90] {$\cdots$};
		\draw (360*3.25/\n-360/\n+180:.9) node[rotate=360*3.25/\n-360/\n+90] {$\cdots$};
		\foreach \x/\y in {-0.588/-0.809,0.309/-0.951,.906/-0.423,-0.309/0.951,-0.951/0.309}{\draw[decorate,decoration={snake}] (-0.118,-0.038) -- (\x,\y);}
		\draw (-0.118,-0.038) node[shift={(.75,0.2)}] {\scriptsize$c(p)+c(q)$};
		\draw[white, line width=5] (.866,.5) arc (30:50:1);
		\draw (360*7.0625/\n-360/\n+180:1) node[rotate=360*7.0625/\n-360/\n+90] {$\cdots$};
	\end{tikzpicture}\\
	
	\hline
	Merge and Split \\
	\hline
	\begin{tikzpicture}[line width=1, scale=2,baseline=0]
		\draw (-1,0) arc (180:90:1);
		\draw (1,0) arc (0:45:1);
		\draw (1,0) arc (0:-15:1);
		\draw (0,-1) arc (270:310:1);
		\draw (0,-1) arc (270:220:1);
		\def \n {10}
		\draw foreach \v in{6,7}{(360*\v/\n-360/\n+180:1)circle(.4pt)circle(.8pt)circle(1.2pt)node[anchor=360/\n*\v-360/\n] {}};
		\draw (360*6/\n-360/\n+180:1) node[anchor=360/\n*6-360/\n]{$i$};
		\draw (360*7/\n-360/\n+180:1) node[anchor=360/\n*7-360/\n]{$i+1$};
		\draw (360*1.5/\n-360/\n+180:1) node[rotate=360*1.5/\n-360/\n+90] {$\cdots$};
		\draw (360*1.5/\n-360/\n:.65) node[rotate=360*1.5/\n-360/\n+90] {$\cdots$};	
		\draw (360*7.875/\n-360/\n+180:1) node[rotate=360*7.875/\n-360/\n+90] {$\cdots$};
		\draw (360*5/\n-360/\n+180:1) node[rotate=360*5/\n-360/\n+90] {$\cdots$};
		\draw (360*9.75/\n-360/\n+180:.9) node[rotate=360*9.75/\n-360/\n+90] {$\cdots$};
		\draw (360*3.25/\n-360/\n+180:.9) node[rotate=360*3.25/\n-360/\n+90] {$\cdots$};
		\foreach \x/\y in {0.914/0.407,-0.309/0.951,-0.951/0.309}{\draw[decorate,decoration={snake}] (-0.115,0.556) -- (\x,\y);}
		\foreach \x/\y in {-0.588/-0.809,0.309/-0.951,0.978/0.208}{\draw[decorate,decoration={snake}] (0.233,-0.517) -- (\x,\y);}
		\draw (-0.115,0.54) node [below] {\tiny$c(p_1)$};
		\draw (0.233,-0.4) node [above] {\tiny$c(p_m)$};
	\end{tikzpicture} 
	\quad $\rightarrow$ 
	\begin{tikzpicture}[line width=1, scale=2,baseline=0]
		\draw (-1,0) arc (180:90:1);
		\draw (1,0) arc (0:45:1);
		\draw (1,0) arc (0:-15:1);
		\draw (0,-1) arc (270:310:1);
		\draw (0,-1) arc (270:220:1);
		\def \n {10}
		\foreach \x/\y in {-0.588/-0.809,0.309/-0.951,1.00000000000000/0,-0.309/0.951,-0.951/0.309}{\draw[decorate,decoration={snake}] (-0.108,-0.10) -- (\x,\y);}
		\draw[decorate,decoration={snake}] (0.951,0.309) -- ++({atan(0.325}:-.45cm);
		\draw (360*6/\n-360/\n+180:1) circle(.4pt)circle(.8pt)circle(1.2pt) node[anchor=360/\n*6-360/\n]{$i$};
		\draw (360*7/\n-360/\n+180:1) circle(.4pt)circle(.8pt)circle(1.2pt) node[anchor=360/\n*7-360/\n]{$i+1$};
		\draw (360*1.5/\n-360/\n+180:1) node[rotate=360*1.5/\n-360/\n+90] {$\cdots$};
		\draw (360*1.5/\n-360/\n+180:.65) node[rotate=360*1.5/\n-360/\n+90] {$\cdots$};
		\draw (360*7.875/\n-360/\n+180:1) node[rotate=360*7.875/\n-360/\n+90] {$\cdots$};
		\draw (360*5/\n-360/\n+180:1) node[rotate=360*5/\n-360/\n+90] {$\cdots$};
		\draw (360*9.75/\n-360/\n+180:.9) node[rotate=360*9.75/\n-360/\n+90] {$\cdots$};
		\draw (360*3.25/\n-360/\n+180:.9) node[rotate=360*3.25/\n-360/\n+90] {$\cdots$};
		\draw (0,-0.25) node [right] {\scriptsize$\sum_i^mc(p_i)-1$};
	\end{tikzpicture} \\
	\hline
	X-Merge and Split \\
	\hline
	\begin{tikzpicture}[line width=1, scale=2, baseline=0]
		\draw (-1,0) arc (180:90:1);
		\draw (1,0) arc (0:45:1);
		\draw (1,0) arc (0:-15:1);
		\draw (0,-1) arc (270:310:1);
		\draw (0,-1) arc (270:220:1);
		\def \n {20}
		\draw (360*12/\n-360/\n+180:1)circle(.4pt)circle(.8pt)circle(1.2pt) node[anchor=360/\n*12-360/\n]{$i$};
		\draw (360*11/\n-360/\n+180:1)circle(.4pt)circle(.8pt)circle(1.2pt) node[anchor=360/\n*11-360/\n]{$i-1$};
		\draw (360*13/\n-360/\n+180:1)circle(.4pt)circle(.8pt)circle(1.2pt) node[anchor=360/\n*13-360/\n]{$i+1$};
		\draw (360*2/\n-360/\n+180:1) node[rotate=360*2/\n-360/\n+90] {$\cdots$};
		\draw (360*14.75/\n-360/\n+180:1) node[rotate=360*14.75/\n-360/\n+90] {$\cdots$};
		\draw (360*9/\n-360/\n+180:1) node[rotate=360*9/\n-360/\n+90] {$\cdots$};
		\draw (360*18/\n-360/\n+180:.9) node[rotate=360*18/\n-360/\n+90] {$\cdots$};
		\draw (360*5/\n-360/\n+180:.9) node[rotate=360*5/\n-360/\n+90] {$\cdots$};
		\foreach \x/\y in {0.891/0.454,-0.309/0.951,-0.891/0.454}{\draw[decorate,decoration={snake}] (-0.103,0.620) -- (\x,\y);}
		\foreach \x/\y in {-0.707/-0.707,0/-1.00000000000000,0.988/0.156}{\draw[decorate,decoration={snake}] (0.094,-0.517) -- (\x,\y);}
	\end{tikzpicture}
	\quad $ \rightarrow$ 
	\begin{tikzpicture}[line width=1, scale=2,baseline=0]
		\def \n {20}
		\draw (-1,0) arc (180:90:1);
		\draw (1,0) arc (0:45:1);
		\draw (1,0) arc (0:-15:1);
		\draw (0,-1) arc (270:310:1);
		\draw (0,-1) arc (270:220:1);
		\def \n {20}
		\draw (360*12/\n-360/\n+180:1)circle(.4pt)circle(.8pt)circle(1.2pt) node[anchor=360/\n*12-360/\n]{$i$};
		\draw (360*11/\n-360/\n+180:1)circle(.4pt)circle(.8pt)circle(1.2pt) node[anchor=360/\n*11-360/\n]{$i-1$};
		\draw (360*13/\n-360/\n+180:1)circle(.4pt)circle(.8pt)circle(1.2pt) node[anchor=360/\n*13-360/\n]{$i+1$};
		\draw (360*2/\n-360/\n+180:1) node[rotate=360*2/\n-360/\n+90] {$\cdots$};
		\draw (360*14.75/\n-360/\n+180:1) node[rotate=360*14.75/\n-360/\n+90] {$\cdots$};
		\draw (360*9/\n-360/\n+180:1) node[rotate=360*9/\n-360/\n+90] {$\cdots$};
		\draw (360*18/\n-360/\n+180:.9) node[rotate=360*18/\n-360/\n+90] {$\cdots$};
		\draw (360*5/\n-360/\n+180:.9) node[rotate=360*5/\n-360/\n+90] {$\cdots$};
		\foreach \x/\y in {-0.707/-0.707,0/-1.00000000000000,1.00000000000000/0,0.809/0.588,-0.309/0.951,-0.891/0.454}{\draw[decorate,decoration={snake}] (-0.016,0.048) -- (\x,\y);}
		\draw[decorate,decoration={snake}] (0.951,0.309) -- ++({atan(0.325}:-.5 cm);
	\end{tikzpicture}\\
	\hline
	Retract or Split\\
	\hline
	\begin{tikzpicture}[line width=1, scale=2,baseline=0]
		\draw (-1,0) arc (180:90:1);
		\draw (1,0) arc (0:45:1);
		\draw (1,0) arc (0:-15:1);
		\draw (0,-1) arc (270:310:1);
		\draw (0,-1) arc (270:220:1);
		\def \n {10}
		\foreach \x/\y in {0.951/0.309,0/-1.00000000000000,-0.309/0.951,-0.951/0.309}{\draw[decorate,decoration={snake}] (-0.077,0.142) -- (\x,\y);}
		\draw (0.1,0.2)node[above] {\scriptsize $c(p)$};
		\draw (360*6.5/\n-360/\n+180:1)circle(.4pt)circle(.8pt)circle(1.2pt) node[anchor=360/\n*6.5-360/\n]{$i$};
		\draw (360*1.5/\n-360/\n+180:1) node[rotate=360*1.5/\n-360/\n+90] {$\cdots$};
		\draw (360*7.875/\n-360/\n+180:1) node[rotate=360*7.875/\n-360/\n+90] {$\cdots$};
		\draw (360*5/\n-360/\n+180:1) node[rotate=360*5/\n-360/\n+90] {$\cdots$};
		\draw (360*1.5/\n-360/\n+180:.3) node[rotate=360*1.5/\n-360/\n+100] {$\cdots$};
	\end{tikzpicture} \quad $ \rightarrow$ \quad $\begin{cases}
		\begin{tikzpicture}[line width=1, scale=2,baseline=0]
			\draw (-1,0) arc (180:90:1);
			\draw (1,0) arc (0:45:1);
			\draw (1,0) arc (0:-15:1);
			\draw (0,-1) arc (270:310:1);
			\draw (0,-1) arc (270:220:1);
			\def \n {10}
			\foreach \x/\y in {0/-1.00000000000000,-0.309/0.951,-0.951/0.309}{\draw[decorate,decoration={snake}] (-0.420,0.087) -- (\x,\y);}
			\draw (-0.420,0.087) node[right] {\scriptsize $c(p)-1$};
			\draw[decorate,decoration={snake}] (0.951,0.309) -- ++({atan(0.325}:-.6 cm);
			\draw (360*6.5/\n-360/\n+180:1)circle(.4pt)circle(.8pt)circle(1.2pt) node[anchor=360/\n*6.5-360/\n]{$i$};
			\draw (360*1.5/\n-360/\n+180:1) node[rotate=360*1.5/\n-360/\n+90] {$\cdots$};
			\draw (360*7.875/\n-360/\n+180:1) node[rotate=360*7.875/\n-360/\n+90] {$\cdots$};
			\draw (360*5/\n-360/\n+180:1) node[rotate=360*5/\n-360/\n+90] {$\cdots$};
			\draw (360*1.5/\n-360/\n+180:.6) node[rotate=360*1.5/\n-360/\n+100] {$\cdots$};
		\end{tikzpicture} & \text{ if } c(p)>1 \\
		\begin{tikzpicture}[line width=1, scale=2,baseline=0]
			\draw (-1,0) arc (180:90:1);
			\draw (1,0) arc (0:45:1);
			\draw (1,0) arc (0:-15:1);
			\draw (0,-1) arc (270:310:1);
			\draw (0,-1) arc (270:220:1);
			\def \n {10}
			\foreach \x/\y in {0/-1.00000000000000,-0.309/0.951,-0.951/0.309}{\draw[decorate,decoration={snake}] (-0.420,0.087) -- (\x,\y);}
			\draw (360*6.5/\n-360/\n+180:1)circle(.4pt)circle(.8pt)circle(1.2pt) node[anchor=360/\n*6.5-360/\n]{$i$};
			\draw (360*1.5/\n-360/\n+180:1) node[rotate=360*1.5/\n-360/\n+90] {$\cdots$};
			\draw (360*7.875/\n-360/\n+180:1) node[rotate=360*7.875/\n-360/\n+90] {$\cdots$};
			\draw (360*5/\n-360/\n+180:1) node[rotate=360*5/\n-360/\n+90] {$\cdots$};
			\draw (360*1.5/\n-360/\n+180:.6) node[rotate=360*1.5/\n-360/\n+100] {$\cdots$};
		\end{tikzpicture}& \text{ if } c(p)=1
	\end{cases}$\\
	\hline
\end{longtable}

\subsection{Boundary moves \label{sec:moves}} 
In this section, we lay out our five main diagrammatic moves.

In Section \ref{sec:movesrealizeboundaries} we prove that each move, when applied to the appropriate generalized non-crossing Wilson loop diagram, gives a boundary of the corresponding positroid. Therefore, we will call these diagrammatic operations boundary moves.

Each of the moves focuses locally on a set of ends supported on at most two adjacent edges and ignores the rest of the diagram.
Note also that any given generalized (non-crossing) Wilson loop diagram and any given propagator may have many different boundary moves that can be applied to it.

\subsubsection{The slide move\label{sec:slide}}
Let $W=(\cP, [n])$ be a non-crossing generalized Wilson loop diagram.  Without loss of generality,  by Remark~\ref{rmk:drawing conventions}, we may assume that if a single-ended propagator is supported on an edge, then it is the only propagator supported on that edge.

Let $p \in \cP$ be a propagator with an end $a$ supported on an edge $e_i$, i.e.\ with support $V(a) = \{ i, i+1\}$. Furthermore, suppose $p$ is either the first or the last propagator on the edge $e_i$ in the cyclic ordering of ends given in Definition \ref{dfn:new def of prop order at edge or vertex}. In other words, among the propagators supported on $e_i$, $p$ is either closest to the vertex $i$ or closest to the vertex $i+1$ (or both).

If $p$ is the first propagator on $e_i$, then the result of the \emph{slide move} applied to the end $a$ moving in the direction of $i$ is the generalized Wilson loop diagram obtained from $W$ by changing the support of $a$ to $V(a)= \{i\}$. We say the move slides $a$ over to $i$, as illustrated in the following picture. 

$$
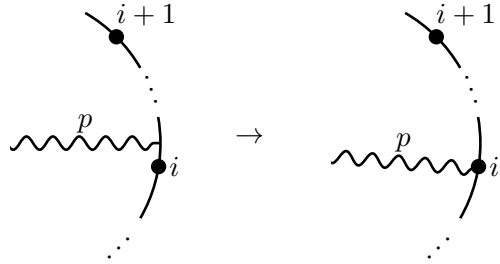
\begin{tikzpicture}[line width=1, scale=2,baseline=0]
\clip (1.25,0)+(90:1.25cm) arc[start angle = 90, end angle =270,radius=1.25cm] --(1.25,0);
\draw (0,0)+(30:1cm) arc[start angle = 30, end angle =60,radius=1cm];
\draw (1,0) arc (0:10:1);
\draw (1,0) arc (0:-30:1);
\def \n {10}
\draw foreach \v in{5.75,7.25}{(360*\v/\n-360/\n+180:1)circle(.4pt)circle(.8pt)circle(1.2pt)node[anchor=360/\n*\v-360/\n] {}};
\draw (360*5.75/\n-360/\n+180:1) node[anchor=360/\n*6-360/\n]{$i$};
\draw (360*7.25/\n-360/\n+180:1) node[anchor=360/\n*7-360/\n]{$i+1$};
\draw (360*6.55/\n-360/\n+180:1) node[rotate=360*6.55/\n-360/\n+90] {$\cdots$};
\draw (360*8/\n-360/\n+180:1) node[rotate=360*8/\n-360/\n+90] {$\cdots$};
\draw (360*4.75/\n-360/\n+180:1) node[rotate=360*4.75/\n-360/\n+90] {$\cdots$};
\draw[decorate,decoration={snake}] (-1,0) -- (360*6/\n-360/\n+180:1) node[pos=.75,above]{$p$};
\end{tikzpicture}
\quad \rightarrow \qquad
\begin{tikzpicture}[line width=1, scale=2,baseline=0]
\clip (1.25,0)+(90:1.25cm) arc[start angle = 90, end angle =270,radius=1.25cm] --(1.25,0);
\draw (0,0)+(30:1cm) arc[start angle = 30, end angle =60,radius=1cm];
\draw (1,0) arc (0:10:1);
\draw (1,0) arc (0:-30:1);
\def \n {10}
\draw foreach \v in{5.75,7.25}{(360*\v/\n-360/\n+180:1)circle(.4pt)circle(.8pt)circle(1.2pt)node[anchor=360/\n*\v-360/\n] {}};
\draw (360*5.75/\n-360/\n+180:1) node[anchor=360/\n*6-360/\n]{$i$};
\draw (360*7.25/\n-360/\n+180:1) node[anchor=360/\n*7-360/\n]{$i+1$};
\draw (360*6.55/\n-360/\n+180:1) node[rotate=360*6.55/\n-360/\n+90] {$\cdots$};
\draw (360*8/\n-360/\n+180:1) node[rotate=360*8/\n-360/\n+90] {$\cdots$};
\draw (360*4.75/\n-360/\n+180:1) node[rotate=360*4.75/\n-360/\n+90] {$\cdots$};
\draw[decorate,decoration={snake}] (-1,0) -- (.985,-.174) node[pos=.75,above]{$p$};
\end{tikzpicture}
$$ 

Similarly, if $p$ is the last end on $e_i$, then the result of the \emph{slide move} applied to the end $a$ and moving in the direction of $i+1$ is the generalized Wilson loop diagram obtained from $W$ by changing the support of $a$ to $V(a)= \{i+1\}$ and we say the move slides $a$ to $i+1$.

Note that the vertex to which we are sliding the propagator $p$ may have arbitrarily many (or no) other ends supported on it.

Note as well that by restricting the application of the move to those propagators closest to a vertex the resulting generalized Wilson loop diagram is guaranteed to be non-crossing. 

\subsubsection{The merge move\label{sec:merge}} 
Let $W$ be a non-crossing generalized Wilson loop diagram.  
Let $p,q \in \cP$ be two distinct propagators each with capacity one less than the number of ends: $c(p) = |\cE(p)|-1$ and $c(q) = |\cE(q)|-1$. Let $p_e$ and $q_e$ be ends of $p$ and $q$ respectively which are consecutive in the end order of Definition~\ref{dfn:global cyclic end order} and which are either 
\begin{itemize}
    \item both on the same edge $e_i$, 
    \item both on the same vertex $v_i$,
\item or one on an edge $e_i$ and the other on an adjacent vertex $v_i$ or $v_{i+1}$.
\end{itemize}
Additionally, suppose that if either $p_e$ or $q_e$ are on the edge $e_i$ then no other end of the same propagator is on $e_i$, $v_i$ or $v_{i+1}$. 

For two propagators $p$ and $q$ satisfying these conditions, applying the \emph{merge move} to the ends $p_e$ and $q_e$ results in the generalized Wilson loop diagram that is obtained from $W$ by replacing $p$ and $q$ by a new propagator $r$ with capacity $c(r)=c(p)+c(q)$ and with ends $(\cE(p)\setminus p_e)\sqcup (\cE(q)\setminus p_q) \sqcup \{r_e\}$ where $r_e$ is a new end with support $V(r_e) = V(p_e)\cap V(q_e)$. The merge move on two propagators $p$ and $q$ is illustrated in the following pictures, first in the case where $p_e$ and $q_e$ are on $e_i$.

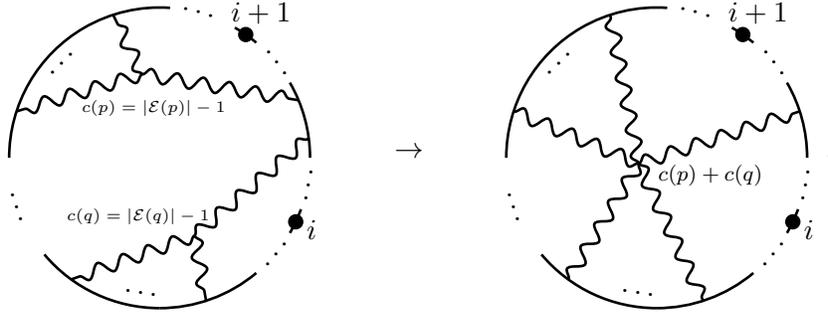
$$
\begin{tikzpicture}[line width=1, scale=2,baseline=0]
	\draw (-1,0) arc (180:86:1);
	\draw (1,0) arc (0:60:1);
	\draw (1,0) arc (0:-30:1);
	\draw (0,-1) arc (270:310:1);
	\draw (0,-1) arc (270:220:1);
	
	\def \n {10}
	\draw foreach \v in{5.3,7.53}{(360*\v/\n-360/\n+180:1)circle(.4pt)circle(.8pt)circle(1.2pt)node[anchor=360/\n*\v-360/\n] {}};
	\draw (360*5.3/\n-360/\n+180:1) node[right, anchor=(360*5.3/\n-360/\n]{$i$};
	\draw (360*7.53/\n-360/\n+180:1) node[anchor=360*7.53/\n-360/\n]{$i+1$};
 \draw (360*7.0625/\n-360/\n+180:1) node[rotate=360*7.0625/\n-360/\n+90] {$\cdots$};
	\draw (360*1.5/\n-360/\n+180:1) node[rotate=360*1.5/\n-360/\n+90] {$\cdots$};
	\draw (360*8.0625/\n-360/\n+180:1) node[rotate=360*8.0625/\n-360/\n+90] {$\cdots$};
	\draw (360*4.875/\n-360/\n+180:1) node[rotate=360*4.875/\n-360/\n+90] {$\cdots$};
	
	\draw (360*9.75/\n-360/\n+180:.9) node[rotate=360*9.75/\n-360/\n+90] {$\cdots$};
	\draw (360*3.25/\n-360/\n+180:.9) node[rotate=360*3.25/\n-360/\n+90] {$\cdots$};
	\foreach \x/\y in {0.924/0.383,-0.309/0.951,-0.951/0.309}{\draw[decorate,decoration={snake}] (-0.115,0.556) -- (\x,\y);}
	\foreach \x/\y in {-0.588/-0.809,0.309/-0.951,.991/0.131}{\draw[decorate,decoration={snake}] (0.233,-0.517) -- (\x,\y);}
	\draw (-0.115,0.556) node[shift={(.15,-0.45)}] {\tiny$c(p)= |\cE(p)| -1$};
	\draw (0.233,-0.517) node[shift={(-.75,0.25)}] {\tiny$c(q)= |\cE(q)| -1$};
	\draw[white, line width=5] (.866,.5) arc (30:50:1);
	\draw[white, line width=5] (1,0) arc (0:-20:1);
	\draw (360*5.6875/\n-360/\n+180:1) node[rotate=360*5.6875/\n-360/\n+90] {$\cdots$};
	\draw (360*7.0625/\n-360/\n+180:1) node[rotate=360*7.0625/\n-360/\n+90] {$\cdots$};
\end{tikzpicture}
\qquad \rightarrow \qquad 
	\begin{tikzpicture}[line width=1, scale=2,baseline=0]
	\draw (-1,0) arc (180:86:1);
	\draw (1,0) arc (0:60:1);
	\draw (1,0) arc (0:-30:1);
	\draw (0,-1) arc (270:310:1);
	\draw (0,-1) arc (270:220:1);
	\def \n {10}
	\draw foreach \v in{5.3,7.53}{(360*\v/\n-360/\n+180:1)circle(.4pt)circle(.8pt)circle(1.2pt)node[anchor=360/\n*\v-360/\n] {}};
	\draw (360*5.3/\n-360/\n+180:1) node[right, anchor=(360*5.3/\n-360/\n]{$i$};
	\draw (360*7.53/\n-360/\n+180:1) node[anchor=360*7.53/\n-360/\n]{$i+1$};
	\draw (360*1.5/\n-360/\n+180:1) node[rotate=360*1.5/\n-360/\n+90] {$\cdots$};
	\draw (360*8.0625/\n-360/\n+180:1) node[rotate=360*8.0625/\n-360/\n+90] {$\cdots$};
	\draw (360*4.875/\n-360/\n+180:1) node[rotate=360*4.875/\n-360/\n+90] {$\cdots$};
	\draw (360*9.75/\n-360/\n+180:.9) node[rotate=360*9.75/\n-360/\n+90] {$\cdots$};
	\draw (360*3.25/\n-360/\n+180:.9) node[rotate=360*3.25/\n-360/\n+90] {$\cdots$};
	\foreach \x/\y in {-0.588/-0.809,0.309/-0.951,0.951/0.309,-0.309/0.951,-0.951/0.309}{\draw[decorate,decoration={snake}] (-0.118,-0.038) -- (\x,\y);}
	\draw (-0.118,-0.038) node[shift={(.95,-.15)}] {\scriptsize$c(p)+c(q)$};
	\draw[white, line width=5] (.866,.5) arc (30:50:1);
	\draw[white, line width=5] (1,0) arc (0:-20:1);
	\draw (360*5.6875/\n-360/\n+180:1) node[rotate=360*5.6875/\n-360/\n+90] {$\cdots$};
	\draw (360*7.0625/\n-360/\n+180:1) node[rotate=360*7.0625/\n-360/\n+90] {$\cdots$};
\end{tikzpicture},
$$
and second when $p_e$ is supported on $e_i$ and $q_e$ is supported on $v_i$. 
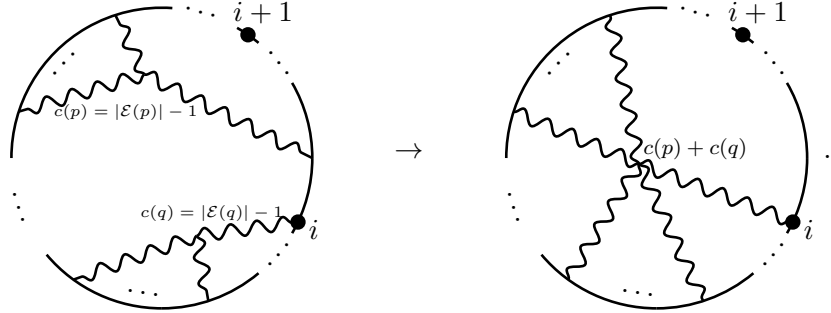
$$\begin{tikzpicture}[line width=1, scale=2,baseline=0]
	\draw (-1,0) arc (180:86:1);
	\draw (1,0) arc (0:60:1);
	\draw (1,0) arc (0:-30:1);
	\draw (0,-1) arc (270:310:1);
	\draw (0,-1) arc (270:220:1);
	\def \n {10}
	\draw foreach \v in{5.3,7.53}{(360*\v/\n-360/\n+180:1)circle(.4pt)circle(.8pt)circle(1.2pt)node[anchor=360/\n*\v-360/\n] {}};
	\draw (360*5.3/\n-360/\n+180:1) node[right, anchor=(360*5.3/\n-360/\n]{$i$};
	\draw (360*7.53/\n-360/\n+180:1) node[anchor=360*7.53/\n-360/\n]{$i+1$};
	\draw (360*1.5/\n-360/\n+180:1) node[rotate=360*1.5/\n-360/\n+90] {$\cdots$};
	\draw (360*8.0625/\n-360/\n+180:1) node[rotate=360*8.0625/\n-360/\n+90] {$\cdots$};
	\draw (360*4.875/\n-360/\n+180:1) node[rotate=360*4.875/\n-360/\n+90] {$\cdots$};
	\draw (360*9.75/\n-360/\n+180:.9) node[rotate=360*9.75/\n-360/\n+90] {$\cdots$};
	\draw (360*3.25/\n-360/\n+180:.9) node[rotate=360*3.25/\n-360/\n+90] {$\cdots$};
	\foreach \x/\y in {1/0,-0.309/0.951,-0.951/0.309}{\draw[decorate,decoration={snake}] (-0.115,0.556) -- (\x,\y);}
	\foreach \x/\y in {-0.588/-0.809,0.309/-0.951,.906/-0.423}{\draw[decorate,decoration={snake}] (0.233,-0.517) -- (\x,\y);}
	\draw (-0.115,0.556) node [shift={(-.25,-0.5)}] {\tiny $c(p)= |\cE(p)| -1$};
	\draw (0.233,-0.517) node [shift={(.2,0.3)}] {\tiny $c(q)= |\cE(q)| -1$};
	\draw[white, line width=5] (.866,.5) arc (30:50:1);
	\draw (360*7.0625/\n-360/\n+180:1) node[rotate=360*7.0625/\n-360/\n+90] {$\cdots$};
\end{tikzpicture} \qquad \rightarrow \qquad
\begin{tikzpicture}[line width=1, scale=2,baseline=0]
	\draw (-1,0) arc (180:86:1);
	\draw (1,0) arc (0:60:1);
	\draw (1,0) arc (0:-30:1);
	\draw (0,-1) arc (270:310:1);
	\draw (0,-1) arc (270:220:1);
	\def \n {10}
	\draw foreach \v in{5.3,7.53}{(360*\v/\n-360/\n+180:1)circle(.4pt)circle(.8pt)circle(1.2pt)node[anchor=360/\n*\v-360/\n] {}};
	\draw (360*5.3/\n-360/\n+180:1) node[right, anchor=(360*5.3/\n-360/\n]{$i$};
	\draw (360*7.53/\n-360/\n+180:1) node[anchor=360*7.53/\n-360/\n]{$i+1$};
	\draw (360*1.5/\n-360/\n+180:1) node[rotate=360*1.5/\n-360/\n+90] {$\cdots$};
	\draw (360*8.0625/\n-360/\n+180:1) node[rotate=360*8.0625/\n-360/\n+90] {$\cdots$};
	\draw (360*4.875/\n-360/\n+180:1) node[rotate=360*4.875/\n-360/\n+90] {$\cdots$};
	\draw (360*9.75/\n-360/\n+180:.9) node[rotate=360*9.75/\n-360/\n+90] {$\cdots$};
	\draw (360*3.25/\n-360/\n+180:.9) node[rotate=360*3.25/\n-360/\n+90] {$\cdots$};
	\foreach \x/\y in {-0.588/-0.809,0.309/-0.951,.906/-0.423,-0.309/0.951,-0.951/0.309}{\draw[decorate,decoration={snake}] (-0.118,-0.038) -- (\x,\y);}
	\draw (-0.118,-0.038) node[shift={(.75,0.2)}] {\scriptsize$c(p)+c(q)$};
	\draw[white, line width=5] (.866,.5) arc (30:50:1);
	\draw (360*7.0625/\n-360/\n+180:1) node[rotate=360*7.0625/\n-360/\n+90] {$\cdots$};
\end{tikzpicture}.
$$

This move can be visualized by sliding the two ends that are being merged together, similarly to the slide move, and then identifying them.  If one end is on an edge and the other on a vertex then the end on the edge is slid to the vertex, explaining why $V(r_e)= V(p_e)\cap V(q_e)$.

Note as well that by restricting the application of the move to adjacent propagators the resulting generalized Wilson loop diagram is guaranteed to be non-crossing. 

\subsubsection{The merge and split move \label{sec:merge and split}} 

As before, let $W = (\cP, [n])$ be a generalized Wilson loop diagram with no crossing propagators.  Assume, without loss of generality by Remark~\ref{rmk:drawing conventions}, that if a single-ended propagator is supported on an edge, then it is the only propagator supported on that edge. 

Let $e_i$ be an edge that supports multiple propagators each with exactly one end on $e_i$.  Further suppose each of those propagators has capacity one less than its number of ends.  Let $P$ be the set of propagators supported on $e_i$.  By assumption $|P|>1$ and $P$ contains no single-ended propagators.  For $p\in P$, write $p_e$ for the end of $p$ on $e_i$. 

For $P$ as above, the result of the \emph{merge and split move} applied to the ends on $e_i$ is the generalized Wilson loop diagram obtained from $W$ by replacing the propagators of $P$ with two new propagators $r$ and $s$ defined as follows.  The propagator $r$ has capacity $(\sum_{p \in P} c(p))-1$ and ends $\bigsqcup_{p \in P}(\cE(p)\setminus p_e) \sqcup \{r_e\}$ where $r_e$ is a new end with support $\{v_i\}$.  The propagator $s$ has capacity $1$ and has its single end on the edge $e_i$.
This is illustrated in the following picture. 

$$
	
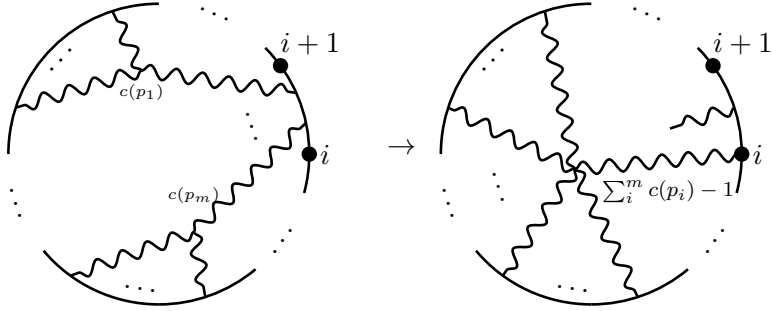
\begin{tikzpicture}[line width=1, scale=2,baseline=0]
	\draw (-1,0) arc (180:90:1);
	\draw (1,0) arc (0:45:1);
	\draw (1,0) arc (0:-15:1);
	\draw (0,-1) arc (270:310:1);
	\draw (0,-1) arc (270:220:1);
	\def \n {10}
	\draw foreach \v in{6,7}{(360*\v/\n-360/\n+180:1)circle(.4pt)circle(.8pt)circle(1.2pt)node[anchor=360/\n*\v-360/\n] {}};
	\draw (360*6/\n-360/\n+180:1) node[anchor=360/\n*6-360/\n]{$i$};
	\draw (360*7/\n-360/\n+180:1) node[anchor=360/\n*7-360/\n]{$i+1$};
	\draw (360*1.5/\n-360/\n+180:1) node[rotate=360*1.5/\n-360/\n+90] {$\cdots$};
	\draw (360*1.5/\n-360/\n:.65) node[rotate=360*1.5/\n-360/\n+90] {$\cdots$};	
	\draw (360*7.875/\n-360/\n+180:1) node[rotate=360*7.875/\n-360/\n+90] {$\cdots$};
	\draw (360*5/\n-360/\n+180:1) node[rotate=360*5/\n-360/\n+90] {$\cdots$};
	\draw (360*9.75/\n-360/\n+180:.9) node[rotate=360*9.75/\n-360/\n+90] {$\cdots$};
	\draw (360*3.25/\n-360/\n+180:.9) node[rotate=360*3.25/\n-360/\n+90] {$\cdots$};
	\foreach \x/\y in {0.914/0.407,-0.309/0.951,-0.951/0.309}{\draw[decorate,decoration={snake}] (-0.115,0.556) -- (\x,\y);}
	\foreach \x/\y in {-0.588/-0.809,0.309/-0.951,0.978/0.208}{\draw[decorate,decoration={snake}] (0.233,-0.517) -- (\x,\y);}
	\draw (-0.115,0.54) node [below] {\tiny$c(p_1)$};
	\draw (0.233,-0.4) node [above] {\tiny$c(p_m)$};
\end{tikzpicture} 
\quad \rightarrow
\begin{tikzpicture}[line width=1, scale=2,baseline=0]
	\draw (-1,0) arc (180:90:1);
	\draw (1,0) arc (0:45:1);
	\draw (1,0) arc (0:-15:1);
	\draw (0,-1) arc (270:310:1);
	\draw (0,-1) arc (270:220:1);
	\def \n {10}
	\foreach \x/\y in {-0.588/-0.809,0.309/-0.951,1.00000000000000/0,-0.309/0.951,-0.951/0.309}{\draw[decorate,decoration={snake}] (-0.108,-0.10) -- (\x,\y);}
	\draw[decorate,decoration={snake}] (0.951,0.309) -- ++({atan(0.325}:-.45cm);
	\draw (360*6/\n-360/\n+180:1) circle(.4pt)circle(.8pt)circle(1.2pt) node[anchor=360/\n*6-360/\n]{$i$};
	\draw (360*7/\n-360/\n+180:1) circle(.4pt)circle(.8pt)circle(1.2pt) node[anchor=360/\n*7-360/\n]{$i+1$};
	\draw (360*1.5/\n-360/\n+180:1) node[rotate=360*1.5/\n-360/\n+90] {$\cdots$};
	\draw (360*1.5/\n-360/\n+180:.65) node[rotate=360*1.5/\n-360/\n+90] {$\cdots$};
	\draw (360*7.875/\n-360/\n+180:1) node[rotate=360*7.875/\n-360/\n+90] {$\cdots$};
	\draw (360*5/\n-360/\n+180:1) node[rotate=360*5/\n-360/\n+90] {$\cdots$};
	\draw (360*9.75/\n-360/\n+180:.9) node[rotate=360*9.75/\n-360/\n+90] {$\cdots$};
	\draw (360*3.25/\n-360/\n+180:.9) node[rotate=360*3.25/\n-360/\n+90] {$\cdots$};
	\draw (0,-0.25) node [right] {\scriptsize$\sum_i^mc(p_i)-1$};
\end{tikzpicture}
$$

Note that the new diagram has that $s$ is the only propagator on its edge, so no further rewriting via Remark~\ref{rmk:drawing conventions} is necessary.

Note also that since $W$ has no crossing propagators, and the merge and split move combines all propagators with ends on the same edge, it is merging a set of adjacent propagators. Therefore, the diagram resulting from this move has no crossing propagators.

\subsubsection{The X-merge and split move \label{sec:X-merge and split}} 

Let $W$ be a generalized Wilson loop diagram with no crossing propagators.
Let $e_{i-1}$ and $e_{i}$ be two adjacent edges of the diagram, each supporting a single multi-ended propagator of capacity 1. Furthermore, suppose the vertex $v_{i}$ between the two supports no propagators. Let $p$ and $q$ be the propagators supported on $e_{i-1}$ and $e_{i}$ respectively, and $p_e$ and $q_e$ the ends on those edges.  

Note that this is the only move in which we must have no additional ends on the vertices and edges involved in the move ($v_{i-1}, v_{i}$ and $v_{i+1}$).  

Suppose $p$ and $q$ are as above.  Then the result of the \emph{X-merge and split move} applied to the ends $p_e$ and $q_e$ is the generalized Wilson loop diagram obtained from $W$ by replacing the propagators $p$ and $q$ with two new propagators $r$ and $s$ both of capacity $1$ and with ends defined as follows.  The propagator $r$ has ends $\cE(p) \sqcup \cE(p)$ where additionally the support of $p_e$ is changed to $\{v_{i-1}\}$ and the support of $q_e$ is changed to $\{v_{i+1}\}$.  The propagator $s$ has its single end on the vertex $i$.
This is illustrated in the following picture. 

$$
	
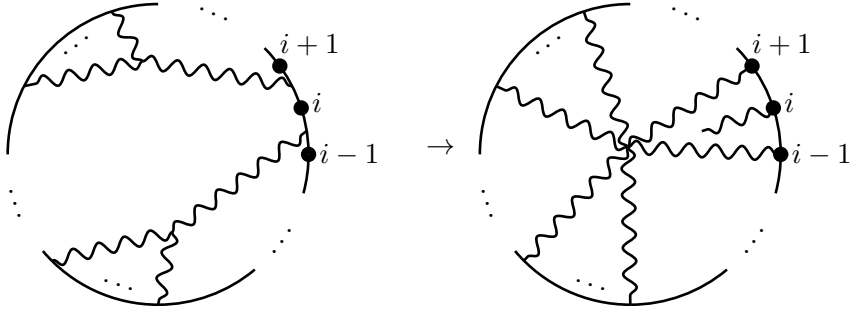
\begin{tikzpicture}[line width=1, scale=2, baseline=0]
	\draw (-1,0) arc (180:90:1);
	\draw (1,0) arc (0:45:1);
	\draw (1,0) arc (0:-15:1);
	\draw (0,-1) arc (270:310:1);
	\draw (0,-1) arc (270:220:1);
	\def \n {20}
	\draw (360*12/\n-360/\n+180:1)circle(.4pt)circle(.8pt)circle(1.2pt) node[anchor=360/\n*12-360/\n]{$i$};
	\draw (360*11/\n-360/\n+180:1)circle(.4pt)circle(.8pt)circle(1.2pt) node[anchor=360/\n*11-360/\n]{$i-1$};
	\draw (360*13/\n-360/\n+180:1)circle(.4pt)circle(.8pt)circle(1.2pt) node[anchor=360/\n*13-360/\n]{$i+1$};
	\draw (360*2/\n-360/\n+180:1) node[rotate=360*2/\n-360/\n+90] {$\cdots$};
	\draw (360*14.75/\n-360/\n+180:1) node[rotate=360*14.75/\n-360/\n+90] {$\cdots$};
	\draw (360*9/\n-360/\n+180:1) node[rotate=360*9/\n-360/\n+90] {$\cdots$};
	\draw (360*18/\n-360/\n+180:.9) node[rotate=360*18/\n-360/\n+90] {$\cdots$};
	\draw (360*5/\n-360/\n+180:.9) node[rotate=360*5/\n-360/\n+90] {$\cdots$};
	\foreach \x/\y in {0.891/0.454,-0.309/0.951,-0.891/0.454}{\draw[decorate,decoration={snake}] (-0.103,0.620) -- (\x,\y);}
	\foreach \x/\y in {-0.707/-0.707,0/-1.00000000000000,0.988/0.156}{\draw[decorate,decoration={snake}] (0.094,-0.517) -- (\x,\y);}
\end{tikzpicture}
\quad  \rightarrow
\begin{tikzpicture}[line width=1, scale=2,baseline=0]
	\def \n {20}
	\draw (-1,0) arc (180:90:1);
	\draw (1,0) arc (0:45:1);
	\draw (1,0) arc (0:-15:1);
	\draw (0,-1) arc (270:310:1);
	\draw (0,-1) arc (270:220:1);
	\def \n {20}
	\draw (360*12/\n-360/\n+180:1)circle(.4pt)circle(.8pt)circle(1.2pt) node[anchor=360/\n*12-360/\n]{$i$};
	\draw (360*11/\n-360/\n+180:1)circle(.4pt)circle(.8pt)circle(1.2pt) node[anchor=360/\n*11-360/\n]{$i-1$};
	\draw (360*13/\n-360/\n+180:1)circle(.4pt)circle(.8pt)circle(1.2pt) node[anchor=360/\n*13-360/\n]{$i+1$};
	\draw (360*2/\n-360/\n+180:1) node[rotate=360*2/\n-360/\n+90] {$\cdots$};
	\draw (360*14.75/\n-360/\n+180:1) node[rotate=360*14.75/\n-360/\n+90] {$\cdots$};
	\draw (360*9/\n-360/\n+180:1) node[rotate=360*9/\n-360/\n+90] {$\cdots$};
	\draw (360*18/\n-360/\n+180:.9) node[rotate=360*18/\n-360/\n+90] {$\cdots$};
	\draw (360*5/\n-360/\n+180:.9) node[rotate=360*5/\n-360/\n+90] {$\cdots$};
	\foreach \x/\y in {-0.707/-0.707,0/-1.00000000000000,1.00000000000000/0,0.809/0.588,-0.309/0.951,-0.891/0.454}{\draw[decorate,decoration={snake}] (-0.016,0.048) -- (\x,\y);}
	\draw[decorate,decoration={snake}] (0.951,0.309) -- ++({atan(0.325}:-.5 cm);
\end{tikzpicture}
$$

We may visualize the change in support of the ends $p_e$ and $q_e$ as sliding them to the vertices $v_{i-1}$ and $v_{i+1}$ respectively.

Note that unlike the merge or the merge and split move, the number of ends of $r$ is the sum of the number of ends of $p$ and $q$: $|\cE(r)| = |\cE(p)| + |\cE(q)|$. The propagators $s$ is a single-ended propagator on the vertex $v_i$. This is the only move where the number of ends of the diagram increases.

Since $W$ has no crossing propagators, and $p$ and $q$ are the only propagators on $e_{i-1}$ and $e_i$ respectively, the X-merge and split does not generate any crossing propagators in the image.  

\subsubsection{Retract or split move\label{sec:retraction}} 
Let $W$ be a generalized Wilson loop diagram.
Let $v_i$ be a vertex on which there is an end $a$ of a multi-ended propagator $p$.

If $p$ has capacity one, then the result of applying this move to $p$ at $v_i$ is the generalized Wilson loop diagram obtained from $W$ by removing the end $a$ from $p$.  In this case the move is called \emph{retraction}.  If $p$ has capacity greater than one then the result of applying this move to $p$ at $v_i$ is the generalized Wilson loop diagram obtained from $W$ by removing the end $a$ from $p$, decreasing the capacity of $p$ by $1$ and adding a new propagator of capacity $1$ whose only end is $a$.  In this case the move is called \emph{split}. This is illustrated in the following picture.

$$	
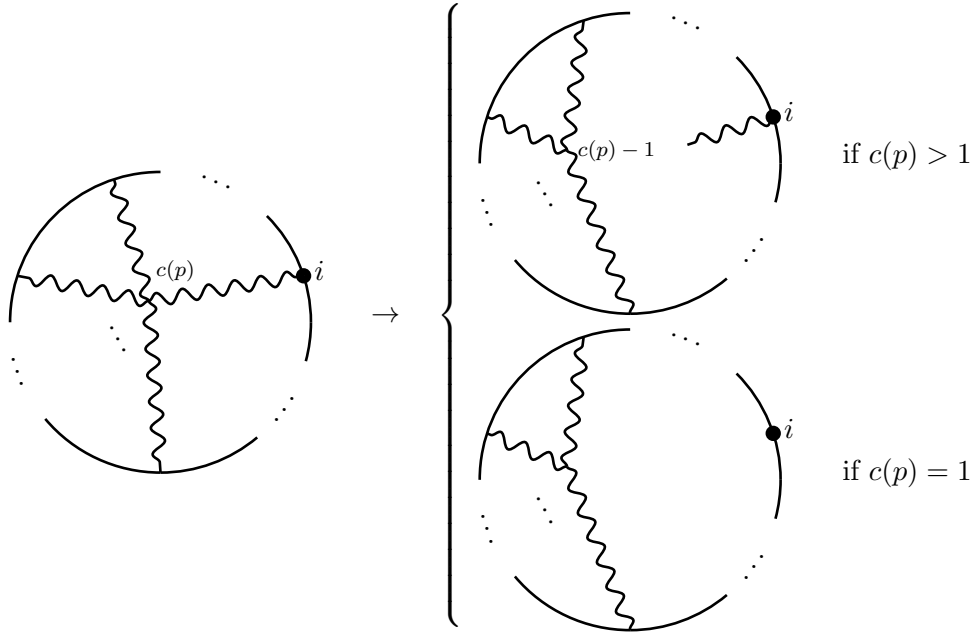
\begin{tikzpicture}[line width=1, scale=2,baseline=0]
	\draw (-1,0) arc (180:90:1);
	\draw (1,0) arc (0:45:1);
	\draw (1,0) arc (0:-15:1);
	\draw (0,-1) arc (270:310:1);
	\draw (0,-1) arc (270:220:1);
	\def \n {10}
	\foreach \x/\y in {0.951/0.309,0/-1.00000000000000,-0.309/0.951,-0.951/0.309}{\draw[decorate,decoration={snake}] (-0.077,0.142) -- (\x,\y);}
	\draw (0.1,0.2)node[above] {\scriptsize $c(p)$};
	\draw (360*6.5/\n-360/\n+180:1)circle(.4pt)circle(.8pt)circle(1.2pt) node[anchor=360/\n*6.5-360/\n]{$i$};
	\draw (360*1.5/\n-360/\n+180:1) node[rotate=360*1.5/\n-360/\n+90] {$\cdots$};
	\draw (360*7.875/\n-360/\n+180:1) node[rotate=360*7.875/\n-360/\n+90] {$\cdots$};
	\draw (360*5/\n-360/\n+180:1) node[rotate=360*5/\n-360/\n+90] {$\cdots$};
	\draw (360*1.5/\n-360/\n+180:.3) node[rotate=360*1.5/\n-360/\n+100] {$\cdots$};
\end{tikzpicture} \quad  \rightarrow \quad \begin{cases}
	\begin{tikzpicture}[line width=1, scale=2,baseline=0]
		\draw (-1,0) arc (180:90:1);
		\draw (1,0) arc (0:45:1);
		\draw (1,0) arc (0:-15:1);
		\draw (0,-1) arc (270:310:1);
		\draw (0,-1) arc (270:220:1);
		\def \n {10}
		\foreach \x/\y in {0/-1.00000000000000,-0.309/0.951,-0.951/0.309}{\draw[decorate,decoration={snake}] (-0.420,0.087) -- (\x,\y);}
		\draw (-0.420,0.087) node[right] {\scriptsize $c(p)-1$};
		\draw[decorate,decoration={snake}] (0.951,0.309) -- ++({atan(0.325}:-.6 cm);
		\draw (360*6.5/\n-360/\n+180:1)circle(.4pt)circle(.8pt)circle(1.2pt) node[anchor=360/\n*6.5-360/\n]{$i$};
		\draw (360*1.5/\n-360/\n+180:1) node[rotate=360*1.5/\n-360/\n+90] {$\cdots$};
		\draw (360*7.875/\n-360/\n+180:1) node[rotate=360*7.875/\n-360/\n+90] {$\cdots$};
		\draw (360*5/\n-360/\n+180:1) node[rotate=360*5/\n-360/\n+90] {$\cdots$};
		\draw (360*1.5/\n-360/\n+180:.6) node[rotate=360*1.5/\n-360/\n+100] {$\cdots$};
	\end{tikzpicture} & \text{ if } c(p)>1 \\
	\begin{tikzpicture}[line width=1, scale=2,baseline=0]
		\draw (-1,0) arc (180:90:1);
		\draw (1,0) arc (0:45:1);
		\draw (1,0) arc (0:-15:1);
		\draw (0,-1) arc (270:310:1);
		\draw (0,-1) arc (270:220:1);
		\def \n {10}
		\foreach \x/\y in {0/-1.00000000000000,-0.309/0.951,-0.951/0.309}{\draw[decorate,decoration={snake}] (-0.420,0.087) -- (\x,\y);}
		\draw (360*6.5/\n-360/\n+180:1)circle(.4pt)circle(.8pt)circle(1.2pt) node[anchor=360/\n*6.5-360/\n]{$i$};
		\draw (360*1.5/\n-360/\n+180:1) node[rotate=360*1.5/\n-360/\n+90] {$\cdots$};
		\draw (360*7.875/\n-360/\n+180:1) node[rotate=360*7.875/\n-360/\n+90] {$\cdots$};
		\draw (360*5/\n-360/\n+180:1) node[rotate=360*5/\n-360/\n+90] {$\cdots$};
		\draw (360*1.5/\n-360/\n+180:.6) node[rotate=360*1.5/\n-360/\n+100] {$\cdots$};
	\end{tikzpicture}& \text{ if } c(p)=1
\end{cases}
$$

Note that if there is another propagator, $q$, supported on the vertex $v_i$, and the capacity of $p$ is greater than one, then the image of the retract or split move may not satisfy the drawing convention laid out in Remark \ref{rmk:drawing conventions}. However, as shown in Example \ref{eg: half props clear}, there is an equivalent representation of the image where $q$ is not supported on $v_i$. 

Since $W$ has no crossing propagators, there are no crossing propagators in the image of the retract or split move.

\subsection{Diagrammatic boundary moves realize positroid boundaries}
\label{sec:movesrealizeboundaries} 

In this section, we show that, when applied to capacity ranked and locally minimal diagrams,the diagrammatic moves defined in the previous section all give rise either to boundaries of the positroids associated to the original generalized Wilson loop diagram, or result in an equivalent representation of the diagram. Along the way, we show that if, after applying the moves to a capacity ranked generalized Wilson loop diagram, the rank of the matroid associated to the resulting diagram is the same as the capacity of the original diagram, then the same conclusion holds. Recall that the boundary moves can only be applied to generalized Wilson loop diagrams with non-crossing propagators. The version of the result with the rank assumption on the resulting diagram is algorithmically beneficial as checking the rank on the resulting diagram is faster than checking local minimality, as we will discuss later.

For the purposes of this section, let $W$ be the generalized Wilson loop diagram to which the boundary moves are applied, i.e.\ $W$ has no crossing propagator, and let $W'$ be the diagram resulting from the application of such a move. The structure of this section is as follows. First we show that if $W$ is locally minimal generalized Wilson loop diagram, then after the application of a boundary move, the rank of the resulting matroid does not drop: $\cM(W) \leq \cM(W')$. Then we show that that for each boundary move, the basis sets of $\cM(W')$ are independent in $\cM(W)$. These two results combined show that the rank of the matroids associated to $W$ and $W'$ have the same rank: $\rk \cM(W) = \rk \cM(W')$. 
This implies that $W'$ is also capacity ranked, and so, since the moves preserve non-crossingness, the matroid associated to $W'$ is also a positroid. Then we put these results together to show that if $W$ is capacity ranked and either $W$ is locally minimal, or $\rk \cM(W) = \rk \cM(W')$, then the positroid cell associated to $\cM(W')$ is a boundary of the cell associated to $\cM(W)$. 

We begin by setting some notation that we use throughout the proofs of the above results. Let $U \subset [n]$ be the set of vertices in $W$ affected by a boundary move. Namely, \begin{itemize}
	\item for a slide move on the edge $e_i$, $U = \{i, i+1\}$ 
	\item for a merge move on the edge $e_i$, $U = \{i, i+1\}$
	\item for a merge move on the vertex $v_i$, $U = \{i\}$
	\item for the merge and split move on the edge $e_i$, $U = \{i, i+1\}$
	\item for the X - merge and split move on the edges $e_{i-1}$ and $e_i$, $U = \{i-1, i, i+1\}$
	\item for a retract or split move at the vertex $v_i$, $U = \{i\}$.
\end{itemize}

Note that, for each move, by construction, the matroids associated to the diagrams, restricted to the complement of $U$ are equal: \ba \cM(W)|_{U^c} = \cM(W')|_{U^c} \label{eq:restricted matroids same} \;.\ea This follows from Lemma \ref{res:restrictedmatroid} and the following equality of bipartite graphs: $G_{W'}^{\cE'}|_{U^c} = G_{W}^\cE|_{U^c}$. 

We show that after applying a boundary move to a locally minimal diagram, the rank of matroid associated to the resulting diagram is no less than the rank of the original diagram.

\begin{prop} \label{res:rank does not drop}
	Let $W$ be a locally minimal generalized Wilson loop diagram with non-crossing propagators. Let $W'$ be the diagram resulting from the application of a boundary move to $W$. Then \bas \rk \cM(W) \leq \rk \cM(W'). \eas
\end{prop}

\begin{proof}
	With notation as above, since $W$ is locally minimal, the only case where an element $x \in U$ could be a co-loop of $\cM(W)$ is in the retract or split move when $c(p)>1$. To see this, note that in all other cases, either $x$ doesn't satisfy $\End(x) = \{a\}$ and $V(a) = x$ or the propagator that $a$ belongs to has capacity $1$. Thus by Lemma \ref{res:non min and co-loops}, $x$ is not a co-loop.
	
	If the restricted matroid $\cM(W)|_{U^c}$ has the same rank as the full matroid, then we are done, since restricting a matroid can only decrease the rank. That is, in this case, by \eqref{eq:restricted matroids same}, \bas \rk \cM(W) = \rk \cM(W)|_{U^c} = \rk \cM(W')|_{U^c} \leq \rk \cM(W')\;.  \eas 
	
 It remains to check the when $\rk \cM(W)|_{U^c} < \rk \cM(W)$.  Let $B$ be a basis of $\cM(W)|_{U^c}$. We extend this to a basis $B \cup V$ of $\cM(W)$, with $V\subseteq U$. If the elements of $U$ are not co-loops, we know that $B \cup V$ cannot be $B \cup U$, otherwise, by Definition \ref{dfn:matroid restriction} the elements of $U$ would be in every basis of $\cM(W)$, contradicting the fact that elements of $U$ are not co-loops. Therefore, if the elements of $U$ are not co-loops $V\subsetneq U$.
	
	Next, for each of the moves, we show that either we must have $\rk \cM(W)|_{U^c} = \rk \cM(W)$, or that we can choose the basis extension so that $B \cup V$ is both a basis of $\cM(W)$ and is independent in $\cM(W')$, showing that $ \rk \cM(W) \leq \rk \cM(W')$.  

    For the remainder of this argument use the notation for the propagators and ends involved in the move as established in Section~\ref{sec:moves}.
    
 First consider the retract or split move where $U = i$ is a co-loop. By local minimality implies that $c(p)>1$. By the arguments in Lemma \ref{res:non min and co-loops}, the diagram resulting from the split move has the same matroid, $\cM(W)$, and we are done. 
	
 For all other configurations, the vertices of $U$ are not co-loops. If $|U| = 1$, as in the case of the retract and split (when $U$ is not a co-loop) or the merge at a vertex, since $B' \subsetneq B \cup U$, we have that $B' = B$, and therefore we must be in the case where the two ranks are the same, which has been addressed above. 

	It remains to check the cases where $|U| > 1$. 
	\begin{itemize}
		\item \textbf{Slide move} We know that $V$ is a singleton consisting of one vertex of $U$. Let $M$ be an independent matching of $B \cup V$ in $G_W^\cE$. Note that since the other vertex of $U$ is not used in $M$, if $a$ is matched to the vertex of $V$, it can  instead be matched to the other vertex of $U$.  Therefore, doing so if necessary so that $V$ is the end $a$ is slid to, $M$ is also an independent matching in $G_{W'}^{\cE'}$. Thus $B \cup V$ is independent in $\cM(W')$.
		\item \textbf{Merge move on $e_i$} In this case $V$ is the singleton set consisting of one vertex from $U$. Let $M$ be an independent matching of $B \cup V$ in $G_W^\cE$. If $V$ is matched to $p_e$ or $q_e$ in $M$, replace this with an edge from $V$ to the end $r_e$. This gives an independent $B \cup V$-perfect matching in $G_{W'}^{\cE'}$. Thus $B \cup V$ is independent in $\cM(W')$.
		\item \textbf{Merge and split move} We know that $V$ is a singleton consisting of one vertex of $U$.  Let $M$ be an independent matching of $B \cup V$ in $G_W^\cE$. Note that since the other vertex of $U$ is not used in $M$, if an end supported on $e_i$ is matched to the vertex of $V$, it can instead be matched to the other vertex of $U$. Therefore, doing so if necessary so that $V$ is the vertex in $U$ supporting $r$ in $W'$, we can adjust $M$ to a matching for $W'$ as follows. If $V$ is matched to and end supported on $e_i$ in $M$, replace this with an edge from $V$ to the single end of $s$. This gives an independent $B \cup V$-perfect matching in $G_{W'}^{\cE'}$. Thus $B \cup V$ is independent in $\cM(W')$.
		\item \textbf{X-merge and split move} 
        We know that $V$ is a subset of $U$ of size $1$ or $2$.  Let $M$ be an independent matching of $B \cup V$ in $G_W^\cE$.  If $V=\{i-1, i+1\}$ one of which is matched to an end of $p$ or $q$, then since $i$ is not used in $M$ either the end of $p$ or $q$ matched to $V$ can be instead matched to $i$.  Therefore, doing so if necessary, we can adjust the matching for $W'$ as follows.  If $i$ is matched to an end of $p$ or $q$, replace this with a edge from $i$ to the end of $s$, anything else matched to $p$ or $q$ is then matched to the corresponding end $r$.  At most one end of $r$ is matched by this process in view of the reduction when $V=\{i-1, i+1\}$.  This gives an independent $B \cup V$-perfect matching in $G_{W'}^{\cE'}$. Thus $B \cup V$ is independent in $\cM(W')$.
	\end{itemize}

Thus we have shown that in every single case, $ \rk \cM(W) \leq \rk \cM(W')$.
\end{proof}

Next, we show that for each of the moves, every basis of $\cM(W')$ is independent in $\cM(W)$. This gives the key towards showing the boundary result.

\begin{prop}\label{res:rank does not increase}
	Let $W$ be a generalized Wilson loop diagram with no crossing propagators. Let $W'$ be the diagram resulting from applying any boundary move. Then every basis of $\cM(W')$ is independent in $\cM(W)$. 
\end{prop}

\begin{proof}		
	By \eqref{eq:restricted matroids same}, we know that the matroids associated to $W$ and $W'$, restricted to the complement of $U$ are equal. Let $B$ be a basis of $\cM(W')$ and let $M$ be a $B$-perfect independent matching in $G_{W'}^{\cE'}$. We know that the subset of $M$ that covers vertices in $U^c$ is independent in both $G_{W'}^{\cE'}$ and $G_{W}^{\cE}$.  For each boundary move we show how $M$ can be changed to form a $B$-perfect independent matching in $G_W^{\cE}$. 
	
	Since $M$ is $B$-perfect and independent in $G_{W'}^{\cE'}$, we know that, for any propagator $p$ in $W'$, $M$ matches at most $c(p)$ ends of $p$ in $W'$. The proofs that the matchings in $G_W^{\cE}$ are $B$-perfect and independent comes from checking that, for any propagator $p$ in $W$, $M$ matches at most $c(p)$ ends of $p$ in $W$.

     For the remainder of this argument use the notation for the propagators and ends involved in each move as established in Section~\ref{sec:moves}.
     
	\begin{itemize}
		\item \textbf{Slide move}: Any $B$-perfect independent matching  in $G_{W'}^{\cE'}$ is a $B$-perfect independent matching in $G_W^{\cE}$, since the supports of corresponding propagators are subsets.
		\item \textbf{Merge move}: If $M$ matches $v \in U$ to the end $r_e$,  then, either $M$ matches at most $c(p)$ vertices in $B \setminus v$ to ends of $p$ and at most $c(q) -1 $ vertices in $B \setminus v$ to ends of $q$ in $G_W^{\cE}$ or analogously with $p$ and $q$ reversed. Then in $G_W^{\cE}$, match the vertex $v$ with the end $q_e$ or $p_e$ as appropriate. All other edges in the matching $M$ remain the same. 
		\item \textbf{Merge and split move}: All edges of $M$ that do not cover vertices in $U$ do not change between $G_W^{\cE}$ and $G_{W'}^{\cE}$. 
			\begin{itemize} 
			 	\item If $M$ matches $v \in U$ to the single end of $s$ or to the end $r_e$, then, in $G_W^{\cE}$, $M$ either matches at most $c(p)$ vertices in $B \setminus v$ to ends of $p$ and at most $c(q) -1 $ vertices in $B \setminus v$ to ends of $q$ in $G_W^{\cE}$ or analogously with $p$ and $q$ reversed. Then in $G_W^{\cE}$, match the vertex $v$ with the end $q_e$ or $p_e$ as appropriate.
			 	\item If $M$ matches the two vertices in $U$ both to the single end of $s$ and to the end $r_e$, then, in $G_W^{\cE}$, $M$ matches at most $c(p)-1$ vertices in $B \setminus v$ to ends of $p$ and at most $c(q) -1 $ vertices in $B \setminus v$ to ends of $q$ in $G_W^{\cE}$. Then in $G_W^{\cE}$, match the vertices of $U$ to the ends $p_e$ and $q_e$.
			\end{itemize} 
		\item \textbf{X-Merge and Split}: The matching $M$ must match $i$ to $s$ in $G_{W'}^{\cE'}$. 
		\begin{itemize} 
			\item If $i+1$ and $i-1$ are not matched to an end of $r$ in $G_{W'}^{\cE'}$, then there at most one vertex in $U^c$ that is matched to an end of $r$. If that end corresponds to an end of $p$ in $G_{W}^{\cE}$ then match $i$ to an end of $q$. Otherwise, match $i$ with an end of $p$. If no end of $r$ is matched to an element of $U^c$, match $i$ to either the end of $p$ or $q$ supported on $i$. Leave all other edges of $M$ to be the same in the two graphs.
			\item Otherwise, exactly one of $i-1$ or $i+1$ is matched to an end of $r$ in $G_{W'}^{\cE'}$. Depending on which is matched, match $i-1$ to an end of $p$ and $i$ to an end of $q$ or $i+1$ to and end of $q$ and $i$ to an end of $p$ in $G_{W}^{\cE}$. Leave all other edges of $M$ to be the same in the two graphs.
		\end{itemize}
		\item \textbf{Retract or Split}: For any matching independent $B$-perfect matching $M$ in $G_{W'}^{\cE'}$, keep all edges of $M$ covering vertices of $U^c$ the same.
		\begin{itemize}
			\item If $k = c(p) > 1$ in $W$, then the matching $M$ must match $i$ to the single ended propagator on $i$ in $G_{W'}^{\cE'}$.  Furthermore, there are at most $k -1$ vertices of $U^c$ matched to ends of $p$ in $G_{W'}^{\cE'}$. Therefore, in $G_{W}^{\cE}$, match $i$ to an end of $p$.
		\end{itemize} If $k = c(p) > 1$ in $W$, $i$ is not covered by the matching $M$. 
	\end{itemize}
	
	Therefore, we have shown that if $B$ is a basis of $\cM(W')$ then it is an independent set of in $\cM(W')$. 
	
	\end{proof} 

We can put Propositions \ref{res:rank does not drop} and  \ref{res:rank does not increase} together to show that the ranks of $\cM(W)$ and $\cM(W')$ are the same.

\begin{cor} \label{res:ranks same}
	If $W$ is a locally minimal generalized Wilson loop diagram with non-crossing propagators, the ranks of $\cM(W)$ and $\cM(W')$ are the same: $\rk \cM(W) = \rk \cM(W')$.
\end{cor}

\begin{proof}
	Proposition \ref{res:rank does not increase} implies that $\rk \cM(W') \leq \rk \cM(W)$, so  combined with Proposition \ref{res:rank does not drop}, gives $\rk \cM(W) = \rk \cM(W')$.
\end{proof}

We are now ready to prove the main result of this section, that the positroid associated to $\cM(W')$ is a boundary of the positroid associated to $\cM(W)$ under certain conditions on $W$ or $W'$. Before stating this result, we recall from Definition \ref{dfn:positriodboundary} that a positroid cell defined by the matroid $M'$ is a boundary of the positroid cell defined by the matroid $M$ if and only if both positroids have the same rank ($\rk M = \rk M'$) and the set of bases of one is contained in the set of bases of the other: $\cB(M') \subseteq \cB(M)$. 

\begin{thm}\label{res:boundary cells}
Let $W$ be a capacity ranked generalized Wilson loop diagram with no crossing propagators and $W'$ the diagram resulting from the application of a boundary move to $W$.  If either $W$ is locally minimal or $\rk \cM(W) = \rk \cM(W')$, then the positroid cell associated to the $\cM(W')$ is a boundary of the positroid cell associated to  $\cM(W)$.
\end{thm}
\begin{proof}
We first observe that $\cM(W)$ and $\cM(W')$ are both positroids.  For the former,  by Corollary \ref{res:non-crosspos}, since $W$ is capacity ranked and has no crossing propagators $\cM(W)$ is a positroid. By construction, for any of the boundary moves, when we begin with a non-crossing diagram then $W'$ also has no crossing propagators. Furthermore, the sum of the capacities of the propagators in $W'$ is the same as the sum of the capacities of the propagators in $W$. Therefore, either directly by hypothesis or by Corollary \ref{res:ranks same} if $W$ is locally minimal, $W'$ is also capacity ranked. Therefore, by Corollary \ref{res:non-crosspos}, $\cM(W')$ is a positroid. 

By Proposition \ref{res:rank does not increase}, we see that if $B$ is a basis of $\cM(W')$ then it is independent in $\cM(W)$. However, since both matroids have the same rank (either by hypothesis or by Corollary~\ref{res:non-crosspos}), $B$ is a basis of $\cM(W)$. In other words, there is a containment of the sets of bases: $\cB(\cM(W')) \subseteq \cB(\cM(W))$. Therefore, by Definition \ref{dfn:positriodboundary}, the positroid cell associated to the $\cM(W')$ is a boundary of the positroid cell associated to  $\cM(W)$.
\end{proof}

Algorithmically, it is more convenient to apply a move to $W$ without checking local minimality and then simply check that the rank of $\cM(W')$ is the same as that of $\cM(W)$ as this simply involves applying Algorithm \ref{algo:non-crossgwld} to $W'$ once, while checking local minimality involves, for each edge of $G_{W}^{\cE}$, building the adjusted diagram and applying Algorithm \ref{algo:non-crossgwld} to it in order to check that the Grassmann necklace is unchanged from the Grassmann necklace of $W$.  The amount of work is multiplied by the number of edges of the bipartite graph.  

The code discussed in the subsequent sections uses Algorithm~\ref{algo:non-crossgwld} to generate the Grassmann necklaces.  Note that to obtain a valid Grassmann necklace via Algorithm~\ref{algo:non-crossgwld}, the generalized Wilson loop diagram must be admissible in the sense of Definition~\ref{dfn:admissiblegWLD}.  However, Algorithm~\ref{algo:non-crossgwld} also provides the test for admissibility: if all of the elements of what should be Grassmann necklace elements are not the same size, then the generalized Wilson loop diagram was not admissible.  Likewise, as noted above applying Algorithm~\ref{algo:non-crossgwld} also provides the test that the rank of $\cM(W')$ is the same as the rank of $\cM(W)$ since the size of each  Grassmann necklace element is the rank of the matroid.  In running the code, if any non-admissible diagrams were generated then an error is raised on account of the invalid Grassmann necklace. Additionally the rank of each resulting diagram is checked.  Note that neither problem ever arises in the calculations of Sections~\ref{sec:diagrammaticboundaries} and \ref{sec:boundaries of more than 2 props} since the starting diagrams are very well-behaved in those cases.

\subsection{Diagrammatic moves generate boundaries of admissible Wilson loop diagrams with two propagators}\label{sec:diagrammaticboundaries}

In this section we prove that the diagrammatic moves from Section \ref{sec:moves} generate all the boundaries of all codimensions of admissible ordinary Wilson loop diagrams with two propagators (as opposed to the previous section that demonstrated that all the moves generated boundaries). The proof is an explicit calculation that requires computer assistance, which we have implemented in the accompanying Python files and now explain in detail. 

First we describe how to check whether diagrammatic moves generate all boundaries for a fixed Wilson loop diagram. This involves enumerating all possible boundaries of the positroid associated to an admissible ordinary Wilson loop diagram (Section \ref{sec:enumeratingboundariesviaRWs}) and then comparing with the boundaries that are produced via diagrammatic operations (Section \ref{sec:boundariesfrommoves}). Then we explain why it suffices to check a finite set of admissible ordinary Wilson loop diagrams to claim our result. 

\subsubsection{Enumerating the boundaries of the positroid cell associated to a Wilson loop diagram}
\label{sec:enumeratingboundariesviaRWs}
Let $W$ be a fixed admissible ordinary Wilson loop diagram. For the moment we will not make any assumptions about its vertices or propagators. By construction, these Wilson loop diagrams do not have crossing propagators. Therefore, by Corollary \ref{res:non-crosspos} the associated matroid $\cM(W)$ is a positroid. As mentioned briefly in Section \ref{sec:background}, there are many equivalent combinatorial structures that encode positroids. The most critical such structure for us is the reduced word/subword pair of a permutation, the reason being that there exists a straightforward algorithm to compute the reduced word/subword pairs corresponding to boundaries of the original positroid \cite{Towers}. As far as we are aware, this is the only combinatorial representation that gives us direct access to the positroid boundaries.

However, it is not clear how to compute the boundaries of the positroid associated to a Wilson loop diagram directly from the diagram itself. Instead, starting from a Wilson loop diagram we compute the reduced word/subword pair by passing through several intermediate combinatorial objects, namely Grassmann necklaces and Le diagrams. In this manner one can enumerate the reduced word/subword pairs representing all boundaries of a Wilson loop diagrams. 

Furthermore, among these combinatorial objects (Wilson loop diagrams, Grassmann necklaces, Le diagrams, and reduced word/subword pairs), it is only the Le diagrams from which one can readily read off the {\it dimension} of the boundary.\footnote{For the reader already familiar with Le diagrams, the dimension is simply the number of nonzero boxes in the tableau.} In other words, one may know that a word/subword pair is a boundary of another, but not of what codimension. From a geometric point of view it is useful to organize the boundaries of a positroid by codimension, and therefore as a last step in enumerating the boundaries of a Wilson loop diagram we convert them back to Le diagrams.

It is this sequence of steps that makes it possible to compare the (Le diagram) boundaries represented by generalized Wilson loop diagram obtained from diagrammatic operations on $W$ with the full list of (Le diagrams representing) boundaries of the positroid associated to $W$. As this procedure is somewhat convoluted we have provided a visual aid for the reader in the form of Figure~\ref{fig:wldboundariesfromRWs}.

\begin{figure}[h]

\begin{tikzcd}
\begin{tikzpicture}[line width=1]
\node[draw,rectangle, minimum size=1.5cm,align=center] at (0,0) {Wilson loop\\ diagram};
\end{tikzpicture} \ar[drr,dashed,swap, "\text{Enumerate boundaries }"] \ar[r] & \begin{tikzpicture}[line width=1]
\node[draw,rectangle, minimum size=1.5cm,align=center] at (0,0) {Grassmann\\ necklace};
\end{tikzpicture} \ar[r] & \begin{tikzpicture}[line width=1]
\node[draw,rectangle, minimum size=1.5cm,align=center] at (0,0) {Le\\ diagram};
\end{tikzpicture} \ar[r] & \begin{tikzpicture}[line width=1]
\node[draw,rectangle, minimum size=1.5cm,align=center] at (0,0) {Reduced word/subword \\ pair of permutation};
\end{tikzpicture} \ar[d,"\text{Enumerate boundaries}"] \\
& & \begin{tikzpicture}[line width=1]
\node[draw,rectangle, minimum size=1.5cm,align=center] at (0,0) {Le\\ diagram};
\end{tikzpicture} & \ar[l] \begin{tikzpicture}[line width=1]
\node[draw,rectangle, minimum size=1.5cm,align=center] at (0,0) {Reduced word/subword \\ pair of permutation};
\end{tikzpicture}
\end{tikzcd}
\caption{Enumerating the Le diagrams corresponding to the boundaries of the positroid associated to a Wilson loop diagram requires passing through several intermediate combinatorial objects.}
\label{fig:wldboundariesfromRWs}
\end{figure}
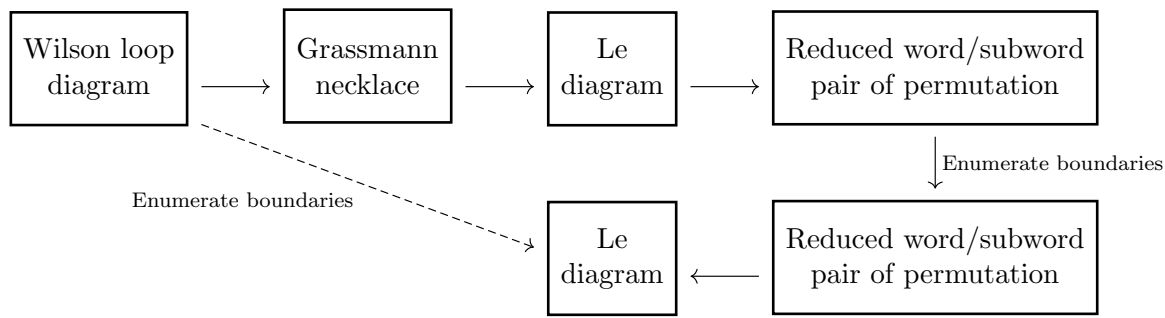

In the Python code accompanying this article a list of Le diagrams representing boundaries of a Wilson loop diagram $W$ can be computed by calling the function

\begin{lstlisting}[language=Python]
LeBoundariesByCodim(W)
\end{lstlisting}
which outputs a list whose $i$th entry is a list of Le diagrams corresponding to codimension $i$ boundaries of $W$. 

\subsubsection{Computing boundaries of a positroid cell via diagrammatic operations on generalized Wilson loop diagrams}\label{sec:boundariesfrommoves}

Given an admissible ordinary Wilson loop diagram $W$ of dimension $d$, it is straightforward to compute the generalized Wilson loop diagrams which correspond to codimension 1 boundaries. First we apply all possible diagrammatic moves in all possible ways (slide, merge, merge and split, X-merge and split, and retract or split) to $W$ to obtain a list of generalized Wilson loop diagrams. Then we throw out any diagrams with dimension not equal to $d-1$. 
This achieves two aims, first if any move does not change the dimension, that is the move results in another representation of the starting diagram, then it is removed, though the boundary calculus is sufficiently well-behaved on admissible ordinary Wilson loop diagrams that this does not occur in our main computations.  Second, this removes any boundary with codimension larger than 1. Later when we iterate this process, this will ensure that all of the boundaries are obtained by a sequence of moves where the codimension drops by 1 at each step, which is not strictly necessary, but is convenient.

As mentioned briefly in Section \ref{sec:enumeratingboundariesviaRWs}, it is not clear how to read off the dimension directly from the diagram. Therefore, discarding the diagrams that are not codimension 1 boundaries requires converting all the boundary diagrams to Le diagrams as an intermediate step. This leaves us with a list of generalized Wilson loop diagrams corresponding to codimension 1 boundaries of the positroid $\cM(W)$, although not necessarily uniquely, see Remark \ref{rmk:nonuniquenessofgWLDsasboundaries}. 

The function \lstinline{Codim1DiagrammaticBoundaries} implements the aforementioned calculation of generalized Wilson loop diagrams corresponding to codimension 1 boundaries in Python. 

\begin{lstlisting}[columns=fullflexible,language=Python]
def Codim1DiagrammaticBoundaries(W):
    # compute the dimension of W
    d = dim(W) 
    # list the generalized Wilson loop diagrams obtained from applying moves
    diagrammatic_boundaries = AllSlideMoves(W)+AllMergeMoves(W)+AllMergeAndSplitMoves(W)+AllXMergeAndSplitMoves(W)+AllRetractOrSplitMoves(W)
    # remove diagrams which fail to have codimension 1
    codim1_diagrammatic_boundaries = [gWLD for gWLD in diagrammatic_boundaries if my_flatten(gWLDtoLe(gWLD)).count(1) == d-1]
    # remove diagrams whose Le diagram fails to have the right rank
    codim1_boundary_gWLDs_with_correct_rank = [bdyWLD for bdyWLD in codim1_diagrammatic_boundary_gWLDs if len(gWLDtoLe(bdyWLD)) == rank]
    return codim1_diagrammatic_boundaries
\end{lstlisting}

The generalized Wilson loop diagrams corresponding to codimension 2 boundaries are then computed recursively from those corresponding to codimension 1 boundaries in the same manner, and so on until one obtains generalized Wilson loop diagrams corresponding to codimension $d$ boundaries. 
\begin{lstlisting}[language=Python]
def DiagrammaticBoundariesByCodim(gWLD): 
    # starting at codimension 0
    d = dim(gWLD)
    bdyWLDs = [[gWLD]]
    while d > 0:
        fixed_dim_bdyWLDs = []
        for bdy in bdyWLDs[-1]:
            # compute the next higher codimension diagrams
            fixed_dim_bdyWLDs += Codim1DiagrammaticBoundaries(bdy)
        # throw out multiple instances of the same diagram
		DeleteDuplicates(fixed_dim_bdyWLDs)
        # add the list of codimension 1 diagrams to the last place in the list 
		  bdyWLDs += [fixed_dim_bdyWLDs]
        # increase the codimension by 1
		d -= 1
	return bdyWLDs
\end{lstlisting}

\begin{rmk}
\label{rmk:nonuniquenessofgWLDsasboundaries}
    In general, different sequences of diagrammatic moves can result in the same generalized Wilson loop diagram. For example, both ways of sliding one end of a propagator to and then retracting it from a vertex will always yield the same boundary:
    $$
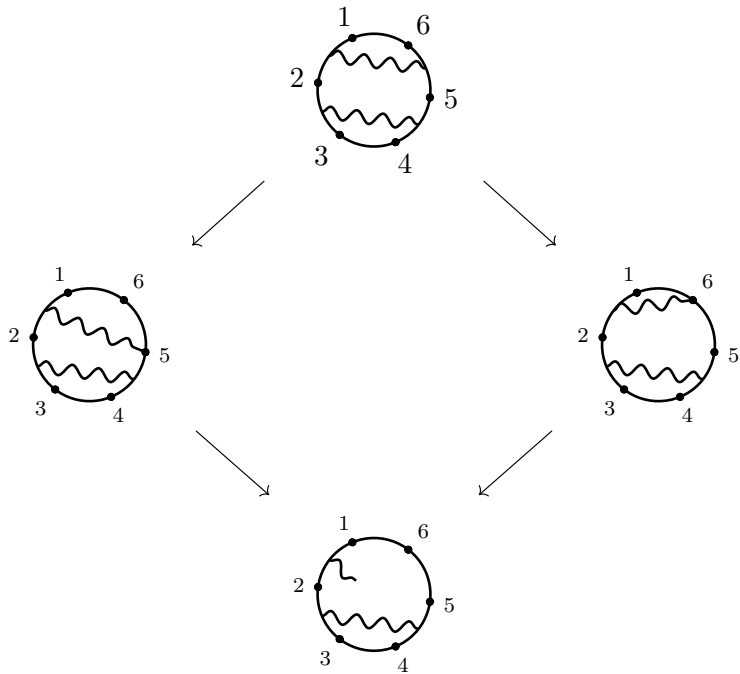
\begin{tikzcd}
        & \begin{tikzpicture}[rotate=-67.5, line width=1, scale=.75,baseline=0]
 \def \n {6}
 \draw circle(1) 
 \foreach \v in{1,...,\n} 
 {(360*\v/\n-360/\n+180:1)circle(.4pt)circle(.8pt)circle(1.2pt)circle(1.4pt) node[anchor=360/\n*\v-360/\n-67.5]{$\v$}};
 \draw[decorate,decoration={snake,amplitude=0.8mm}] (-0.866,-0.500) -- (0,1.00000000000000);
 \draw[decorate,decoration={snake,amplitude=0.8mm}] (0,-1.00000000000000) -- (0.866,0.500);
\end{tikzpicture} \ar[dl] \ar[dr] & \\\begin{tikzpicture}[rotate=-67.5, line width=1, scale=0.75]
 \def \n {6}
 \draw circle(1) 
 \foreach \v in{1,...,\n} 
 {(360*\v/\n-360/\n+180:1)circle(.4pt)circle(.8pt)circle(1.2pt)circle(1.4pt) node[anchor=360/\n*\v-360/\n-67.5]{$\scriptstyle \v$}};
 \draw[decorate,decoration={snake,amplitude=0.8mm}] (-0.866,-0.500) -- (0.500,0.866);
 \draw[decorate,decoration={snake,amplitude=0.8mm}] (0,-1.00000000000000) -- (0.866,0.500);
\end{tikzpicture} \ar[dr] &  &\begin{tikzpicture}[rotate=-67.5, line width=1, scale=0.75]
 \def \n {6}
 \draw circle(1) 
 \foreach \v in{1,...,\n} 
 {(360*\v/\n-360/\n+180:1)circle(.4pt)circle(.8pt)circle(1.2pt)circle(1.4pt) node[anchor=360/\n*\v-360/\n-67.5]{$\scriptstyle \v$}};
 \draw[decorate,decoration={snake,amplitude=0.8mm}] (-0.866,-0.500) -- (-0.500,0.866);
 \draw[decorate,decoration={snake,amplitude=0.8mm}] (0,-1.00000000000000) -- (0.866,0.500);
\end{tikzpicture} \ar[dl] \\
& \begin{tikzpicture}[rotate=-67.5, line width=1, scale=0.75]
 \def \n {6}
 \draw circle(1) 
 \foreach \v in{1,...,\n} 
 {(360*\v/\n-360/\n+180:1)circle(.4pt)circle(.8pt)circle(1.2pt)circle(1.4pt) node[anchor=360/\n*\v-360/\n-67.5]{$\scriptstyle \v$}};
 \draw[decorate,decoration={snake,amplitude=0.8mm}] (-0.866,-0.500) -- ++({atan(0.577)}:.6 cm);
 \draw[decorate,decoration={snake,amplitude=0.8mm}] (0,-1.00000000000000) -- (0.866,0.500);
\end{tikzpicture} & 
    \end{tikzcd}$$
 Furthermore, different generalized Wilson loop diagrams may represent the same boundary. Take for example the $6$-dimensional Wilson loop diagram
$$
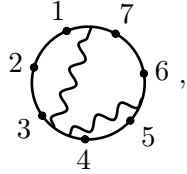
\begin{tikzpicture}[rotate=-67.5, line width=1, scale=.75,baseline=0]
 \def \n {7}
 \draw circle(1) 
 \foreach \v in{1,...,\n} 
 {(360*\v/\n-360/\n+180:1)circle(.4pt)circle(.8pt)circle(1.2pt)circle(1.4pt) node[anchor=360/\n*\v-360/\n-67.5]{$\v$}};
 \draw[decorate,decoration={snake,amplitude=0.8mm}] (0.500,-0.866) -- (-0.901,0.434);
 \draw[decorate,decoration={snake,amplitude=0.8mm}] (0.733,-0.680) -- (0.623,0.782);
\end{tikzpicture},
$$
from which one can obtain the following two codimension 4 generalized Wilson loop diagrams

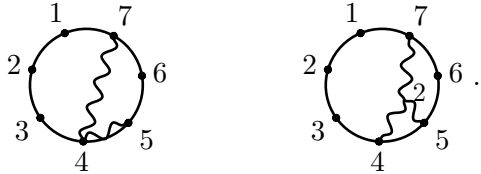
\begin{align*}
    \begin{tikzpicture}[rotate=-67.5, line width=1, scale=.75,baseline=0]
 \def \n {7}
 \draw circle(1) 
 \foreach \v in{1,...,\n} 
 {(360*\v/\n-360/\n+180:1)circle(.4pt)circle(.8pt)circle(1.2pt)circle(1.4pt) node[anchor=360/\n*\v-360/\n-67.5]{$\v$}};
 \draw[decorate,decoration={snake,amplitude=0.8mm}] (0.901,-0.434) -- (-0.623,0.782);
 \draw[decorate,decoration={snake,amplitude=0.8mm}] (0.901,-0.434) -- (0.901,0.434);
\end{tikzpicture}
\qquad &\qquad
\begin{tikzpicture}[rotate=-67.5, line width=1, scale=.75,baseline=0]
 \def \n {7}
 \draw circle(1) 
 \foreach \v in{1,...,\n} 
 {(360*\v/\n-360/\n+180:1)circle(.4pt)circle(.8pt)circle(1.2pt)circle(1.4pt) node[anchor=360/\n*\v-360/\n-67.5]{$\v$}};
 \foreach \x/\y in {0.901/-0.434,0.901/0.434,-0.623/0.782}{\draw[decorate,decoration={snake,amplitude=0.8mm}] (0.393,0.261) -- (\x,\y);}
 \draw (0.393,0.261) node[shift = {(.2,.1)}] {\small 2};
\end{tikzpicture}.
\end{align*}

By applying Algorithm \ref{algo:non-crossgwld} one can check that these two generalized Wilson loop diagrams correspond to the same Grassmann necklace and hence boundary.

\end{rmk}

The structure of algebraic relations between sequences of diagrammatic moves giving the same diagram or the same boundary remains a subject for future investigation, see Section~\ref{sec:finale}.
However,
in order to determine whether the diagrammatic moves generate all boundaries in a given situation and hence to prove Theorem~\ref{thm:boundarymovesgenerate} we do not need to characterize when the same boundary appears multiple times and so need only keep track of their associated positroids, which as in Section \ref{sec:enumeratingboundariesviaRWs} we parametrize via Le diagrams.

Calling the function \lstinline{LeBoundariesFromMovesByCodim(W)} defined below returns a codimension-sorted list of the boundaries derived from applying the graphical calculus to the ordinary Wilson loop diagram $W$. 

\begin{lstlisting}[language=Python]
def LeBoundariesFromMovesByCodim(W):
    # apply the graphical calculus to compute the boundary diagrams
	diagrammatic_boundary_gWLDs = DiagrammaticBoundariesByCodim(W)
    # create a list to be populated with Le diagrams arising from the diagrams
	boundary_Les = []
    # for each codimension
	for i in range(len(diagrammatic_boundary_gWLDs)):
		fixed_dim_boundary_Les = []
        # compute the Le diagram associated to the diagrams with that codimension
        # add it to the list of Le diagrams if it hasn't already appeared
		[fixed_dim_boundary_Les.append(gWLDtoLe(bdy)) for bdy in diagrammatic_boundary_gWLDs[i] if gWLDtoLe(bdy) not in fixed_dim_boundary_Les]
        # store the list of Le diagrams of the current codimension
		boundary_Les += [fixed_dim_boundary_Les]
	return boundary_Les
\end{lstlisting}

The output Le diagrams can be compared codimension by codimension to the complete list of boundaries described in the previous Section \ref{sec:enumeratingboundariesviaRWs}. For a fixed Wilson loop diagram the graphical calculus generates all of the boundaries using the method outlined above if

\begin{lstlisting}[language=Python]
NumLeBoundariesFromMovesByCodim(W) == NumLeBoundariesByCodim(W)
\end{lstlisting}
returns \lstinline|True|.

\begin{rmk}
We know that throwing out the higher codimension diagrams at each stage does not result in missing any boundaries in our main computation (Theorem \ref{thm:boundarymovesgenerate}) since we show that we obtain all of them through our process, namely taking iterated codimension 1 boundaries.   
\end{rmk}

The next example illustrates the method to compute boundaries via the graphical calculus described within this section.

\begin{eg}
\label{eg:boundariesfrommoves}
Consider the following admissible ordinary Wilson loop diagram $W$ with two propagators on $n=6$ vertices.
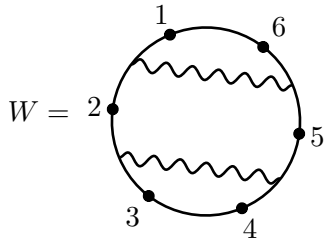
\begin{align*}
    W = \begin{tikzpicture}[rotate=-67.5, line width=1, scale=1.25,baseline=0]
 \def \n {6}
 \draw circle(1) 
 \foreach \v in{1,...,\n} 
 {(360*\v/\n-360/\n+180:1)circle(.4pt)circle(.8pt)circle(1.2pt)circle(1.4pt) node[anchor=360/\n*\v-360/\n-67.5]{$\v$}};
 \draw[decorate,decoration={snake,amplitude=0.8mm}] (-0.866,-0.500) -- (0,1.00000000000000);
 \draw[decorate,decoration={snake,amplitude=0.8mm}] (0,-1.00000000000000) -- (0.866,0.500);
\end{tikzpicture}
\end{align*}
By running \lstinline{LeBoundariesByCodim(W)} one can check that the number of boundaries at each dimension are given by the entries in the table below.
\begin{center}
\begin{tabular}{c|c|c|c|c|c|c}
Codimension & 1 & 2 & 3 & 4 & 5 & 6\\
\hline
$\#$ Boundaries of $W$ & 8 & 28 & 56 & 67 & 44 & 13
\end{tabular}
\end{center}

We exhibit a generalized Wilson loop diagram representing each boundary obtained from the sequential application of diagrammatic moves to $W$. Per Remark \ref{rmk:nonuniquenessofgWLDsasboundaries}, in general there may be multiple generalized Wilson loop diagrams representing the same boundary; we have simply picked a representative to illustrate each. While all of the codimension 1 boundaries are generated by slide moves, one can see instances of the merge, X-merge and split, and retract or split moves by inspecting the codimension 2 boundaries. Note that the regular merge and split move can never be applied to  this $W$ or to any of its boundary diagrams since there will never be two propagators ending on the same edge.

\textbf{Codimension 1 Boundaries}
\begin{align*}

\end{align*}
\end{eg}

\subsubsection{Generating boundaries of all ordinary Wilson loop diagrams with two propagators}

Now that we have explained how to analyze the boundaries of a positroid associated to an admissible ordinary Wilson loop diagram via the graphical calculus we are ready to state our second main result.

Let $\text{WLD}(k,n)$ denote the set of admissible ordinary $k$-propagator Wilson loop diagrams on $n$ vertices. 

\begin{thm}
\label{thm:boundarymovesgenerate}
The five boundary moves (slide, merge, merge and split, X-merge and split, and retract or split) applied iteratively generate all boundaries of the positroid cell associated to an admissible ordinary Wilson loop diagram with two propagators.
\end{thm}

\begin{proof}

Since the number of vertices $n$ of a Wilson loop diagram with two propagators can be arbitrarily large, there are a priori infinitely many ordinary Wilson loop diagrams with two propagators to check. However, since each propagator of an ordinary Wilson loop diagram is supported by at most 4 vertices, every ordinary Wilson loop diagram with two propagators and $n\ge 9$ vertices contains at least one vertex which does not support any propagators. Let $F$ be the set of all vertices in the diagram not supporting any propagators, which is nonempty. By Lemma \ref{res:ignorenon-supporting}, the positroid $\cM(W)$ can be written as the direct sum of matroids as follows: $\cM(W) = \cM(W|_{F^c}) + \cM(W)|_{F}$. Here, $\cM(W)|_{F}$ is a matroid of rank $0$ and $\cM(W|_{F^c})$ is the matroid associated to the diagram formed by removing all the vertices that do not support propagators from $W$. 
Since summands of rank 0 do not affect the boundaries (since they intersect trivially with all bases, see Definition~\ref{dfn:positriodboundary}), it suffices to check that the graphical calculus generates the boundaries for Wilson loop diagrams in the range $n \le 8$.

Additionally, since we are working with admissible ordinary Wilson loop diagrams, by definition, no diagram has a set of propagators $P$ with $V(P) < |P| + 3$.  Therefore, with $2$ propagators, one need not consider fewer than $5$ vertices.

Therefore, it suffices to check that the graphical calculus generates the boundaries for admissible ordinary Wilson loop diagrams in the range $5 \leq n \leq 8$, of which there are finitely many according to the following table. 

\begin{center}
\begin{tabular}{c|c|c|c|c}
$n$ & 5 & 6 & 7 & 8 \\
\hline
$\# WLD(2,$n$)$ & 5 & 21 & 56 & 120 \\ 
\end{tabular}
\end{center}

Note that since we know every move does give a boundary (Theorem~\ref{res:boundary cells}), it suffices to check that the number of boundaries that we obtain at each codimension is correct.

The result then follows from comparing the boundaries produced by the graphical calculus to the full list of boundaries for each Wilson loop diagram in this range as outlined in Section \ref{sec:boundariesfrommoves}. This can be implemented by calling the following function \lstinline{CheckMovesGenerateBoundaries()} and verifying that it returns \lstinline{True}. We did this and obtained the expected \lstinline{True} answer.

\begin{lstlisting}[language=Python]
def CheckMovesGenerateBoundaries():
    # fix the number of propagators k
    k = 2
    # for each choice of number of vertices n between 5 and 8
    for n in range(5,8):
		# enumerate all the ordinary two-propagator Wilson loop diagrams on n vertices
		WLDs_at_fixed_num_vertices = allWLDasgWLDs(k,n)
		# for every two-propagator ordinary Wilson loop diagram on n vertices
		for wld in WLDs_at_fixed_num_vertices:
			# check whether the graphical calculus produces all the boundaries
			if NumBoundariesFromMovesByCodim(gWLD) != NumLeBoundariesByCodim(gWLD)
				return False
	return True
\end{lstlisting}
\end{proof}

\subsection{Graphical calculus for boundaries of Wilson loop diagrams more than two propagators \label{sec:boundaries of more than 2 props}}

Moving to boundaries of admissible ordinary Wilson loop diagrams with $|\cP|>2$, we do not yet have sufficient moves to obtain all boundaries at all codimensions.  The moves are sufficient to generate the codimension 1 boundaries of all diagrams in $\text{WLD}(3,6)$,  $\text{WLD}(3,7)$, and the majority of diagrams in $\text{WLD}(3,8)$.  The missing boundaries from $\text{WLD}(3,8)$ then are particularly interesting.  We make some observations on them which give a complete explanation of the codimension one boundaries in $\text{WLD}(3,8)$, and which give indications of where to look for new systematic rules to extend our boundary calculus. The missing boundaries are investigated in Examples \ref{eg:missingcodim1boundary1} and \ref{eg:missingcodim1boundary2}.
 For $\text{WLD}(3,9)$ and beyond the computation as currently formulated became too large to pursue.

\subsubsection{Codimension 1 boundaries of ordinary Wilson loop diagrams with three propagators}

One can verify that not all the codimension 1 boundaries of ordinary Wilson loop diagrams with three propagators are generated by the graphical calculus with one of the functions in the accompanying Python files. The evaluation of the function \begin{lstlisting}
CheckMovesGenerateCodim1Boundaries(3,8)  
\end{lstlisting} returns \lstinline{False}, although the analogous code for \lstinline{CheckMovesGenerateCodim1Boundaries(3,6)} and \lstinline{CheckMovesGenerateCodim1Boundaries(3,7)} does return \lstinline{True}. 

We now investigate the boundary calculus on $\text{WLD}(3,8)$ further.
Among the 300 admissible ordinary Wilson loop diagrams with three capacity 1 propagators and $8$ vertices, the moves are sufficient to generate the codimension 1 boundaries for all but 32 of them. 

Each of the 32 diagrams in $\text{WLD}(3,8)$ with a codimension 1 boundary missed by the boundary moves in Section \ref{sec:moves} are equivalent -- up to rotation or reflection --  to either the diagram in the Example \ref{eg:missingcodim1boundary1} or \ref{eg:missingcodim1boundary2}.

\begin{eg}
\label{eg:missingcodim1boundary1}
Consider the following ordinary Wilson loop diagram $W$ with three propagators, each of capacity 1.

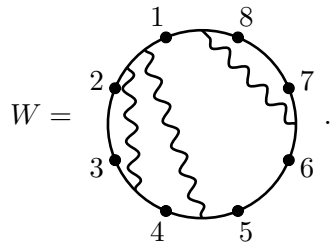
\begin{align*}
W=
    \begin{tikzpicture}[rotate=-67.5, line width=1, scale=1.25,baseline=0]
 \def \n {8}
 \draw circle(1) 
 \foreach \v in{1,...,\n} 
 {(360*\v/\n-360/\n+180:1)circle(.4pt)circle(.8pt)circle(1.2pt)circle(1.4pt) node[anchor=360/\n*\v-360/\n-67.5]{$\v$}};
 \draw[decorate,decoration={snake,amplitude=0.8mm}] (-0.966,-0.259) -- (0.924,-0.383);
 \draw[decorate,decoration={snake,amplitude=0.8mm}] (-0.866,-0.500) -- (0.383,-0.924);
 \draw[decorate,decoration={snake,amplitude=0.8mm}] (0.383,0.924) -- (-0.924,0.383);
\end{tikzpicture}.
\end{align*}
There are ten codimension 1 boundaries, only nine of which are generated by the moves in Section \ref{sec:moves}. Eight come from the slide move:
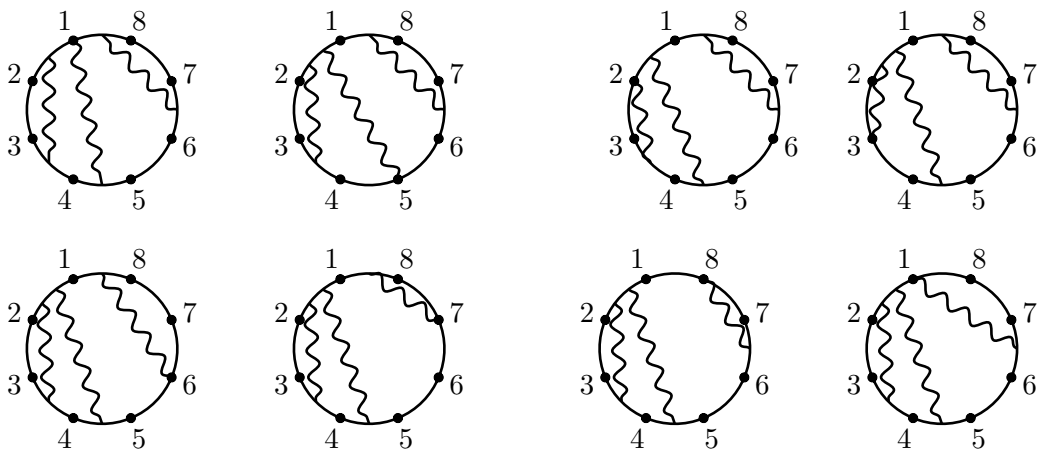
\begin{align*}
\begin{tikzpicture}[rotate=-67.5, line width=1, scale=1]
 \def \n {8}
 \draw circle(1) 
 \foreach \v in{1,...,\n} 
 {(360*\v/\n-360/\n+180:1)circle(.4pt)circle(.8pt)circle(1.2pt)circle(1.4pt) node[anchor=360/\n*\v-360/\n-67.5]{$\v$}};
 \draw[decorate,decoration={snake,amplitude=0.8mm}] (-1.00000000000000,0) -- (0.924,-0.383);
 \draw[decorate,decoration={snake,amplitude=0.8mm}] (-0.924,-0.383) -- (0.383,-0.924);
 \draw[decorate,decoration={snake,amplitude=0.8mm}] (0.383,0.924) -- (-0.924,0.383);
\end{tikzpicture}
\quad &  \quad \begin{tikzpicture}[rotate=-67.5, line width=1, scale=1]
 \def \n {8}
 \draw circle(1) 
 \foreach \v in{1,...,\n} 
 {(360*\v/\n-360/\n+180:1)circle(.4pt)circle(.8pt)circle(1.2pt)circle(1.4pt) node[anchor=360/\n*\v-360/\n-67.5]{$\v$}};
 \draw[decorate,decoration={snake,amplitude=0.8mm}] (-0.966,-0.259) -- (1.00000000000000,0);
 \draw[decorate,decoration={snake,amplitude=0.8mm}] (-0.866,-0.500) -- (0.383,-0.924);
 \draw[decorate,decoration={snake,amplitude=0.8mm}] (0.383,0.924) -- (-0.924,0.383);
\end{tikzpicture}  & \begin{tikzpicture}[rotate=-67.5, line width=1, scale=1]
 \def \n {8}
 \draw circle(1) 
 \foreach \v in{1,...,\n} 
 {(360*\v/\n-360/\n+180:1)circle(.4pt)circle(.8pt)circle(1.2pt)circle(1.4pt) node[anchor=360/\n*\v-360/\n-67.5]{$\v$}};
 \draw[decorate,decoration={snake,amplitude=0.8mm}] (-0.924,-0.383) -- (0.924,-0.383);
 \draw[decorate,decoration={snake,amplitude=0.8mm}] (-0.707,-0.707) -- (0.383,-0.924);
 \draw[decorate,decoration={snake,amplitude=0.8mm}] (0.383,0.924) -- (-0.924,0.383);
\end{tikzpicture} & \quad \begin{tikzpicture}[rotate=-67.5, line width=1, scale=1]
 \def \n {8}
 \draw circle(1) 
 \foreach \v in{1,...,\n} 
 {(360*\v/\n-360/\n+180:1)circle(.4pt)circle(.8pt)circle(1.2pt)circle(1.4pt) node[anchor=360/\n*\v-360/\n-67.5]{$\v$}};
 \draw[decorate,decoration={snake,amplitude=0.8mm}] (-0.966,-0.259) -- (0.924,-0.383);
 \draw[decorate,decoration={snake,amplitude=0.8mm}] (-0.866,-0.500) -- (0,-1.00000000000000);
 \draw[decorate,decoration={snake,amplitude=0.8mm}] (0.383,0.924) -- (-0.924,0.383);
\end{tikzpicture} \\
\begin{tikzpicture}[rotate=-67.5, line width=1, scale=1]
 \def \n {8}
 \draw circle(1) 
 \foreach \v in{1,...,\n} 
 {(360*\v/\n-360/\n+180:1)circle(.4pt)circle(.8pt)circle(1.2pt)circle(1.4pt) node[anchor=360/\n*\v-360/\n-67.5]{$\v$}};
 \draw[decorate,decoration={snake,amplitude=0.8mm}] (-0.966,-0.259) -- (0.924,-0.383);
 \draw[decorate,decoration={snake,amplitude=0.8mm}] (-0.866,-0.500) -- (0.383,-0.924);
 \draw[decorate,decoration={snake,amplitude=0.8mm}] (0.707,0.707) -- (-0.924,0.383);
\end{tikzpicture} \quad & \quad \begin{tikzpicture}[rotate=-67.5, line width=1, scale=1]
 \def \n {8}
 \draw circle(1) 
 \foreach \v in{1,...,\n} 
 {(360*\v/\n-360/\n+180:1)circle(.4pt)circle(.8pt)circle(1.2pt)circle(1.4pt) node[anchor=360/\n*\v-360/\n-67.5]{$\v$}};
 \draw[decorate,decoration={snake,amplitude=0.8mm}] (-0.966,-0.259) -- (0.924,-0.383);
 \draw[decorate,decoration={snake,amplitude=0.8mm}] (-0.866,-0.500) -- (0.383,-0.924);
 \draw[decorate,decoration={snake,amplitude=0.8mm}] (0,1.00000000000000) -- (-0.924,0.383);
\end{tikzpicture}  & \begin{tikzpicture}[rotate=-67.5, line width=1, scale=1]
 \def \n {8}
 \draw circle(1) 
 \foreach \v in{1,...,\n} 
 {(360*\v/\n-360/\n+180:1)circle(.4pt)circle(.8pt)circle(1.2pt)circle(1.4pt) node[anchor=360/\n*\v-360/\n-67.5]{$\v$}};
 \draw[decorate,decoration={snake,amplitude=0.8mm}] (-0.966,-0.259) -- (0.924,-0.383);
 \draw[decorate,decoration={snake,amplitude=0.8mm}] (-0.866,-0.500) -- (0.383,-0.924);
 \draw[decorate,decoration={snake,amplitude=0.8mm}] (0.383,0.924) -- (-0.707,0.707);
\end{tikzpicture} \quad & \quad \begin{tikzpicture}[rotate=-67.5, line width=1, scale=1]
 \def \n {8}
 \draw circle(1) 
 \foreach \v in{1,...,\n} 
 {(360*\v/\n-360/\n+180:1)circle(.4pt)circle(.8pt)circle(1.2pt)circle(1.4pt) node[anchor=360/\n*\v-360/\n-67.5]{$\v$}};
 \draw[decorate,decoration={snake,amplitude=0.8mm}] (-1.00000000000000,0) -- (0.383,0.924);
 \draw[decorate,decoration={snake,amplitude=0.8mm}] (-0.966,-0.259) -- (0.924,-0.383);
 \draw[decorate,decoration={snake,amplitude=0.8mm}] (-0.866,-0.500) -- (0.383,-0.924);
\end{tikzpicture}
\end{align*}
and one comes from a merge move:
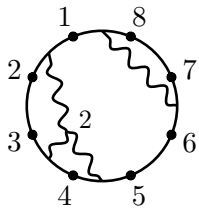
\begin{align*}
\begin{tikzpicture}[rotate=-67.5, line width=1, scale=1]
 \def \n {8}
 \draw circle(1) 
 \foreach \v in{1,...,\n} 
 {(360*\v/\n-360/\n+180:1)circle(.4pt)circle(.8pt)circle(1.2pt)circle(1.4pt) node[anchor=360/\n*\v-360/\n-67.5]{$\v$}};
 \foreach \x/\y in {-0.924/-0.383,0.383/-0.924,0.924/-0.383}{\draw[decorate,decoration={snake,amplitude=0.8mm}] (0.128,-0.563) -- (\x,\y);}
 \draw (0.128,-0.563) node[shift = {(.25,0.15)}] {\small 2};
 \draw[decorate,decoration={snake,amplitude=0.8mm}] (0.383,0.924) -- (-0.924,0.383);
\end{tikzpicture}.
\end{align*}
(One can check that any other generalized Wilson loop diagrams resulting from a single application of the moves is codimension 2 or greater.)

Let $\partial W$ denote the missing boundary, which has Grassmann necklace $$\cI = \{ 123, 231, 341, 451, 561, 612, 712, 812\}.$$

In this case, this boundary is already obtained in our calculation as a codimension one boundary of an equivalent admissible ordinary Wilson loop diagram with three propagators and 8 vertices.
Specifically, let 

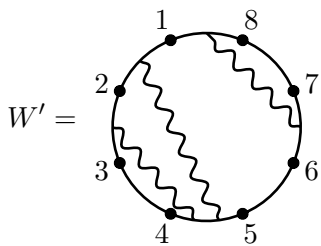
\begin{align*}
	W' = \begin{tikzpicture}[rotate=-67.5, line width=1, scale=1.25, baseline=0]
		\def \n {8}
		\draw circle(1)
		\foreach \v in{1,...,\n}
		{(360*\v/\n-360/\n+180:1)circle(.4pt)circle(.8pt)circle(1.2pt)circle(1.4pt) node[anchor=360/\n*\v-360/\n-67.5]{$\v$}};
		\draw[decorate,decoration={snake,amplitude=0.8mm}] (-0.924,-0.383) -- (0.966,-0.259);
		\draw[decorate,decoration={snake,amplitude=0.8mm}] (-0.383,-0.924) -- (0.866,-0.500);
		\draw[decorate,decoration={snake,amplitude=0.8mm}] (0.383,0.924) -- (-0.924,0.383);
	\end{tikzpicture} \;.
\end{align*} As this is an admissible ordinary Wilson loop diagram, one may check by the retriangulation arguments laid out in \cite{generalcombinatoricsI}, by using the Grassmann necklace algorithm in \cite{generalcombinatoricsII}, or by applying Algorithm \ref{algo:non-crossgwld} that the both $W$ and $W'$ give rise to the same positroid: $\cM(W)= \cM(W')$. Unlike in the diagram $W$, there is a diagrammatic move from $W'$ that gives rise to the same positroid as $\partial W$. Namely, one may apply the X-merge and split at vertex $1$: 

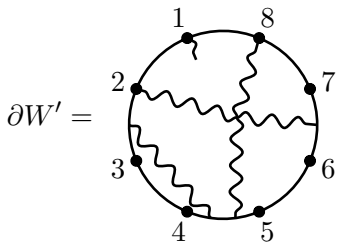
\begin{align*}
	\partial W' = \begin{tikzpicture}[rotate=-67.5, line width=1, scale=1.25, baseline = 0]
		\def \n {8}
		\draw circle(1)
		\foreach \v in{1,...,\n}
		{(360*\v/\n-360/\n+180:1)circle(.4pt)circle(.8pt)circle(1.2pt)circle(1.4pt) node[anchor=360/\n*\v-360/\n-67.5]{$\v$}};
		\draw[decorate,decoration={snake,amplitude=0.8mm}] (-1.00000000000000,0) -- ++({atan(0)}:.25 cm);
		\foreach \x/\y in {-0.707/-0.707,0.966/-0.259,0.383/0.924,-0.707/0.707}{\draw[decorate,decoration={snake,amplitude=0.8mm}] (-0.016,0.166) -- (\x,\y);}
		\draw[decorate,decoration={snake,amplitude=0.8mm}] (-0.383,-0.924) -- (0.866,-0.500);
	\end{tikzpicture}
\end{align*}

One can check from Algorithm \ref{algo:non-crossgwld} that $\cM(\partial W') = \cM(\partial W)$.

An alternative way to realize the missing boundary without passing to equivalent representations of $W$ can be found by introducing a decoration on a certain pair of propagator ends which indicates they cannot both be used in the formation of a particular Grassmann necklace element when applying Algorithm \ref{algo:non-crossgwld}. We indicate this in the following picture with an $x$ in the angle between the two ends which cannot be used together.  With this convention, the diagram below produces the missing codimension 1 boundary.

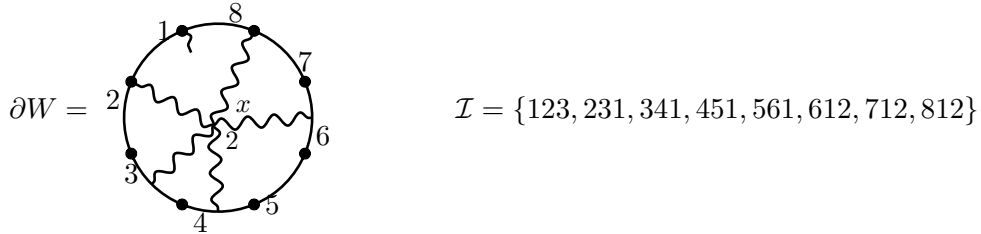
\begin{align*}
\partial W = \begin{tikzpicture}[rotate=-67.5, line width=1, scale=1.25,baseline = 0]
 \def \n {8}
 \draw circle(1) 
 \foreach \v in{1,...,\n} 
 {(360*\v/\n-360/\n+180:1)circle(.4pt)circle(.8pt)circle(1.2pt)circle(1.4pt) node[anchor=360/\n*\v-360/\n]{$\v$}};
 \draw[decorate,decoration={snake,amplitude=0.8mm}] (-1.00000000000000,0) -- ++({atan(0)}:.25 cm);
 \foreach \x/\y in {-0.707/-0.707,0.383/-0.924,0.924/-0.383,0.383/0.924,-0.707/0.707}{\draw[decorate,decoration={snake,amplitude=0.8mm}] (0.055,-0.077) -- (\x,\y);}
 \draw (0.055,-0.077) node[shift = {(.25,-.2)}] {\small $2$};
  \draw (0.055,-0.077) node[shift = {(.4,.25)}] {\small $x$};
\end{tikzpicture}
 \qquad & \qquad \cI = \{ 123, 231, 341, 451, 561, 612, 712, 812\}
\end{align*}

This gives a diagrammatic derivation of the missing boundary directly from the original diagram $W$, and the new move this hints at remains local, but the resulting diagram is outside the current scope of the generalized Wilson loop diagram formalism as outlined in this work (see Remark \ref{rmk:matroids on props def}.

\end{eg}

\begin{eg}
\label{eg:missingcodim1boundary2}
Similarly to the previous example, applying the boundary moves to the ordinary Wilson loop diagram

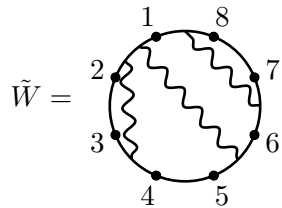
\begin{align*}
\tilde{W}= \begin{tikzpicture}[rotate=-67.5, line width=1, scale=1,baseline=0]
 \def \n {8}
 \draw circle(1) 
 \foreach \v in{1,...,\n} 
 {(360*\v/\n-360/\n+180:1)circle(.4pt)circle(.8pt)circle(1.2pt)circle(1.4pt) node[anchor=360/\n*\v-360/\n-67.5]{$\v$}};
 \draw[decorate,decoration={snake,amplitude=0.8mm}] (-0.966,-0.259) -- (0.924,0.383);
 \draw[decorate,decoration={snake,amplitude=0.8mm}] (-0.866,-0.500) -- (0.383,-0.924);
 \draw[decorate,decoration={snake,amplitude=0.8mm}] (0.383,0.924) -- (-0.924,0.383);
\end{tikzpicture}
\end{align*}
results in generalized Wilson loop diagrams corresponding to all but one of the codimension 1 boundaries of $\mathcal{M}(\tilde{W})$.

In this case, we do not have an equivalent representation -- no equivalent admissible ordinary Wilson loop diagram exists, and moving to generalized Wilson loop diagrams does not result in any useful equivalences for the present purposes.

On the other hand, we can represent the missing boundary using the propagators with the $x$ marking two ends which cannot be used together in a Grassmann necklace element, similarly to Example \ref{eg:missingcodim1boundary1}.  This is illustrated below. 

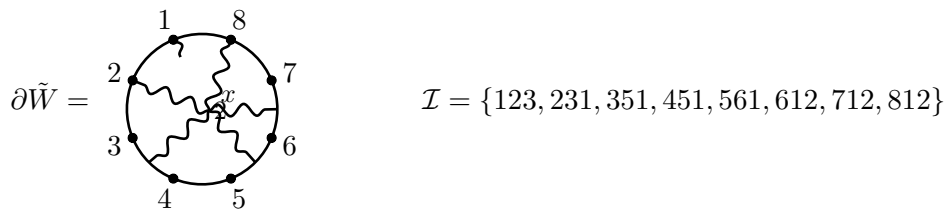
\begin{align*}
\partial \tilde{W} =
\begin{tikzpicture}[rotate=-67.5, line width=1, scale=1,baseline=0]
 \def \n {8}
 \draw circle(1) 
 \foreach \v in{1,...,\n} 
 {(360*\v/\n-360/\n+180:1)circle(.4pt)circle(.8pt)circle(1.2pt)circle(1.4pt) node[anchor=360/\n*\v-360/\n-67.5]{$\v$}};
 \draw[decorate,decoration={snake,amplitude=0.8mm}] (-1.00000000000000,0) -- ++({atan(0)}:.25 cm);
 \foreach \x/\y in {-0.707/-0.707,0.383/-0.924,0.924/0.383,0.383/0.924,-0.707/0.707}{\draw[decorate,decoration={snake,amplitude=0.8mm}] (0.055,0.077) -- (\x,\y);}
 \draw (0.055,0.077) node[shift = {(.15,0)}] {\small $2$};
  \draw (0.055,-0.077) node[shift = {(.4,.25)}] {\small $x$};
\end{tikzpicture}
 \qquad & \qquad \cI = \{ 123, 231, 351, 451, 561, 612, 712, 812\}
\end{align*}

There is a generalized Wilson loop diagram (without the additionally generalized propagator as in $\partial \tilde{W}$) with the same matroid as $\partial \tilde{W}$, as illustrated below.  However, it is not apparently obtainable through a local move from $\tilde{W}$ and so this remains a subject for future investigation.

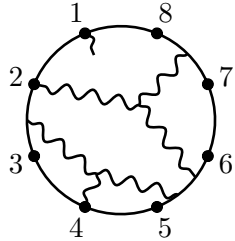
\begin{align*}
\begin{tikzpicture}[rotate=-67.5, line width=1, scale=1.25,baseline=0]
 \def \n {8}
 \draw circle(1) 
 \foreach \v in{1,...,\n} 
 {(360*\v/\n-360/\n+180:1)circle(.4pt)circle(.8pt)circle(1.2pt)circle(1.4pt) node[anchor=360/\n*\v-360/\n-67.5]{$\v$}};
 \draw[decorate,decoration={snake,amplitude=0.8mm}] (-1.00000000000000,0) -- ++({atan(0)}:.25 cm);
 \foreach \x/\y in {-0.707/-0.707,0.866/0.500,-0.383/0.924}{\draw[decorate,decoration={snake,amplitude=0.8mm}] (-0.075,0.239) -- (\x,\y);}
 \foreach \x/\y in {-0.383/-0.924,0.707/-0.707,0.966/0.259}{\draw[decorate,decoration={snake,amplitude=0.8mm}] (0.430,-0.457) -- (\x,\y);}
\end{tikzpicture}
\end{align*}
\end{eg}

Since the $W$ of Example \ref{eg:missingcodim1boundary1} and the $\tilde{W}$ of Example \ref{eg:missingcodim1boundary2} are up to rotation and reflection the only diagrams with missing boundaries and we have identified diagrams for those boundaries in two different ways, then ultimately we have obtained all boundaries of codimension one for all of $\text{WLD}(3,8)$.

\subsubsection{Further remarks on implementation and applications}

The Python implementation of the graphical calculus of Section \ref{sec:moves} was designed specifically for the goal of proving Theorem \ref{thm:boundarymovesgenerate}. As such there are some aspects of the general theory that are not reflected in its implementation and therefore some limitations that are worth pointing out for the prospective end user. 

Generally speaking the implementation can only be applied to generalized Wilson loop diagrams in which (1) no propagators have repeated ends and (2) half propagators do not share support with other propagators. Diagrams with these features do not arise from iterated application of the moves in Section \ref{sec:moves} to admissible ordinary Wilson loop diagrams and thus their omission from the implementation does not affect our main result. Similarly we can largely ignore diagrams in which two propagators end at the same vertex, since whenever such diagrams are produced by applying some sequence of moves to ordinary Wilson loop diagrams, there is some equivalent representation without this property attainable by some other sequence of moves.

Apart from the restriction to a simpler class of generalized Wilson loop diagrams, there are a couple of other differences between the definitions of the moves and their implementation worth pointing out. For example, for our purposes it is not necessary to implement the slide move when a propagator ends on either side of a given vertex. We are also able to get away without implementing one of the cases of the merge move from (Section \ref{sec:merge}), namely the case when two propagators ending at the same vertex are merged. 

With these considerations in mind the implementation of the graphical calculus in the accompanying files should be viewed mainly as a graphical calculus for computing boundaries of ordinary Wilson loop diagrams with a small number of propagators. 

\section{Conclusion\label{sec:finale}}

In this paper we define a notion of generalized Wilson loop diagrams which generalizes ordinary Wilson loop diagrams (Section~\ref{sec:gWLDdiagrammatics}) and provides a diagrammatic way to understand many boundaries of ordinary Wilson loop diagrams, including at higher codimension (Section~\ref{sec:diagrammaticmoves}). The generalized Wilson loop diagrams are associated to matroids using the formalism of Rado matroids (Section~\ref{sec:allWLDtomatroid}). We prove properties of these objects, notably including that non-crossing generalized Wilson loop diagrams always give positroids (Corollary~\ref{res:non-crosspos}), and an algorithm to read off the Grassmann necklace of said positroid (Algorithm~\ref{algo:non-crossgwld}). We define a set of diagrammatic boundary moves on non-crossing generalized Wilson loop diagrams (Section \ref{sec:moves}), prove that they give rise to boundaries of the associated positroid cells (Section~\ref{sec:movesrealizeboundaries}), and implement them in the code included in the supplemental files to the arXiv version of this paper. We show that our moves give all boundaries at all codimensions for all admissible ordinary Wilson loop diagrams with two propagators, give all codimension 1 boundaries for all admissible ordinary Wilson loop diagrams with $3$ propagators on $6$, $7$, and $8$ vertices, although, when $n= 8$, one is sometimes required to use an equivalent representation (Section~\ref{sec:boundaries of more than 2 props}). 

Understanding boundaries of positroid cells is important both in mathematics and in physics. Mathematically, while there are many different (equivalent) combinatorial representations of positroids, to date, only one of them lets us directly find boundaries while a different one is needed to calculate the codimension of said boundary. While it is known that the positroid cells have a CW structure in the positive Grassmannian \cite{Postnikov, GrBall}, needing to pass back and forth between different representations makes tracing geometric information across the many bijections quite opaque. From the physics perspective, many boundaries correspond to poles in the associated Feynman integral. While it is understood that $\mathcal{N}=4$ SYM is a finite theory, and thus the singularities should only occur on in the boundaries in such a way as to cancel, without a proper characterization of the boundaries, proving this property is far from straightforward. For instance \cite{cancellation} shows that all singularities that appear on codimension one poles cancel, but identifies singularities that appear on higher codimension boundaries, about which nothing can be said. A better understanding of how boundaries relate between different diagrams lets us understand how the poles cancel in the overall amplitude. Furthermore, there are boundaries that do not contain poles, and so one also wants to understand how to distinguish between boundaries containing and not containing poles. 

Furthermore, generalized Wilson loop diagrams are a tool to understand the boundaries that remain tethered to the original physics context. Just as the ordinary Wilson loop diagrams encode the 4-momenta of the propagators via their support in the diagram, generalized Wilson loop diagrams index the colinearities between momenta. Therefore, unlike other combinatorial or geometric approaches to characterizing positroid boundaries, these diagrams are a combinatorial tool to understand boundaries that maintains many aspects of the original physical intuition. 

Many interesting questions remain. First of all one would like to have diagrammatic moves which generate all boundaries at all codimensions at least of admissible ordinary Wilson loop diagrams, if not of all non-crossing generalized Wilson loop diagrams. As we saw in Remark \ref{rmk:matroids on props M(W)} and Examples \ref{eg:missingcodim1boundary1} and \ref{eg:missingcodim1boundary2}, it looks like further generalizing our generalized Wilson loop diagrams in order to allow the ends of a propagator to be decorated by a matroid (not just a uniform matroid) will be necessary, at least if the moves are to be local. The other option is less clear. In Example \ref{eg:missingcodim1boundary1} we see that in that case one can obtain the missing boundary by taking an equivalent diagram to the starting diagram, and then using one of the moves which we already had, while in Example~\ref{eg:missingcodim1boundary2}, we have a representation of the missing boundary as a generalized Wilson loop diagram with a different structure throughout, not just a local change.  In this situation there is a trade-off between local moves, acting only on ends supported in a limited region, and moves not needing matroids beyond uniform on propagators. 
We are interested in formalizing the moves which go beyond the uniform case. The first step in this direction is clear: when merging multiple-ended propagators with capacity less than one fewer than the number of ends, the matroid we put on the new merged propagator must be the one where the independences are as before the merge except with at most one copy of the now merged ends allowed.
On the other hand from the perspective of the positroid, there is no particular advantage to insisting on local moves.  We also hope that at the expense of non-local changes, whether equivalences or not, we can capture all boundaries with the generalized Wilson loop diagrams of this paper, using additional moves that have not yet been specified.

The value of moving between equivalent representations also brings up another family of interesting questions, that of understanding all equivalent representations of generalized Wilson loop diagrams. In the ordinary Wilson loop diagram case, equivalent representations come about precisely from retriangulating regions of the diagram that have a critical density condition on propagators known as \emph{exactness} \cite{generalcombinatoricsI}, and so they are enumerable and understandable. A similar characterization in the generalized case would be very desirable. One could also formulate the question in terms of giving a set of diagrammatic equivalence moves, similar to the diagrammatic boundary moves, but which give equivalent diagrams rather than boundary diagrams. Some of these equivalences are given in Section \ref{sec:manyreps}.

Equivalent representations come about because of a kind of density condition, which we do not have a rigorous handle on, but which we see in examples. In the ordinary Wilson loop diagram case, admissibility is also (in part) a density condition, see Definition \ref{dfn:admissibleWLD}. So in the ordinary case we have one density condition (admissibility) to establish the class of well behaved diagrams, and another (exactness) to characterize equivalence. In the generalized case we do not understand precisely the required density conditions, and so we work with various conditions: local minimality, capacity-rankedness, generalized admissibility, as needed. All of these are in some sense density conditions and sorting these out into simple and useful density conditions with structural definitions would be very appealing.

Returning to the boundary moves themselves, it would be nice to understand which moves give rise to boundaries containing poles in the Feynman context. For codimension 1, one can infer from the definitions in this paper and in \cite{cancellation} that the merge move and slide move are sufficient to find the boundaries with poles. The eventual goal would be to use these moves to understand pole cancellation at higher codimensions. Algebraically, we intend to look at the structure of the algebra generated by composing these moves in the spirit of cluster algebras. Finally, because of our initial physics motivation we focus here on those positroids which come from admissible ordinary Wilson loop diagrams and their boundaries, but one could also ask if some notion of generalized Wilson loop diagram can provide another insightful diagrammatic form for all positroids.

\bibliographystyle{abbrv}
\bibliography{Bibliography}

\end{document}